\documentclass[iop]{emulateapj}
\usepackage{amsmath}
\usepackage{longtable}
\usepackage{threeparttablex} 
\newcommand{\copyrightnotice}{\enlargethispage{24pt}
\let\thefootnote\relax\footnote{\copyright\ 2015 Christopher Heil}}

\newcommand{\EQ}{\; = \;}

\newcommand{\svee}{{\hbox{\raise.4ex \hbox{${\scriptscriptstyle{\vee}}$}}}}
\newcommand{\swedge}{{\hbox{\raise.4ex \hbox{${\scriptscriptstyle{\wedge}}$}}}}

\newcommand{\Ac}{{\mathcal{A}}}

\newcommand{\Bc}{{\mathcal{B}}}

\newcommand{\Cc}{{\mathcal{C}}}

\newcommand{\CHI}{\hbox{\raise.4ex \hbox{$\chi$}}}

\newcommand{\scap}{\hbox{\raise.25ex
\hbox{${\operatornamewithlimits{\scriptstyle\bigcap}}$}}}
\newcommand{\scup}{\hbox{\raise.25ex
\hbox{${\operatornamewithlimits{\scriptstyle\bigcup}}$}}}

\newcommand{\deltacheck}
  {\overset{\lower.4ex \hbox{${\scriptscriptstyle{\hskip 2 pt\vee}}$}} \delta}

\newcommand{\Ec}{{\mathcal{E}}}

\newcommand{\Fcheck}
    {\overset{\lower.4ex \hbox{${\scriptscriptstyle{\hskip 2 pt\vee}}$}} F}

\newcommand{\fwedgehat}
    {\overset{\lower.6ex \hbox{${\scriptscriptstyle{\hskip 3 pt\wedge}}$}} f}
\newcommand{\fveecheck}
    {\overset{\lower.4ex \hbox{${\scriptscriptstyle{\hskip 2 pt\vee}}$}} f}
\newcommand{\fkcheck}
   {\overset{\lower.4ex \hbox{${\scriptscriptstyle{\hskip 1 pt\vee}}$}} {f_k}}

\newcommand{\raiseprime}{\hbox{\raise.3ex \hbox{${\scriptstyle{\prime}}$}}}

\newcommand{\Gc}{{\mathcal{G}}}
\newcommand{\Gcheck}
    {\overset{\lower.4ex \hbox{${\scriptscriptstyle{\hskip 2 pt\vee}}$}} G}

\newcommand{\gveecheck}
    {\overset{\lower.4ex \hbox{${\scriptscriptstyle{\hskip 2 pt\vee}}$}} g}

\newcommand{\hcheck}
    {\overset{\lower.4ex \hbox{${\scriptscriptstyle{\hskip 2 pt\vee}}$}} h}

\newcommand{\Kcheck}
  {\overset{\lower.4ex \hbox{${\scriptscriptstyle{\hskip 2 pt\vee}}$}} K}

\newcommand{\Lc}{{\mathcal{L}}}

\newcommand{\mucheck}{\overset{\lower.4ex \hbox{${\scriptscriptstyle{\hskip 2 pt\vee}}$}} \mu}

\newcommand{\norm}[1]{\left\|#1\right\|}

\newcommand{\Oc}{{\mathcal{O}}}

\newcommand{\varphicheck}
    {\overset{\lower.4ex \hbox{${\scriptscriptstyle{\hskip 1 pt\vee}}$}}
              \varphi}

\newcommand{\R}{\mathbb{R}}
\newcommand{\Rc}{{\mathcal{R}}}

\newcommand{\Rbar}
  {{\overset{\hskip -0.9 pt \lower\ 1.5pt \hbox{{\rule{6.7pt}{0.45pt}}}} \R}}
\newcommand{\subRbar}
   {{\overset{\hskip -0.8 pt \lower\ 1.5pt \hbox{{\rule{4.5pt}{0.5pt}}}} \R}}

\newcommand{\Sc}{{\mathcal{S}}}

\newcommand{\SNa}{S_N^{\hskip 0.5 pt
  \hbox{\raise.3ex \hbox{\small\textup{a}}}}}
\newcommand{\SNaa}[1]{S_#1^{\hskip 0.5 pt
  \hbox{\raise.3ex \hbox{\small\textup{a}}}}}
\newcommand{\SNo}{S_N^{\hskip 0.5 pt
  \hbox{\raise.3ex \hbox{\small\textup{o}}}}}
\newcommand{\SNt}{S_N^{\hskip 0.5 pt
  \hbox{\raise.3ex \hbox{\small\textup{t}}}}}

\newcommand{\Tc}{{\mathcal{T}}}

\newcommand{\zeroveecheck}
    {\overset{\lower.4ex \hbox{${\scriptscriptstyle{\hskip 0.5 pt\vee}}$}} 0}

\usepackage{amsmath,amsthm,amsfonts,amssymb,bm}
\usepackage[dvipsnames]{xcolor}
\usepackage{nicefrac} 
\usepackage{hyperref}
\usepackage{cleveref}
\usepackage{tikz}
\usetikzlibrary{shapes.gates.logic.US,trees,positioning,arrows}
\newtheorem{theorem}{Theorem}

\crefname{appendix}{Appendix}{Appendices}
\crefname{table}{Table}{Tables}
\crefname{equation}{Eq.}{Eqs.}
\crefname{figure}{Fig.}{Figs.}
\crefname{section}{Sec.}{Secs.}

\newcommand{\qij}{\bm{q}_i-\bm{q}_j}
\newcommand{\disqij}{\norm{\bm{q}_i-\bm{q}_j}}
\newcommand{\Mc}{\mathcal{M}}
\newcommand{\Euler}{\text{Euler}}
\newcommand{\Verlet}{\text{Verlet}}
\newcommand{\TriJump}{\text{TriJump}}
\newcommand{\fast}{\text{fast}}
\newcommand{\slow}{\text{slow}}
\newcommand{\multi}{\text{multi}}
\newcommand{\Kc}{\mathcal{K}}
\newcommand{\GRIT}{\texttt{GRIT}}

\newcommand{\smercuryt}{\texttt{SMERCURY-T}}
\newcommand{\epsK}{\varepsilon_{_K}}

\shorttitle{N-Rigid-Body Integrator}
\shortauthors{Chen et al.}

\begin{document}

\title{\texttt{GRIT}: a package for structure-preserving simulations of gravitationally interacting rigid-bodies}
\author{Renyi Chen \altaffilmark{1}, Gongjie Li \altaffilmark{2}, Molei Tao \altaffilmark{1}}
\affil{$^1$ School of Mathematics, Georgia Institute of Technology, Atlanta, GA 30332, USA}
\affil{$^2$ Center for Relativistic Astrophysics, School of Physics, Georgia Institute of Technology, Atlanta, GA 30332, USA}
\email{gongjie.li@physics.gatech.edu}

\begin{abstract}
Spin-orbit coupling of planetary systems plays an important role in the dynamics and habitability of planets. However, symplectic integrators that can accurately simulate not only how orbit affects spin but also how spin affects orbit have not been constructed for general systems. Thus, we develop symplectic Lie-group integrators to simulate systems consisting gravitationally interacting rigid bodies. A user friendly package (\texttt{GRIT}\footnote{\url{https://github.com/GRIT-RBSim/GRIT}}) is provided and external forcings such as tidal interactions are also included. As a demonstration, this package is applied to Trappist-I. It shows that the differences in transit timing variations due to spin-orbit coupling could reach a few min in ten year measurements, and strong planetary perturbations can push Trappist-I f, g and h out of the synchronized states. 
\bigskip
\end{abstract}

%\keywords{Earth--chaos--instabilities, Exoplanets--dynamics}

% \tao{TODO: add a link to our package}

\section{Introduction}
%\revision{Several fixes that are all over the paper: 1. Several typos of the name of the inertia tensor are fixed ($I\to J$). 2. Some tiny fixes in equations and tiny fixed of spellings / grammar are not highlighted.}
Among thousands of detected exoplanetary systems, a significant fraction of them involve planets with close-in orbits. In particular, the occurrence rate for the compact systems (e.g., multiple planets with periods of less than 10 days) are estimated to be $\sim 20-30\%$ \citep{Muirhead15, Zhu18}. The close separations between the planets allow strong planetary interactions that could lead to rich features in the dynamical evolution of the compact planetary systems.

In particular, spin-axis dynamics becomes very interesting in compact planetary systems. For instance, rotational and tidal distortion of the planets can lead orbital precession due to planet spin-orbit coupling, and this causes variations in transit timing. Recently, \citep{Bolmont20} showed that the transit timing variations due to spin-orbit coupling could be detectable for Trappist-I, which could in turn help one constrain physical properties of the planets. In addition, although tidal effects are strong for planets with close-in orbits, strong interactions between the planets could push these planets (with orbital periods in $\sim 10$ days) out of synchronized states \citep{Vinson19}. Moreover, secular resonance-driven spin-orbit coupling could drive to large obliquity variations and lead to obliquity tides. This sculpts the exoplanetary systems: the obliquity tide could explain the overabundance of planet pairs that reside just wide of the first order mean-motion resonances \citep{Millholland19}.

Integrator involving spin-axis coupling have been developed to study these effects. There are mainly two different approaches: 1) evolving the orbital dynamics separately from the spin-axis evolutions \citep[e.g.,][]{Laskar93, li2014spin, Vinson19}; 2) evolving the spin and orbit evolution simultaneously \citep[e.g,][]{Hut81, eggleton1998equilibrium, Mardling02, Lissauer12, Bolmont15, Blanco-Cuaresma17, Millholland19}. In the first approach, orbital evolution of the systems are first integrated using N-body simulation packages assuming the objects are point-mass particles, and then spin-axis dynamics are computed using the results of the orbital evolution. This approach assumes that the effects of the spin on orbital dynamics are weak. In the second approach, additional force due to spin-orbit coupling is included in the N-body simulation package, which could affect the orbital evolution as well as the spin-axis evolution. 

To carefully study effects of spin-orbit coupling, we develop symplectic algorithm (``Gravitationally interacting Rigid-body InTegrator'', \texttt{GRIT}) starting with the first-principal rigid body dynamics, so that the mutual interactions between spin and orbital dynamics can be accurately accounted for. Symplectic Lie-Poisson integrator for rigid body has already been constructed in the seminal work of \citet{Touma94} for systems with near Keplerian orbits, focusing on a system with 1 rigid body (for systems with more than 1 rigid bodies, the spin dynamics of each rigid body is considered separately under it's own frame).
However, for systems that involve close-encounters, the orbits of object are no longer Keplerian.
The original version of the method in \cite{Touma94} was not high-order (in the time step) either. %\li{which order is it? what expansion is the "order" based on? We need to be more specific} 
Building upon the existing progress, we no longer assume near Keplerian orbits for wider applicability, and our package includes several high-order implementations.
Moreover, we put all the bodies under the same inertia frame such that the spin orbit interactions are all considered altogether in one Hamiltonian framework.
We also note that symplectic integrator for secular spin-orbit dynamics have been developed by \citet{Breiter05}, while our method is based on direct (non-secular) numerical simulations and therefore suitable for resonant situations.

% \chen{They split the Hamiltonian to be $H_{Kepler}+H_{interaction}+H_{Euler}$, it seems that they didn't require the orbit to be near Keplerian? I changed some part of the last paragraph and put it below.} \tao{very good. however, are you sure that their frames are awkward?}
% \chen{Yes, I am pretty sure. Their frame is the (moving) body frame and they didn't track the motion of the frame.} \tao{great!}

% \tao{regarding your question, they use the exact solution of $H_{Kepler}$, so if the orbits are far from Keplerian, this introduces large error; e.g., instantaneous e can be >1, and then integration can even blow up. do you agree?}
% \chen{$H_{Kepler}$ is just part of their splitting, $H_{interaction}$ represent the perturbation part. They evolve $H_{Kepler}$, $H_{interaction}$ and $H_{Euler}$ in the order of half step $H_{Kepler}$, $H_{Euler}$, then full step $H_{interaction}$, then half step $H_{Kepler}$, $H_{Euler}$.} \tao{the order of splitting works only if $H=A+B$ and both $A$ and $B$ give bounded solutions. $H_{kepler}$ can give unbounded solutions in close encounters and therefore splitting based on $H_{kepler}$ won't work.}
% \chen{I see. I updated the paragraph above with discussions of frames.}

The development of our integrator is tightly based on the profound field of geometric integration. This is because rigid body dynamics can be intrinsically characterized by mechanical systems on Lie groups. More precisely, the phase space is $T^* \mathsf{SE}(3)^{\bigotimes n}$, where $n$ is the number of interacting bodies and the special Euclidean group $\mathsf{SE}(3)$ is where the center of mass and rotational orientation of each body lives. How to properly simulate such systems in a structure preserving way, so that symplecticity can be conserved and the dynamics remain on the Lie group, has been extensively studied. See e.g., \cite{iserles2000lie, bou2009hamilton, celledoni2014introduction} for general Lie group integrators, and more broadly, \cite{Hairer06,MR2132573,blanes2017concise,calvo1994numerical} for monographs on geometric integration. 

Regarding rigid body integrators in particular, the following is an incomplete list in addition to \cite{Touma94}. Firstly, the work of \cite{dullweber1997symplectic} used a splitting approach (similar to \cite{Touma94} in essence, however split differently) to construct symplectic and Lie-group-preserving integrators for rigid molecules. The main idea is to split the Hamiltonian into a free rigid body part, including both translational and rotational kinetic energies, plus a potential part. The latter can be exactly integrated, and the former too when the rigid body is axial symmetric; otherwise, it is further split into a symmetric top and a correction term, both of which can be exactly integrated in a cheap way (without using special functions). Methods in the proposed package (which are explicit, high-order integrators) are largely based on this idea. Secondly, we note various splitting schemes for integrating free rigid bodies were compared in \cite{fasso2003comparison}. Recall that the free rigid body is integrable, and its numerical simulation based on multiple ways of expressing the exact solution were also proposed (e.g., \cite{van2007symplectic, celledoni2008exact}), but the exact expressions involve special functions (unless the bodies are axial symmetric), which can be computationally expensive. Moreover, the `exact' solutions are not exact due to round-off errors, and this complication is studied (and remedied) in \cite{vilmart2008reducing}. For simple and robust arithmetic, the free-rigid-body part of our method will be based on a sub-splitting into an axial-symmetric part and a small correction, as most rotating celestial bodies relevant to this study are (almost) axial-symmetric. Also worth mentioning is, geometric integrators for (non-free but) gravitationally interacting rigid bodies have also been proposed; besides \cite{Touma94}, \cite{Lee07} constructed variational integrators using elegant geometric treatments; however, those integrators are implicit, and computational efficiency is hence not optimal.

As we are interested in gravitationally interacting rigid bodies, \GRIT{} uses tailored splitting schemes. This way, the existence of small parameters and separation of timescales in the system is utilized so that a better trade-off between efficiency and accuracy can be achieved (see Sec.\ref{sec:ourSplitting} for details). Our treatment is of course based on extensive existing studies of splitting methods, and some more general discussions on integrators based on splitting and composition can be found, e.g., in \cite{mclachlan2002splitting,blanes2008splitting,blanes2013new,tao2010nonintrusive}.

% \tao{Renyi: please revise and confirm the 4 paragraphs that i wrote above; this should be not hard. However, please also add references and comments later on in Methods, so that consistency can be maintained.}

%One remark is, since we assume the objects are rigid bodies, our package cannot be applied to objects with fast deformation timescales comparing with the orbital or spin dynamical timescales of the systems. \li{Let us move this to the discussion, it doesn't fit in well here.}

This article is organized as the following: section \textsection \ref{sec:rigid_body_representation} describes the rigid body formulation adopted in our article, and section \textsection  \ref{sec:algorithm} presents our symplectic algorithms. We then show consistency between our simulation using \GRIT{} and secular theories in section \textsection \ref{sec:codecomp}, for the case of a moonless Earth and the case of a hypothetical Earth-Moon system that include tidal interactions. In the end, we apply our package to simulate Trappist-I in section \textsection \ref{sec:Trap}, in order to investigate the effects of spin-orbit coupling in transit-timing variations, as well as in the tidally synchronized states of Trappist-I planets.

% \begin{itemize}
%     \item New designed integrators considering different scales of the system.
%     \item Tidal force integrated.
%     \item Versatile: body can be set to be rigid / point mass.
% \end{itemize}

\section{Rigid Body Representation}\label{sec:rigid_body_representation}

As we are interested in the dynamics of the planet's spin-axis, the planet is modeled as a rigid body to account for its finite size and rotation.
Other than the spatial position and the linear momentum,
the rotational orientation and angular momentum of the rigid body are necessary to represent its state.
The above are $12$-dimensional in total, and besides the spatial position ($3$-dim) and its conjugate linear momentum ($3$-dim),
we still need a set of variables ($6$-dim) to represent the orientation and the rotation of the rigid body.

\subsection{The Body Frame and The Rotation Matrix}\label{subsec:rigid_body_representation_1}

Under a specific fixed reference frame of Euclidean space $\R^3$ with basis $\left( \bm{E}_1, \bm{E}_2, \bm{E}_3 \right)$,
the spatial position and the translational speed of a body can be expressed by vectors in $\R^3$.
On the other hand, the body frame (\cref{fig:body_frame}) attached to the body gives fixed coordinates of each small particle of the body.
As can be seen from \cref{fig:body_frame}, orthogonal bases $\left( \bm e_1, \bm e_2, \bm e_3 \right)$ form a body frame of this rigid body.
As the body moving along the dashed trajectory following arrows as well as self rotating from time $t_0$ to time $t_1$,
coordinates of points $P, Q$ under the body frame stay the same, without subjecting to the motion of the rigid body.
% for non-uniform objects \tao{Renyi: what does non-uniform objects mean?}.

\begin{figure}
\centering
\includegraphics[width=0.8\linewidth]{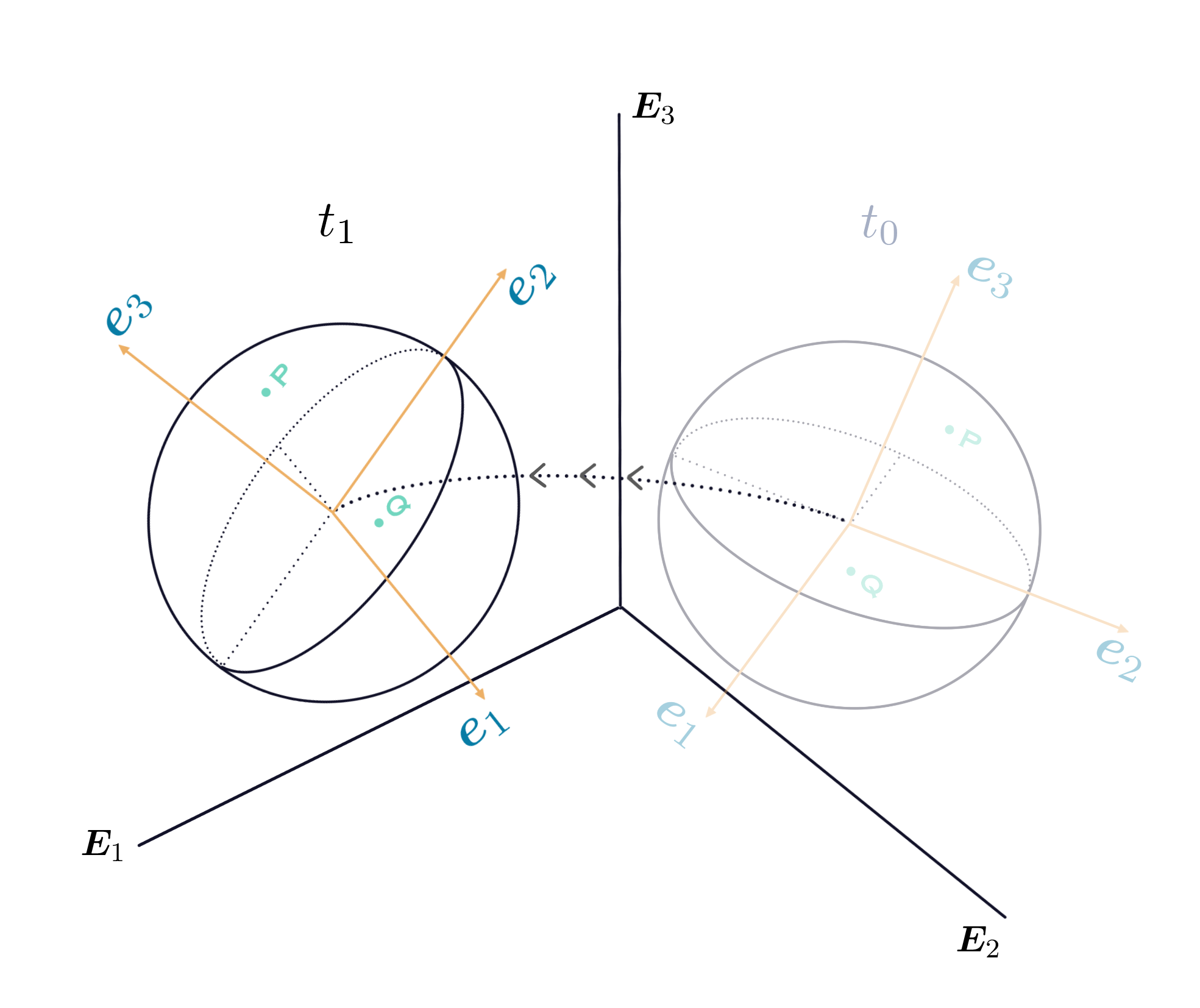}
\caption{The body frame.}
\label{fig:body_frame}
\end{figure}

The configuration of a rigid body is described by both the position of its center of mass and its rotational orientation. The orientation in the reference frame can be expressed as an rotation by an orthogonal matrix $\bm{R}(t) \in \mathsf{SO}(3)$ from the body frame (e.g., z-axis of the body frame at time $t$ will be $\bm{R}(t) \cdot \begin{bmatrix} 0 & 0 & 1 \\ \end{bmatrix}^T$ in the reference frame). To switch between the inertia frame and the body frame, one can simply left multiply the rotation matrix $\bm R$ or $\bm R^{-1}$. Note that $\bm R \in \mathsf{SO}(3)$ and if a numerical method can keep $\bm R$ exactly in this Lie group, its inverse will be equal to its transpose, i.e.\ $\bm R^{-1}=\bm R^T$.

\subsection{The Angular Velocity and the Angular Momentum}
Denoting $\bm{\Omega} = \begin{bmatrix} \Omega_1 & \Omega_2 & \Omega_3 \end{bmatrix}^T \in \R^3$
the angular velocity of the rigid body under the body frame,
then the direction of $\bm\Omega$ matches the rotational axis and $\norm{\bm{\Omega}}_2$ represents the rotational speed.
Consider a mass point $\bm{x} = \begin{bmatrix} x_1 & x_2 & x_3 \end{bmatrix}^T$ in one rigid body under the body frame,
its speed under the body frame can be expressed as
\begin{align}
    \bm{\Omega} \times \bm{x} =
    \begin{bmatrix}
        0 & -\Omega_3 & \Omega_2 \\
        \Omega_3 & 0 & -\Omega_1 \\
        -\Omega_2 & \Omega_1 & 0 \\
    \end{bmatrix}
    \begin{bmatrix}
        x_1 \\
        x_2 \\
        x_3 \\
    \end{bmatrix}
    = \hat{\bm{\Omega}} \bm{x},
\end{align}
where the hat-map $\hat{\cdot}$ is an isomorphism from the Lie algebra $\mathfrak{so}(3)$ to 3-by-3 skew-symmetric matrices, defined by
\begin{align}
\hat{\bm{\Omega}} := 
\begin{bmatrix}
    0 & -\Omega_3 & \Omega_2 \\
    \Omega_3 & 0 & -\Omega_1 \\
    -\Omega_2 & \Omega_1 & 0 \\
\end{bmatrix}.
\end{align}

In addition, the inverse map of $\hat{\cdot}$ is denoted by $\check{\cdot}$.

With the angular velocity, the rotational kinetic energy of this rigid body can be expressed as
\begin{align}
\begin{split}
    T^{rot}\left( \bm{\Omega} \right)
%    & = \int_\Bc \rho(\bm{x}) \frac{1}{2} \norm{\bm{\Omega} \times \bm{x}}_2^2 \, d \bm{x} \\
%    & = \int_\Bc \rho(\bm{x}) \frac{1}{2} \norm{\hat{\bm x} \bm \Omega}_2^2 \, d \bm{x} \\
%    & = \frac{1}{2} \bm{\Omega}^T \left[ \int_\Bc {\rho(\bm{x})} \hat{\bm{x}}^T \hat{\bm{x}} \, d \bm{x} \right] \bm{\Omega}  \\
    & = \frac{1}{2} \bm{\Omega}^T \bm{J} \bm{\Omega} \\
\end{split}
\label{eq:T_rot_J}
\end{align}
% \tao{where did the 1/2 in the 1st line go?} \chen{fixed}
% \tao{also, x is 3-by-1, so isn't $x^T x$ 1-by-1 instead of a matrix?} \chen{$\hat{x}$ is the vector under the hat map, resulting in a 3x3 matrix.}
with $\bm{J}$ the (standard) moment of inertia tensor. 
Specifically, for an ellipsoid with semiaxes $a,b,c$ and mass $M$, choosing the principal axes as the body frame such that $x,y,z$-axes matches semi-axes
and taking the integral, we have the following moment of inertia tensor for a uniform density object.
\begin{align}
\begin{split}
\bm{J} 
    & =\int_\Bc \rho(\bm{x}) \hat{\bm{x}}^T \hat{\bm{x}} \, d^3\bm{x} \\
    & = \begin{bmatrix}
    \frac{1}{5} M (b^2 + c^2) & 0                         & 0                         \\
    0                         & \frac{1}{5} M (a^2 + c^2) & 0                         \\
    0                         & 0                         & \frac{1}{5} M (a^2 + b^2) \\
        \end{bmatrix}. \\
\end{split}
\label{eq:ellipsoid_moment_of_inertia_matrix}
\end{align}
Note that one may substitute this with the principal moment of inertia directly.
 
Alternatively, the rotational kinetic energy can also be expressed as
\begin{align}
    T^{rot} \left( \hat{\bm{\Omega}} \right)
%& = \int_\Bc \rho(\bm{x}) \frac{1}{2} \norm{\hat{\bm{\Omega}} \bm{x}}_2^2 \, d\bm{x} \nonumber\\
%& = \frac{1}{2} \int_\Bc \rho(\bm{x})
%    Tr \left[ \hat{\bm{\Omega}} \bm{x} \bm{x}^T \hat{\bm{\Omega}}^T  \right] \, d\bm{x} \nonumber\\
 & = \frac{1}{2} Tr \left[ \hat{\bm{\Omega}} \bm{J}_d \hat{\bm{\Omega}}^T \right].
\label{eq:T_rot_Jd}
\end{align}
% \tao{sometimes you use $\cdot$ and sometimes you use nothing. they meant the same thing, right? we can just get rid of all $\cdot$'s.}\chen{fixed.}

with $\bm{J}^{(d)} = \int_\mathcal{B} \rho(\bm{x}) \bm{x} \bm{x}^T \, d\bm{x}$ (nonstandard) moment of inertia.
We also have $\bm{J}^{(d)} = \frac{1}{2} \, Tr \left[ \bm{J} \right] \bm{I}_{3 \times 3} - \bm{J}$
($\bm{J} = Tr \left[ \bm{J}^{(d)} \right] \bm{I}_{3 \times 3} - \bm{J}^{(d)}$).

By definition, the angular momentum in the body frame is $\bm \Pi=\bm J \bm \Omega$.
Left multiplying the rotation matrix $\bm{R}$, the angular velocity and the angular momentum in the inertia frame are $\bm{\omega}=\bm{R} \bm{\Omega}$ and $\bm \pi = \bm R \bm \Pi$ respectively.

\subsection{The Relation between the Rotation Matrix and the Angular Velocity}
Express $\bm{R}(t)=\begin{bmatrix} \bm{c}_1(t) | \bm{c}_2(t) | \bm{c}_3(t) \end{bmatrix}$ where 
$\bm{c}_1(t),\, \bm{c}_2(t),\, \bm{c}_3(t)$ are columns of $\bm{R}(t)$.
We have $\bm{c}_1(t),\, \bm{c}_2(t),\, \bm{c}_3(t)$
representing directions of three axes of the body in the reference frame respectively.
By the definition of angular velocity, we have $\dot{\bm{c}}_i(t) = \hat{\bm{\omega}} c_i(t)$ for $i=1,2,3$, thus
\begin{equation} \label{eq:rot_ang_vel}
    \dot{\bm{R}}(t) = \hat{\bm{\omega}} \bm{R}(t).
\end{equation}
Multiplying both sides of \cref{eq:rot_ang_vel} with ${\bm{R}(t)}^T$,
we have $\dot{\bm{R}}(t) {\bm{R}(t)}^T = \hat{\bm{\omega}}$ which is a skew-symmetric matrix.  
Considering the speed of an arbitrary mass point $\bm{x}$,
$\bm v_{\bm{x}} = \bm{R} \hat{\bm{\Omega}} \bm{x} = \hat{\bm{\omega}} \left[ \bm{R} \bm{x} \right]$.
Thus $\bm{R} \hat{\bm{\Omega}} = \hat{\bm{\omega}} \bm{R}$, which
implies $\hat{\bm{\Omega}} = \bm{R}^T \hat{\bm{\omega}} \bm{R} = \bm{R}^T \dot{\bm{R}}$.

% With $\hat{\bm{\Omega}} = \bm{R}^T \dot{\bm{R}}$,
% we can actually rewrite the rotational kinetic energy as a function of $\dot{\bm{R}}$ from equation \ref{eq:T_rot_Jd}.

% \begin{align}
%     T^{rot} (\dot{\bm{R}})
%     = \frac{1}{2} Tr\left[ \bm{R}^T \dot{\bm{R}} \bm{J_d} \dot{\bm{R}}^T \bm{R} \right]
%     = \frac{1}{2} Tr\left[ \dot{\bm{R}} \bm{J_d} \dot{\bm{R}}^T  \right].
%     \label{eq:T_rot_Rdot}
% \end{align}
% \tao{was the last term right? i thought you want $\hat{\Omega}$. but don't you already have eq.\ref{eq:T_rot_Jd}? i'm confused about what we're trying to do here.}

To summarize, for a rigid body, it's angular velocity and angular momentum in different frames are denoted as the following,

\begin{table}[h!]
\centering
 \begin{tabular}{|c c c|} 
 \hline
                    & The inertia frame    & The body frame \\ 
                     & (fixed) & (moving) \\
 \hline
 Angular & $\bm{\omega} \, (=\! \bm{R} \bm{\Omega})$ & $\bm{\Omega}$            \\ 
   velocity & & \\
   \hline
    Angular & $\bm{\pi} \, (=\! \bm{R} \bm{\Pi})$       & $\bm{\Pi}$               \\
      momentum & & \\
 \hline
 \end{tabular}
\end{table}

% \begin{table}
% \begin{tabu} to 0.8\linewidth{X[c]X[c]X[c]}
% \toprule
%                      & The inertia frame (fixed)    & The body frame (moving) \\ \midrule
%   Angular  velocity & $\bm{\omega} \, (=\! \bm{R} \bm{\Omega})$ & $\bm{\Omega}$            \\ \midrule
%   Angular  momentum & $\bm{\pi} \, (=\! \bm{R} \bm{\Pi})$       & $\bm{\Pi}$               \\
% \bottomrule
% \end{tabu}
% \end{table}

\noindent
with $\bm{\Pi} = \bm{J} \bm{\Omega}$ and $\bm{\pi} = \bm{J} \bm{\omega}$.
Specifically, we have $\hat{\bm{\Omega}} = \bm{R}^T \dot{\bm{R}}$ and $\hat{\bm{\omega}} = \dot{\bm{R}} \bm{R}^T$.

The rotation matrix $\bm{R}$ and the angular momentum $\bm{\Pi}$ will be utilized to describe a rigid body
when we design an N-rigid-body integrator later (details can be found in Sec.\ref{sec:algorithm}).

\section{Rigid Body Simulation: Algorithms}\label{sec:algorithm}

In this section, we will design symplectic integrators of the N-rigid-body system using splitting methods.
The splitting method is basically to view the Hamiltonian (\cref{eq:algorithm_Hamiltoion}) as the sum of  several integrable parts,
and then to compose the flow of each part over some pre-designed time duration to achieve a certain order of local error.
In the following, we will introduce the Hamiltonian, build the symplectic integrators and analyze the accuracy of integrators step by step.
% \li{I don't find the analysis on the accuracy of the integrators below...} \chen{updated}
In addition, we will provide a way to incorporate non-conservative forces into the integrators, such as the tidal force and post Newtonian effects.

\subsection{The Constrained Hamiltonian of an N-rigid-body System}

Denote $m_i$ the mass of the $i$-th body;
$\bm{q}_i \in \R^3$ the position of the $i$-th body;
$\bm{p}_i \in \R^3$ the linear momentum of the $i$-th body;
$\bm{R}_i \in \mathsf{SO}(3)$ the rotation matrix of the $i$-th body;
$\bm{\Pi}_i \in \R^3$ the angular momentum of the $i$-th body;
$\bm{J}_i \in \R^{3\times 3}$ the (standard) moment of inertia tensor for the $i$-th body.

% \tao{Renyi: use consistent notations. In the previous section you used $\bm J$ but now you use $\bm I$. please fix throughout the paper.}

The Hamiltonian of this system consists of the linear kinetic energy $T^{linear} = \sum_i \frac{1}{2} \bm{p}_i^T \bm{p}_i / m_i$,
the rotational kinetic energy $T^{rot}=\sum_i \frac{1}{2} \bm{\Pi}_i^T \bm J_i^{-1} \bm{\Pi}_i$ and
the potential energy 
\begin{align}
V(\mathsf{q}, \mathsf{R}) = \sum_{i < j} V_{ij} \left( \bm{q}_i, \bm{q}_j, \bm{R}_i, \bm{R}_j \right).
\label{eq:potential_energy}
\end{align}
Denote $\mathsf{q} = \left\{ \bm{q}_1, \bm{q}_2, \ldots, \bm{q}_N \right\}$,
 $\mathsf{p} = \left\{ \bm{p}_1, \bm{p}_2, \ldots, \bm{p}_N \right\}$,
 $\mathsf{\Pi} = \left\{ \bm{\Pi}_1, \bm{\Pi}_2, \ldots, \bm{\Pi}_N \right\}$,
 $\mathsf{R} = \left\{ \bm{R}_1, \bm{R}_2, \ldots, \bm{R}_N \right\}$.
The Hamiltonian can be expressed as 
\begin{align}
\begin{split}
    & H(\mathsf{q},\mathsf{p},\mathsf{\Pi},\mathsf{R})
    = \sum_i \frac{1}{2} \bm{p}_i^T \bm{p}_i / m_i \\
    & \quad + \sum_i \frac{1}{2} \bm{\Pi}_i^T \bm I_i^{-1} \bm{\Pi}_i
    +  V(\mathsf q, \mathsf R)
    \qquad \text{with } \bm{R}_i \in \mathsf{SO}(3).
\end{split}
\label{eq:algorithm_Hamiltoion}
\end{align}
The true potential energy between $i$-th body and $j$-th body is
\begin{align}
    \int_{\mathcal{B}_i} \int_{\mathcal{B}_j}
    - \frac{\mathcal{G} \rho(\bm{x}_i)\rho(\bm{x}_j)}
    {\| \left( \bm{q}_i + \bm{R}_i \bm{x}_i \right) - \left( \bm{q}_j + \bm{R}_j \bm{x}_j \right) \|}
    \, d\bm{x}_j \, d\bm{x}_i.
    \label{eq:true_potential}
\end{align}

%\revision{typos in the potential expressions}
We may approximate it as $V_{ij}$ (in \cref{eq:potential_energy}) by Taylor expanding the denominator.
Expanding to the 2nd order with respect to the radius of the planet over the distance between two bodies (see appendix \ref{appendix:potential}), the approximated potential is,
% \li{mention the order of Taylor series expansion} \chen{added}
% \chen{verify this}
\begin{align}
    \begin{split}
        & \int_{\Bc_i} \int_{\Bc_j}
                - \frac{\Gc \rho(\bm x_i)\rho(\bm x_j)}
                {\| \left( \bm{q}_i + \bm{R}_i \bm{x}_i \right) - \left( \bm{q}_j + \bm{R}_j \bm{x}_j \right) \|}
                \, d\bm{x}_j \, d\bm{x}_i \\
        & \approx -\frac{\Gc m_i m_j}{\disqij} - \frac{\Gc\left( m_i Tr[\bm J_i] + m_j Tr[\bm J_j] \right)}{2 \disqij^3} \\
        & + \frac{3 \Gc {\left( \qij \right)}^T \left( m_j \bm R_i \bm J_i \bm R_i^T + m_i\bm R_j \bm J_j \bm R_j^T \right)\left( \qij \right)}{2 \disqij^5} \\
    \end{split}
    \label{eq:approximated_potential}
\end{align}
with $-\frac{\Gc m_i m_j}{\disqij}$ being the potential of purely point mass interactions and
the rest part being the corrections of the potential due to the body $i$ and $j$ being not point masses.
If we further expand the potential to the 4th order (see appendix \ref{appendix:potential}), rigid body -- rigid body interactions will also be included as higher order corrections. For example, fourth order potential has recently been considered for binary asteroids with large non-spherical terms, and leads to interesting effects \citep{Hou17}.

\subsection{Equations of Motion}

% \chen{Can we find a Poisson Bracket $\left\{ \cdot, \cdot \right\}$
% for $F, H \in \mathcal{F}\left( {\mathfrak{so}(3)}^* \right)$?
% ((free) Rigid body bracket and Heavy top bracket are already known, how to understand them?).
% What is a general approach? (Lie-Poisson reduction?)
% A free rigid body's Hamiltonian is Left invariant, but in our case, the Hamiltonian with a potential term is not left (right) invariant.}

%We firstly find equations of motion of the rotation matrix and the angular momentum for a system consisting of one rigid body then apply it to the N-rigid-body system.

%Following section\ref{subsec:rigid_body_representation_1}, let's name the conjugate variable of $\bm{R}(t)$ as $\bm{P}(t)$.
The Lagrangian for a system consisting of one rigid body is a function of $\bm{R}(t)$ and $\dot{\bm{R}}(t)$ by plugging in $\hat{\bm \Omega}=\bm R^T \dot{\bm R}$ in \cref{eq:T_rot_Jd},
\begin{align}
\begin{split}
 & L\left( \bm{R}, \dot{\bm{R}} \right) %= T \left( \dot{\bm{R}} \right) - V(\bm{R})  \\
%= & \frac{1}{2} Tr\left[ \bm R^T \dot{\bm{R}} \bm{J}_d \dot{\bm{R}}^T \bm R \right] - V(\bm{R}). \\
 = \frac{1}{2} Tr\left[ \dot{\bm{R}} \bm{J}_d \dot{\bm{R}}^T  \right] - V(\bm{R}).
\end{split}
    \label{eq:Lagrangian}
\end{align} 

Utilizing the constraint $\bm R^T \bm R - \bm I = \bm 0$ (appendix \ref{appendix:approach1}) or using the variational principle of Hamilton's for Lie group  (appendix \ref{appendix:approach2}), one can derive the equations of motion
\begin{align}
\left\{ 
    \begin{aligned}
        \dot{\bm{R}} & = \bm R \widehat{\bm{J}^{-1} \bm{\Pi}}, \\
        \dot{\bm{\Pi}} & = \bm{\Pi} \times \bm{J}^{-1} \bm{\Pi} 
        - {\left( \bm{R}^T \frac{\partial V\left( \bm{R} \right)}{\partial \bm{R}} 
        - {\left(\frac{\partial V\left( \bm{R} \right)}{\partial \bm{R}}\right)}^T \bm{R} \right)}^{\vee}. \\
    \end{aligned}
\right.
\label{eq:eom_one_rb}
\end{align}

% \subsubsection{Equations of Motion of the N-rigid-body system}

Similarly, the equations of motion of the
N-rigid-body system for the Hamiltonian (\cref{eq:algorithm_Hamiltoion}) are,

\begin{align}
\left\{ 
    \begin{aligned}
        \dot{\bm{q}}_i   & = \frac{\bm{p}_i}{m_i},                       \\
        \dot{\bm{p}}_i   & = -\frac{\partial{V}}{\partial{\bm{q}_i}},    \\
        \dot{\bm{R}}_i   & = \bm R_i \widehat{\bm{J}_i^{-1} \bm{\Pi}_i}, \\
        \dot{\bm{\Pi}}_i & = \bm{\Pi}_i \times \bm{J}_i^{-1} \bm{\Pi}_i
                            - {\left( \bm{R}_i^T \frac{\partial V}{\partial \bm{R}_i} 
                            - {\left(\frac{\partial V}{\partial \bm{R}_i}\right)}^T \bm{R}_i \right)}^\vee. \\
    \end{aligned}
\right.
\label{eq:eom_N_rb}
\end{align}

% One can check that the angular momentum is conserved.

\subsection{Splitting Methods for the System with Axis-symmetric Bodies}
In this section, we utilize the splitting method to construct symplectic integrators.
A diverse range of symplectic integrators with different accuracy and time complexities can be designed
as the splitting method is quite flexible in terms of splitting and composition.
Based on our Hamiltonian of the N-rigid-body system, we will explore three different types of integrators.
One split the Hamiltonian into two parts with comparable size, the other two split the Hamiltonian into, respectively, three and four parts corresponding to various magnitudes and hence different timescales.
% \chen{Wisdom Holman splitting is now supported (totally three different splitting).} \tao{please revise the text. totally three different splitting -> three different splittingS IN TOTAL (not to be confused with `totally different')}

In terms of the shape of rigid bodies, we make the axis-symmetric assumption in this section for simplicity. That is, without loss of generality, $\bm J_i = \begin{bmatrix} J_i^{(1)} & 0 & 0 \\ 0 & J_i^{(1)} & 0 \\ 0 & 0 & J_i^{(3)} \end{bmatrix}$.
For general rigid bodies that are not axis-symmetric, different splitting mechanisms can be applied (see Sec.\ref{sec:nonsymmetric}).

\subsubsection{%\revision{title updated} 
Classical Splitting for Rigid-Body: \texorpdfstring{$H=H_1+H_2$}{} with \texorpdfstring{$\frac{H_1}{H_2} = \mathcal{O}(1)$}{}} \label{sec:split1}

% \tao{something wrong with the title?}

One way of splitting is $H=H_1+H_2$ following \cite{dullweber1997symplectic}, with
\begin{align}
    \left\{
        \begin{aligned}
            H_1(\mathsf{q},\mathsf{p},\mathsf{\Pi},\mathsf{R}) & = \sum_i \frac{1}{2} \bm{p}_i^T \bm{p}_i / m_i
                            + \sum_i \frac{1}{2} \bm{\Pi}_i^T \bm J_i^{-1} \bm{\Pi}_i, \\
            H_2(\mathsf{q},\mathsf{p},\mathsf{\Pi},\mathsf{R}) & = V(\mathsf{q}, \mathsf{R})
                            = \sum_{i \neq j} V_{ij} \left( \bm{q}_i, \bm{q}_j, \bm{R}_i, \bm{R}_j \right). \\
        \end{aligned}
    \right.
\label{eq:split1}
\end{align}

% \chen{isolate V to below eqn (9)}

For $H_1$, the equations of motion are
\begin{align}\label{eq:EoM_H1}
    \left\{ 
        \begin{aligned}
            \dot{\bm{q}}_i   & = \frac{\bm{p}_i}{m_i},                       \\
            \dot{\bm{p}}_i   & = 0,    \\
            \dot{\bm{R}}_i   & = \bm R_i \widehat{\bm{J}_i^{-1} \bm{\Pi}_i}, \\
            \dot{\bm{\Pi}}_i & = \bm{\Pi}_i \times \bm{J}_i^{-1} \bm{\Pi}_i. \\
        \end{aligned}
    \right.
\end{align}
In \cref{eq:EoM_H1}, the $4$th equation is the Euler equation for a free rigid body. It is exactly solvable, and the solution expression is particularly simple for axial-symmetric bodies:
\begin{align}
    \bm \Pi_i(t) 
    = \exp\left( {-\theta t \widehat{\begin{bmatrix}
                0 \\ 0 \\ 1 \\
        \end{bmatrix}}} \right) \bm \Pi_i(0)
    := \mathrm{R}_z^T(\theta t) \bm \Pi_i(0)
\end{align}
with
$\theta=\left( \frac{1}{J_i^{(3)}} - \frac{1}{J_i^{(1)}} \right) \bm \Pi_i^T(0) \begin{bmatrix}
    0 \\ 0 \\ 1 \\
\end{bmatrix}$ and $\mathrm{R}_z$ being the rotation matrix.
Take $\bm \Pi_i(t)$ back to the $3$rd equation of \cref{eq:EoM_H1}, we can obtain $\bm R_i(t)$ too.

%\revision{The font of the rotation matrices (around $z$-axs / around angular momentum $\Pi$) was changed. The descriptions of the rotation matrices were added.}
Therefore, the flow $\phi^{[1]}_t$ of $H_1$ is,
\begin{align}\label{eq:phi_H1}
    \left\{ 
        \begin{aligned}
            \bm{q}_i (t)   & = \bm{q}_i(0) + \frac{\bm{p}_i}{m_i} t,                       \\
            \bm{p}_i (t)   & = \bm{p}_i(0),    \\
            \bm{R}_i (t)   & = \bm R_i (0) \, \mathrm{R}_{\bm \Pi(0)}\left( \frac{\norm{\bm \Pi(0)}}{J_i^{(1)}} t \right) \,
                                   \mathrm{R}_z \left( \theta t \right), \\
            \bm{\Pi}_i (t)  & = \mathrm{R}_z^T \left( \theta t \right) \, \bm{\Pi}(0), \\
        \end{aligned}
    \right.
\end{align}
with $\mathrm{R}_z$ and $\mathrm{R}_{\bm \Pi(0)}$ being rotation matrices representing the rotations around the $z$-axis and $\bm \Pi(0)$ respectively.

For $H_2$, the equations of motion are
\begin{align}
    \left\{ 
        \begin{aligned}
            \dot{\bm{q}}_i   & = 0,                       \\
            \dot{\bm{p}}_i   & = -\frac{\partial{V}}{\partial{\bm{q}_i}},    \\
            \dot{\bm{R}}_i   & = 0, \\
            \dot{\bm{\Pi}}_i & = - {\left( \bm{R}_i^T \frac{\partial V}{\partial \bm{R}_i} 
            - {\left(\frac{\partial V}{\partial \bm{R}_i}\right)}^T \bm{R}_i \right)}^\vee. \\
        \end{aligned}
    \right.
\end{align}
As $\bm q_i$ and $\bm R_i$ stay constants, we have $\bm p_i$ and $\bm \Pi_i$ changing at constant rates. Therefore, the flow $\phi^{[2]}_t$ for $H_2$ is given by
\begin{align}
    \left\{ 
        \begin{aligned}
            \bm{q}_i(t)   & = \bm{q}_i(0),                       \\
            \bm{p}_i(t)   & = \bm{p}_i(0) -\frac{\partial{V}}{\partial{\bm{q}_i}} t,    \\
            \bm{R}_i(t)   & = \bm{R}_i(0), \\
            \bm{\Pi}_i(t) & = \bm{\Pi}_i(0) - {\left( \bm{R}_i^T \frac{\partial V}{\partial \bm{R}_i} 
            - {\left(\frac{\partial V}{\partial \bm{R}_i}\right)}^T \bm{R}_i \right)}^\vee t. \\
        \end{aligned}
    \right.
\end{align}

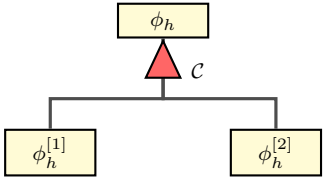
\begin{figure}
\centering
\begin{tikzpicture}[
% Gates and symbols style
    compose1/.style={buffer gate US,thick,draw,fill=blue!60,rotate=90,
		anchor=east,minimum width=0.3cm},
    compose2/.style={buffer gate US,thick,draw,fill=red!60,rotate=90,
		anchor=east,minimum width=0.3cm},
% Label style
    label distance=3mm,
    every label/.style={blue},
% Event style
    event/.style={rectangle,thick,draw,fill=yellow!20,text width=1cm,
		text centered,font=\sffamily,anchor=north},
% Children and edges style
    edge from parent/.style={very thick,draw=black!70},
    edge from parent path={(\tikzparentnode.south) -- ++(0,-0.8cm)
			-| (\tikzchildnode.north)},
    level 1/.style={sibling distance=3cm,level distance=1.2cm,
			growth parent anchor=south,nodes=event},
    level 2/.style={sibling distance=2.25cm},
    level 3/.style={sibling distance=6cm},
    level 4/.style={sibling distance=3cm}
%%  For compatability with PGF CVS add the absolute option:
%   absolute
    ]
%% Draw events and edges
    \node (g1) [event] {$\phi_h$}
	     child{node (g2) {$\phi_h^{[1]}$}}
     	 child {node (g3) {$\phi_h^{[2]}$}};
%% Place gates and other symbols
%% In the CVS version of PGF labels are placed differently than in PGF 2.0
%% To render them correctly replace '-20' with 'right' and add the 'absolute'
%% option to the tikzpicture environment. The absolute option makes the 
%% node labels ignore the rotation of the parent node. 
   \node [compose2]	at (g1.south) [label={[xshift=32, yshift=-2, color=black]$\Cc$}] {};
%   \node [or]	at (g4.south)	[label=-20:G04]	{};
%   \node [or]	at (g5.south)	[label=-20:G05]	{};
%   \node [be]	at (b1.south)	[label=below:B01]	{};
%   \node [be]	at (b2.south)	[label=below:B02]	{};
%   \node [be]	at (b3.south)	[label=below:B03]	{};
%   \node [tr]	at (t1.south)	[label=below:T01]	{};
%   \node [tr]	at (t2.south)	[label=below:T02]	{};
\end{tikzpicture}
\caption{Composition of $\phi_h^{[1]}$ and $\phi_h^{[2]}$. The root node represent the final scheme. The leaves represent the basic ingredients of the composition which are exact flows. The red arrow represent the {composition method} specialized in composing two child flows with comparable scales.}\label{fig:compose1}
\end{figure}
We may compose $\phi_t^{[1]}$ and $\phi_t^{[2]}$ via $\Cc$ to construct different symplectic integrators~\citep[see e.g., ][]{mclachlan2002splitting} (see \cref{fig:compose1}).
To name a few,  set $\Cc$ 
% \tao{what does $\Cc$ mean? was this introduced before?} \chen{added introduction in the previous sentence}
as $\Cc_{\Euler}$, $\phi_h = \phi_h^{[1]} \circ \phi_h^{[2]}$ is a 1st order scheme with $h$ being the step size (see appendix \ref{appendix:composition} for $\Cc_\Euler$ and the following composition methods $\Cc_\Verlet$ and $\Cc_{S6}$).

Applying symmetric composition $\Cc_\Verlet$, a $2$nd order integrator $\Tc_2$ is in the form of
\begin{align}
    \phi_h^{\Verlet}:=\Cc_{\Verlet} \left(\phi_{h}^{[1]}, \phi_{h}^{[2]} \right) = \phi_{\frac{h}{2}}^{[1]} \circ \phi_h^{[2]} \circ \phi_{\frac{h}{2}}^{[1]}
    \label{eq:phi_Verlet}
\end{align}

Applying $\phi_h^{\TriJump}$ \citep{suzuki1990fractal},
we have the following 4th-order scheme $\Tc_4$,
\begin{align}
    \phi_{\gamma_1 h}^\text{Verlet} \circ \phi_{\gamma_2 h}^\text{Verlet} \circ \phi_{\gamma_3 h}^\text{Verlet},
    \label{eq:4th_order_suzuki}
\end{align}
with $\gamma_1 = \gamma_3 = \frac{1}{2-2^{1/3}}$, $\gamma_2=1-2\gamma_1$.
Similarly, a 6th-order scheme $\Tc_6$ can be constructed by composing $\phi_h^{[1]}$, $\phi_h^{[2]}$ with $\Cc_{S6}$
% \tao{description of $\Cc_{S6}$? and/or give a reference}. \chen{reference added in the previous paragraph.}
% Here, $\phi_h^S$ has an $\mathcal{O}\left( h^5 \right)$ local truncation error, meaning this is a $4$th order method. 

In the package,  $\Tc_2$, $\Tc_4$ and $\Tc_6$ are implemented.

% \tao{Renyi: is this section just a review, or some of those methods are actually implemented in our package? if the latter, we should state it; if the former, we should briefly describe why they are not used; connect to the discussion in the next section.}

\subsubsection{%\revision{title updated}
Tailored Splitting I: \texorpdfstring{$H=H_1+H_2+H_3+H_4$}{} with \texorpdfstring{$\frac{H_1}{H_2} = \mathcal{O}(1)$}{}
and \texorpdfstring{$\frac{H_3}{H_1},\frac{H_4}{H_2} = \mathcal{O}(\varepsilon)$}{}}
\label{sec:ourSplitting}

Different from point mass systems which can already exhibit dynamics over multiple timescales, the N-rigid-body system can have additional timescales created by the rotational dynamics. %\li{point mass systems also have different timescales, right? The faster orbital timescale and the longer secular evolution timescale on orbital variations.}

%\begin{itemize}
%    \item Evolution of the linear kinetic energy and the rotational kinetic energy have different \li{time?} scales
%    \item The potential energy between point masses and the correction of the potential energy (due to rigidity) \li{lead to dynamics with} different \li{time} scales. \li{please confirm my changes.}
%\end{itemize}
Thus, we further split the Hamiltonian into more terms of different magnitudes, which produce flows at different timescales, and then carefully compose them\footnote{%\chen{Blanes focus on the splittings of $H=A+\epsilon B$ which only contains two different part. The splitting is not new (can at least sourced back to McLachlan 1995), what they do in 2013 paper is finding parameters for higher order methods. But I feel the footnote looks okay here.}
Similar techniques have already been employed; see e.g., \cite{blanes2013new} and references therein. The structure of our system, however, is new (due to the rigid-body part) and thus so is our specific splitting.}. More specifically, consider $H=H_1+H_2+H_3+H_4$ with
\begin{align}
    \left\{
        \begin{aligned}
            H_1(\mathsf{q},\mathsf{p}) & = \sum_i \frac{1}{2} \bm{p}_i^T \bm{p}_i / m_i \\
            H_2(\mathsf{q},\mathsf{p}) & = -\sum_{i<j} \frac{\mathcal{G} m_i m_j}{\norm{\bm q_i - \bm q_j}}, \\
            H_3(\mathsf{\Pi}) & = \sum_i \frac{1}{2} \bm{\Pi}_i^T \bm J_i^{-1} \bm{\Pi}_i, \\
            H_4(\mathsf{q},\mathsf{R}) & = V(\mathsf{q}, \mathsf{R}) - H_2. \\
        \end{aligned}
    \right.
\label{eq:split2}
\end{align}
Here, $H_1$, $H_2$ have comparable size and $\frac{H_3}{H_1},\frac{H_4}{H_2}=\Oc \left( \varepsilon \right)$
with $\varepsilon$ being a small scaling parameter determined by the properties of the system.
Based on scales of the dynamics, we denote $H^{\fast}=H_1+H_2$ and $H^{\slow}=H_3+H_4$.
For example, consider the solar system, setting all the bodies to be point masses except the Earth, 
%\revision{The H values are deleted here as $H$ is data dependent and not necessarily put here. So only $\varepsilon$'s value is kept (which is the only information we need).}
% $H_1=0.00482304$, $H_2=-0.00925409$, $H_3=5.78393 \cdot 10^{-9}$, $H_4=-5.12936 \cdot 10^{-16}$ with
$\varepsilon\approx10^{-6}$.

The flows $\left\{ \varphi^{[1]}_t,\varphi^{[2]}_t,\varphi^{[3]}_t,\varphi^{[4]}_t\right\}$
of $\left\{H_1,H_2,H_3,H_4\right\}$ can be derived similarly to Sec.\ref{sec:split1} and 
the schemes are build by hierarchically composing $\left\{ \varphi_t^{[i]} \right\}_{i=1}^4$ together.
Specifically, as shown in \cref{fig:hierarchical}, we firstly group the flows of the fast dynamics ( $\varphi_t^{[1]}$ and $\varphi_t^{[2]}$) as a sub-scheme $\varphi_h^{\fast}$ via $\Cc_{\fast}$ and the flows of the slow dynamics ($\varphi_t^{[3]}$ and $\varphi_t^{[4]}$) as a sub-scheme $\varphi_h^{\text{slow}}$ via $\Cc_{\slow}$ respectively. Then composing $\varphi_h^{\text{fast}}$ and $\varphi_h^{\text{slow}}$ together as the final scheme $\varphi_h^{\multi}$ via $\Cc_{\multi}$.
$\Cc_{\fast}$ and $\Cc_{\slow}$ are composition methods of composing two Hamiltonian flows with comparable scales \citep{mclachlan2002splitting}.
$\Cc_{\multi}$ is a composition method specialized in perturbative Hamiltonian systems of the form $H=A+\varepsilon B$ \citep{mclachlan1995composition, laskar2001high, blanes2013new}.
Note that the flows $\varphi_h^{\fast}$, $\varphi_h^{\slow}$ are not exact, so the order of $\varphi_h^{\multi}$ is not the same as the order of $\Cc_{\multi}$ applied for exact flows.
In fact, the global error of $\varphi_h^{\multi}$ is the summation of the global errors of all three methods $\Cc_{\fast}$, $\Cc_{\slow}$ and $\Cc_{\multi}$ (see appendix \ref{appendix:proof_error} for proof).

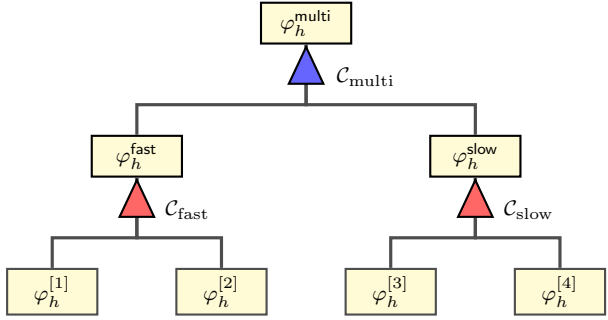
\begin{figure}
\centering
\begin{tikzpicture}[
% Gates and symbols style
    compose1/.style={buffer gate US,thick,draw,fill=blue!60,rotate=90,
		anchor=east,minimum width=0.3cm},
    compose2/.style={buffer gate US,thick,draw,fill=red!60,rotate=90,
		anchor=east,minimum width=0.3cm},
% Label style
    label distance=3mm,
    every label/.style={blue},
% Event style
    event/.style={rectangle,thick,draw,fill=yellow!20,text width=1cm,
		text centered,font=\sffamily,anchor=north},
% Children and edges style
    edge from parent/.style={very thick,draw=black!70},
    edge from parent path={(\tikzparentnode.south) -- ++(0,-0.8cm)
			-| (\tikzchildnode.north)},
    level 1/.style={sibling distance=4.5cm,level distance=1.2cm,
			growth parent anchor=south,nodes=event},
    level 2/.style={sibling distance=2.25cm},
    level 3/.style={sibling distance=6cm},
    level 4/.style={sibling distance=3cm}
%%  For compatability with PGF CVS add the absolute option:
%   absolute
    ]
%% Draw events and edges
    \node (g1) [event] {$\varphi_h^{\multi}$}
	     child{node (g2) {$\varphi_h^{\fast}$}   
	     	   child {node (g4) {$\varphi_h^{[1]}$}}
	     	   child {node (g5) {$\varphi_h^{[2]}$}}
	     	   }
     	 child {node (g3) {$\varphi_h^{\slow}$}
	     	   child {node (g6) {$\varphi_h^{[3]}$}}
	     	   child {node (g7) {$\varphi_h^{[4]}$}}
		};
%% Place gates and other symbols
%% In the CVS version of PGF labels are placed differently than in PGF 2.0
%% To render them correctly replace '-20' with 'right' and add the 'absolute'
%% option to the tikzpicture environment. The absolute option makes the 
%% node labels ignore the rotation of the parent node. 
   \node [compose1]	at (g1.south) [label={[xshift=50, yshift=-2, color=black]$\Cc_{\multi}$}] {};
   \node [compose2]	at (g2.south) [label={[xshift=44, yshift=-2, color=black]$\Cc_{\fast}$}] {};
   \node [compose2]	at (g3.south) [label={[xshift=46, yshift=-2, color=black]$\Cc_{\slow}$}] {};
%   \node [or]	at (g4.south)	[label=-20:G04]	{};
%   \node [or]	at (g5.south)	[label=-20:G05]	{};
%   \node [be]	at (b1.south)	[label=below:B01]	{};
%   \node [be]	at (b2.south)	[label=below:B02]	{};
%   \node [be]	at (b3.south)	[label=below:B03]	{};
%   \node [tr]	at (t1.south)	[label=below:T01]	{};
%   \node [tr]	at (t2.south)	[label=below:T02]	{};
\end{tikzpicture}
\caption{Hierarchical composition tree. The root node represents the final scheme. The leaves represent the basic ingredients of the composition which are exact flows. Nodes in the middle represent the intermediate composition flows. Red arrows represent the {composition methods} specialized in composing two child flows with similar scales. The blue arrow represent the {composition methods} specialized in composing two child flows with different scales.}\label{fig:hierarchical}
\end{figure}

For example, if we set $\Cc_{\fast}$, $\Cc_{\slow}$ and $\Cc_{\multi}$ as $\Cc_{S6}$, $\Cc_\Verlet$ and $\Cc_{ABA42}$
(see appendix \ref{appendix:composition}) respectively.
The global error of the above method is $\Oc(h^6) + \Oc(\varepsilon^2 h^2 ) + \Oc( \varepsilon h^4 + \varepsilon^2 h^2 )$, i.e.\ $\Oc\left( h^6 + \varepsilon h^4 + \varepsilon^2 h^2 \right)$. We name it $\Mc_{642}$ scheme with $6,4,2$ representing the power of $h$ of each term in the order and $\Mc$ representing multiscale splitting.

Similarly, we design the $\Mc_{42}$ scheme by choosing $\Cc_{\fast}$, $\Cc_{\slow}$ and $\Cc_{\multi}$ as $\Cc_\TriJump$, $\Cc_\Verlet$ and $\Cc_{ABA22}$ respectively with the global error being $\Oc(h^4 + \varepsilon h^2)$.
% \begin{align*}
%     \varphi^{\slow}_h & =\varphi_h^{[3]} \circ \varphi_h^{[4]}, \\
%     \varphi^{\multi}_h & =
%     \varphi_{h/2}^{\slow}
%     \circ \varphi_{h}^{\fast}
%     \circ {\varphi_{h/2}^{\slow}}.
% \end{align*}
% with ${\varphi_{h/2}^{\slow}}^*$ the adjoint of $\varphi_{h/2}^{\slow}$.

\begin{table}
\begin{center}
\begin{tabular}{| c | c | c |}
\hline
 scheme & expensive stages & order \\ 
 \hline
$\Tc_2$ & 3 & $(2)$ \\
$\Tc_4$ & 7 & $(4)$ \\
$\Tc_6$ & 15 & $(6)$ \\
$\Mc_{42}$ & 3 & $(4,2)$ \\
$\Mc_{642}$ & 6 & $(6,4,2)$ \\
\hline
\end{tabular}
\end{center}
%\revision{previous expensive stages of $\Tc_6$ was counted incorrectly and we have updated our $\Cc_{S6}$ method in the appendix as the Yoshida 1990's composition solution A, so the number of stages are recounted as the new one (15).}
\caption{Comparisons of different schemes with respect to the number of dominating expensive stages and the global error order. The number of expensive stages are counted in an isolated step without considering the concatenation of the last stage with the first stage of the next step. The notation in the `order' column is explained in the main text.}
\label{tb:cmp_schemes}
\end{table}

Compared with schemes in Sec.\ref{sec:split1}, tailored splitting is able to mixing the fast and slow flows flexibly, thus being able control the time complexity.
In fact, $T(\phi_h^{[1]})=T(\varphi_h^{[1]})+T(\varphi_h^{[3]})$, $T(\phi_h^{[2]})=T(\varphi_h^{[2]})+T(\varphi_h^{[4]})$ with $T(\cdot)$ being the number of operations of the one-step forward flow and
evolving $\varphi_h^{[3]}$, $\varphi_h^{[4]}$ are much more expensive than evolving $\varphi_h^{[1]}$, $\varphi_h^{[2]}$. Since $\varphi_h^{[3]}$ and $\varphi_h^{[4]}$ are (expensive) slow dynamics that can be evolved with less effort (e.g.\ larger step size, less stages) than fast dynamics when evolving together and tailored splitting makes it possible to control the number of expensive stages.
% and schemes in Sec.\ref{sec:split1} and Sec.\ref{sec:ourSplitting} have different styles of mixtures of cheap stages and expensive stages.
% For those splitting schemes in Sec.\ref{sec:split1} that split the Hamiltonian into $2$ parts, evolving all sub-flows are of similar complexity of evolving $\varphi_h^{[3]}$, $\varphi_h^{[4]}$, while the schemes that split the Hamiltonian into $4$ parts is more flexible in deciding the number of expensive stages.
% the scheme and the schemes in this section only include 4 stages (3 stages if the last stage and the first stage of two consecutive steps are concatenated) that involve $\varphi_h^{[3]}$, $\varphi_h^{[4]}$.
To compare, the number of expensive stages and the global errors of all schemes mentioned (in Sec.\ref{sec:split1} and Sec.\ref{sec:ourSplitting}) are listed in \cref{tb:cmp_schemes}.
In \cref{tb:cmp_schemes}, the order index $(o_0, o_1, \ldots)$ represents the power of $h$ in front of $\varepsilon^0, \varepsilon^1, \ldots$ (e.g.\ a scheme of order $(o_0, o_1, o_2)$ has a global error of $\Oc(h^{o_1}+\varepsilon h^{o_2} + \varepsilon^2 h^{o_3})$).

Moreover, since the hierarchical composition is a general framework, one can easily extend the family of numerical schemes, such as to construct higher order schemes, by applying a variety of existing splitting and composition methods.

\subsubsection{Tailored Splitting II: \texorpdfstring{$H=K_1+K_2+K_3$}{} with \texorpdfstring{$\frac{K_3}{K_1},\frac{K_2}{K_1} = \mathcal{O}(\epsK)$}{}}

We also provide an option to use the popular Wisdom-Holman \citep{wisdom1991symplectic} scheme for the orbital part, which works well for the specific but common setup of near Keplerian orbits; such systems usually correspond to $N-1$ well-separated bodies orbiting around a massive central body (indexed by $1$ in our  following description). This method is similar to the approach by \citet{Touma94}, except that their coordinates are set using the body-frame and we provided a higher-order implementation.

Isolating the Keplerian dynamics as $K_1$, combining the rotational kinetic energy with the rest translational kinetic energy as $K_2$, and putting the rest potential energy to $K_3$, $H=K_1+K_2+K_3$ with
\begin{align}
\left\{
\begin{aligned}
K_1(\mathsf q, \mathsf p) & = H_{Kepler}(\mathsf Q, \mathsf P) = \sum_{i=2}^N \frac{1}{2} \bm P_i^T \bm P_i / m_i - \frac{\Gc m_1 m_i}{\norm{\bm Q_i}}, \\
K_2(\mathsf p, \mathsf \Pi) & = \sum_{i=1}^N \frac{1}{2} \bm{\Pi}_i^T \bm J_i^{-1} \bm{\Pi}_i + \frac{\norm{\bm p_1 - \frac{m_1}{m_{tot}} \sum_{i=1}^N \bm p_i}^2}{2 m_1}, \\
K_3(\mathsf{q},\mathsf{R}) & = V(\mathsf{q}, \mathsf{R}) + \sum_{i=2}^N \frac{\Gc m_1 m_i}{\norm{\bm q_i - \bm q_1}}, \\
\end{aligned}
\right.
\end{align}
% \tao{the above doesn't look right. there should be $q_i-q_j$ terms. where are they placed?}
and $K_2, K_3 \ll K_1$.
Here, $V(\mathsf q, \mathsf R)$ is defined in \cref{eq:potential_energy}.
Note that $K_1$ represents Keplerian orbits in $\mathsf Q, \mathsf P$ variables, which are canonical democratic heliocentric variables \citep{duncan1998multiple} with
\begin{align}
\bm Q_i = 
\left\{
\begin{aligned}
    & \bm q_i-\bm q_1 & \quad i \neq 1, \\
    & \frac{\sum_{j=1}^N m_j \bm q_j}{m_{tot}} & \quad i = 1,
\end{aligned}
\right.
\end{align}
and
\begin{align}
    \bm P_i = 
    \left\{
        \begin{aligned}
            & \bm p_i - \frac{m_i}{m_{tot}} \sum_{j=1}^N \bm p_i & \quad i \neq 1, \\
            & \sum_{j=1}^N \bm p_i & \quad i=1. \\
        \end{aligned}
    \right.
\end{align}
So when evolving $K_1$ dynamics, additional steps of switching back and force between $(\mathsf q, \mathsf p)$ and $(\mathsf Q, \mathsf P)$ coordinates are necessary.
In terms of compositions, similarly, we first compose the flows of $K_2$ and $K_3$ together as $\varphi_h^{K,slow}$ via $\Cc_{\slow}^K$, then compose the flow of $K_1$ ($\varphi_h^{K,\fast}$) with $\varphi_h^{K, \slow}$ via a multiscale compositing method $\Cc_{\multi}^K$.
The error of such composition is the summation of the global errors of two methods $\Cc_{\slow}^K$, $\Cc_{\multi}^K$ (and the numerical error of evolving Keplerian orbits).

For instance, $\Kc_{\cdot 2}$ method in our package is based on choosing $\Cc_{\slow}^K$, $\Cc_{\multi}^K$ as $\Cc_{\Verlet}$ and $\Cc_{ABA22}$, and its global error is $\Oc(\epsK h^2)$.
% \tao{Renyi: please add a description of $\mathcal{K}_{\cdot 2}$ etc.. also should Table 1 be modified too?}
% \chen{Table 1 is used for illustrating the M-series schemes are faster than T-series schemes, adding K-series won't fit there as 1. it is very hard to quantify the "expensive" stages. 2. the error depends on different small parameters (K-series:
% $\varepsilon_K$; M-series: $\varepsilon$)}
% \tao{can this be resolved?}
% \chen{Yes, I am moving some comparison discussions to the next section, e.g. compare $\varepsilon$ and $\epsK$.}

\subsubsection{Which One to Use, the $\Tc$-series, the $\Mc$-series, or the $\Kc$-series Methods?}

In general, the orders of the $\Tc$-series methods are only $h$ dependent, while the $\Mc$-series and $\Kc$-series methods are $(h,\varepsilon)$ dependent and $(h,\epsK)$ dependent respectively. Here, $\varepsilon$ and $\epsK$ are system specific, and they affect the choice of method. 
For example, $\varepsilon\approx 10^{-6}$ and $\epsK \approx 10^{-3}$ in Solar system simulations with Earth being the only rigid body --
note that $\epsK$ represents the scale of the orbital planetary interactions while the $\varepsilon$ in Sec.\ref{sec:ourSplitting} represents the scale of the spin and the potential correction due to rigidity, so in practice, $\varepsilon \ll \epsK$. With the small parameters incorporated, the tailored splitting methods are usually more efficient.
In general, the $\Kc$-series methods specialize in near-Keplerian problems, while the $\Mc$-series methods are more generic and at the same time almost always faster than the $\Tc$-series methods with nearly no trade-offs of the accuracy; in fact, oftentimes the $\Mc$-series methods are both more accurate and more efficient due to delicate splittings and compositions\footnote{One should not be misled to think an error like $\mathcal{O}(h^4+\varepsilon h^2)$ is larger than $\mathcal{O}(h^4)$; for example, if $\varepsilon = h^2$, the former may actually be smaller due to different constant factors; see Sec.\ref{sec:NumericsEfficiency} for practical illustrations.}. % in practice given the same leading order (see Sec.\ref{sec:NumericsEfficiency} for details), even additional $\epsilon$ related terms are presented in the order expression (the leading constants of each term under big $\Oc$ notations are actually scheme-dependent, so an $\Oc(h^6+\varepsilon h^4+\varepsilon^2 h^2)$ method is not necessarily less accurate than an $\Oc(h^6)$ method). 
However, $\Tc$-series methods are recommended for extreme cases with large $\varepsilon$ and $\epsK$ (e.g., a super fast spinning body might contribute to a large $\varepsilon$).

% and it is genuinely small such that with a reasonable step size $h$, $\varepsilon = \Oc(h^2)$. And in practice, one can switch to $\Mc_{42}$ instead of $\Tc_4$ for better time efficiency (with similar accuracies) unless under the circumstance of simulating extreme cases (e.g.\ super fast spinning body might contribute to a large $\varepsilon$).

%\tao{is this the only integrator we provide in package? i thought we said we will include several, including some 6th order method.}
%\chen{I completely rewrote our algorithms, and add more integrators under the same framework.}

% \chen{The splitting is not that trivial, after testing the schemes, I will go back and edit the above paragraph and integrate it into our tailored splitting.}
% \tao{Sounds great. Looking forward to hearing more about it.}

\subsection{Adding Non-Conservative Forces}

Non-conservative forces such as tidal forces and post Newtonian corrections are incorporated in the package. As the implemented schemes are based on symmetric splitting and composition
% \chen{new M42 scheme is not symmetric. I tried to incorporate the previous symmetric scheme in our composition techniques, but it seems that the previous version of order 4,2 scheme need to be isolated as a different method. Maybe we should rephrase here.}
the corresponding non-conservative momentum update is inserted in the middle of the composition.
This is similar to how dissipative forces were added in \texttt{REBOUNDx} \citep{Tamayo20}.
% we actually just used forward Euler for these forces

%\tao{can we also explain how to incorporate a general non-conservative forcing?}
%\chen{Reorged the structure of this subsection.}

%\tao{in fact, sec.5 seems to imply we added post-Newtonian forcing etc as well. if that's indeed the case, shall we state it here and also modify the title of Sec.3.4 and the abstract?}\chen{I think so, I also add a subsubsection of general relativistic forces}

% \tao{so the scheme for adding the forces is just 1st-order. this might still be okay, but perhaps it should take place in the middle of a symmetric splitting \& composition scheme. Renyi, could you please confirm?}\chen{it was add at the beginning of each step, moved to the middle now.}

\subsubsection{Tidal Forces}
We model the tidal dissipation between each pair of bodies using the constant time lag equilibrium tide model, following \cite{Hut81, eggleton1998equilibrium}. Note that we only adopted the dissipative component in the tidal force here.  % 
The expression of the acceleration of the tidal force is
\begin{align}
\begin{split}
    & \bm a_{\text{host},\text{guest}}^{tidal} \EQ -\frac{9 \sigma m_{\text{guest}}^2 A^2}{2 \mu_{\text{host},\text{guest}} d^{10}} \bigg[ 3 \bm d \left( \bm d \cdot \dot{\bm d} \right) \\
    & \quad + \left( \left( \bm d \times \dot{\bm d} \right) - \bm \omega d^2 \right) \times \bm d \bigg].
\end{split}
\label{eq:equilibrium_tide}
\end{align}
% \tao{what is $\times d$ at the end? multiplication by a scalar? if yes, remove $\times$.}
Here, $m_{\text{host}},m_{\text{guest}}$ denote the masses of the host and the guest body respectively;
$\bm d \EQ \bm q_{\text{guest}} - \bm q_{\text{host}}$ denotes the relative position of the guest body;
$d \EQ \norm{\bm d}$ denotes the distance between two bodies;
$\mu_{\text{host},\text{guest}} \EQ \frac{m_{\text{host}} \cdot m_{\text{guest}}}{m_{\text{host}}+m_{\text{guest}}}$ denotes the reduced mass;
$\bm \omega$ denotes the angular velocity of the host body under the reference frame (the inertia frame);
the constant $\sigma$ denotes the dissipation rate;
$A$ is defined as 
% \tao{why is it a constant if it depends on $d$?} \chen{fixed}
\begin{align}
    A \EQ \frac{d^5 Q}{1-Q},
\end{align}
with $Q$ the constant that measures quadrupolar deformability of the objects.%(TODO:discuss how to choose Q, maybe a table). \li{Let us leave it as is, and we can include discussions on physical parameters in the next section.} % by treating the host body as an extended body and the guest body as a point mass

The dissipation rate $\sigma$ is related to the time lag $\tau$ by the following formula,
\begin{align}
    \tau \EQ \frac{3 \sigma d^5}{4 \Gc} \cdot \frac{Q}{1-Q}.
\end{align}
%(TODO:add a table for time lags of different kind of bodies?)\li{Let us leave it as is, and we can include discussions on physical parameters in the next section.}

We may integrate the tidal acceleration $\bm a_{i,j}^{tidal}$ to our integrator after each time step by considering all pairs of bodies under tidal interactions. Note that each $\bm a_{i,j}^{tidal}$ only calculates the force of each (host, guest) pair, where each pair $(i,j)$ treat $i$ as the extended object and $j$ as the point mass object. Thus, the equations of motion due to tidal dissipation are listed below:
\begin{align}
\left\{
    \begin{aligned}
        \bm p_i & \EQ \bm p_i + h \sum_{j \neq i} \left( - \mu_{i,j} \, \bm a_{i,j}^{tidal} + \mu_{j,i} \, \bm a_{j,i}^{tidal} \right), \\
        \bm \Pi_i & \EQ \bm \Pi_i - h \sum_{j \neq i} \mu_{i,j} \, \bm R_i^T \left( \left( \bm q_j - \bm q_i \right) \times \bm a_{i,j}^{tidal} \right). \\
    \end{aligned}
\right.
\end{align}
% \tao{i didn't get the $\bm p_i$ update. is it really correct?}
% \chen{I have typos before (an extra $h$ was typed), so for $p_i$, the $-$ term corresponds to the tidal force of the $(i, j)$ pair with $i$ being the host and $j$ being the guest while the $+$ term corresponds to the tidal force of the $(j, i)$ pair with $j$ being the guest while $i$ being the host.}
% \tao{i'm confused. why isn't $\bm \Pi_j$ updated as well?}\chen{previously it was the update of one tidal pair. updated}

\subsubsection{General Relativistic Effects}
\label{sec:GR}
% \chen{Prof. Li. How can I find expressions / references for the general relativistic forces added in our package?} \li{Added}

We added the first-order post-Newtonian correction for general relativistic effects following e.g., \citet[][]{Blanchet06}. For planetary systems, we assumed the central object (the host star) is much more massive comparing to the surrounding objects (the planets). Thus, we only included the correction due to the star. The acceleration can be expressed as the following \citep[e.g.,][]{Anderson75, Benitez08}:
\begin{align}
    {\bf a} = \frac{G M_{star}}{r^3 c^2}\Big[\Big(\frac{4GM_{star}}{r}-{\bf v}^2\Big){\bf r} + 4({\bf v}\cdot{\bf r}){\bf v}\Big]
\end{align}

\subsection{Asymmetric Case}\label{sec:nonsymmetric}

For planets with close-in orbits, both rotational flattening and tidal force distort the shape of the planets, and lead to non-axial symmetric distortions. Thus, we include the option to study non-axial symmetric planets here, where one could specify the principal moment of inertia or the semi-axes of the planets directly. 
% \li{I added the paragraph above. Renyi, please add some descriptions on the following expressions. It is not clear how the solution is included in the integrator:} \tao{I can't agree more. I revised, but Renyi still needs to work more here.}
In this case, $J_i^{(1)} \neq J_i^{(2)} \neq J_i^{(3)}$ in $\bm J_i$, and our splitting of the Hamiltonian is modified as
the previous Hamiltonian plus $H_{asymmetric}$, where
\begin{align}
    H_{asymmetric} (\mathsf{R}, \mathsf{\Pi}) = \sum_{i} \left( \frac{1}{J_i^{(2)}} - \frac{1}{J_i^{(1)}} \right) \cdot \frac{{\Pi_i^{(2)}}^2}{2},
\label{eq:H5}
\end{align}
and $H_{asymmetric} \ll H_3$ in \cref{eq:split2}.

% The equations of motion of \cref{eq:H5} is
% \begin{align}
% \left\{
% \begin{aligned}
% \dot{\bm R}_i & = \bm R_i \cdot \widehat{\begin{bmatrix} 0 \\ \delta \Pi_i^{(2)}  \\ 0 \\ \end{bmatrix}} \\
% \dot{\bm \Pi}_i & = \delta \Pi_i^{(2)} \cdot \begin{bmatrix} -\Pi_i^{(3)} \\ 0 \\ \Pi_i^{(1)} \\ \end{bmatrix} \\
% \end{aligned},
% \right.
% \end{align}
The dynamics of \cref{eq:H5} is
\begin{align}
\left\{
\begin{aligned}
\bm R_i(t) & = \mathrm R_y \left( \delta \Pi_i^{(2)} t \right) \, \bm R_i(0), \\
\bm \Pi_i(t) & = \mathrm R_y (-\delta \Pi_i^{(2)} t) \, \bm \Pi_i(0), \\
\end{aligned}
\right.
\end{align}
with $\delta=\frac{1}{J_i^{(2)}} - \frac{1}{J_i^{(1)}}$.
Based on the symmetric schemes in Sec.\ref{sec:ourSplitting}, we simply evolve $H_{asymmetric}$ half step at the beginning and the end of each step.
% \tao{without affecting the order of the local truncation error?} \chen{it is inaccurate saying without affecting the order of the LTE, to rigorously write the error, we will have more terms that is scaled by another small parameter $\gamma \le \epsilon$, but since $\gamma$ is way too small to contribute numerically in practice. I commented this sentence because we have already specify it in the paragraph above about $H_{asymmetric} \ll H_3$ and too specific clarification might be misleading}.
% without affecting the expressions of the local truncation errors.
% Note that $H_{asymmetric} \ll H_3, H_4$ for near sphere bodies and $H_3, H_4 \ll H_1, H_2$.

\section{Code Validation}
\label{sec:codecomp}

\subsection{Numerical Tests}
%\revision{the whole section is newly added}
\label{sec:numericalTests}
%\tao{use paragraph. conservation properties, convergence tests, what else? currently it is too messy and key statements are buried inside technicalities. say, for example, Fig.X shows both A and B conserve linear and angular momentum (except there are limitations due to machine precision), and their energies are only O(X) fluctations around the conserved value. and then you can described experimental details.}

% \textbf{
% \begin{itemize}
% \item Compare with SMERCURY-T (structure preserving, efficiency as proposed by the reviewer) for
%     \begin{itemize}
%         \item Rotating back and force between two stars.
%         \item Solar system (8 planets) with Earth rigid. Small, Middle, Large $h$ need to be investigated.
%     \end{itemize}
%     \tao{where is the convergence test? where is the alternative capture circumbinary example?}\chen{This is the first version of the plan, I haven't updated it yet.}
% \item Conservation of Energy and Momentum. (Solar system without moon (rigid Earth) example)
% \tao{where are figures referred?}
% \item Compare T4 with M42 in terms of error (how to quantify the error, energy? \tao{no, please use benchmark generated by a high-order method with tiny stepsize, and consider error of the trajectory as a function of time, for example in 2-norm}) and efficiency \tao{efficiency is the key in this demonstration, but is 4th-order potential necessary?} (the efficiency of M42 maximize when 4th order potential is applied).
% \end{itemize}
% }

\subsubsection{Conservation Properties}
The conservation properties of the integrators are tested for $\Tc_4$ and $\Mc_{42}$ schemes in the \textit{Sun-Earth-Moon} system with all three bodies being rigid. 
% \tao{why only 1 rigid body? sounds very fishy.}
%\li{We used the case for an Earth-Moon-Sun system, treating each of the object as a rigid body (?). The initial conditions are set following... Tides are not included in this simulation, in order to test the accuracy of our rigid-body interaction integrator that we constructed here. (comment: Renyi, you need to include enough details so that the readers could reproduce our results)} \chen{Reproducibility will be guaranteed by providing examples in our package. Shall we refer the ICs to the package. Other experimental details are added at the end of this paragraph.} 
As shown in \cref{fig:energy_momentum}, both schemes conserve linear momentum and angular momentum (except there are arithmetic inaccuracies due to machine precision), and the energies exhibit no drift but only fluctuate at magnitudes $\Oc(h^4)$ and $\Oc(h^4+\varepsilon h^2)$ for $\Tc_4$ and $\Mc_{42}$ respectively.
In the simulations, tides are not included (otherwise the system is no longer conservative) and initial conditions are set to be the data of epoch J2000 from JPL HORIZONS System. %\chen{do we need reference or link to JPL HORIZONS system} 

Here floating-point format is set to be double-precision, although our package can also use long-double or single.

Our integrators also (exactly) preserve symplecticity when tidal dissipation is excluded, because they are Hamiltonian splitting schemes. The definition of symplecticity in a non-Euclidean setup is not completely trivial, but the symplecticity of splitting approaches considered here has been established in, e.g., \cite{tao2020variational} (with $r(t)=0$; otherwise one gets a more general result, namely conformal symplecticity).

%\li{It looks good now.}
% Note that a free (symmetric) rigid dynamics (Euler's equation) is isolated in the splittings, and evolving this part incorporates round off errors of trigonometric functions (embedded in rotations of \cref{eq:phi_H1}). Such errors are almost deterministic in a fairly long time span as the angular momentum vectors ($\bm \Pi_i$) are slowly evolving. For this reason,
% the numerical errors will accumulate as the number of steps executed increase and this explains why in \cref{fig:energy_momentum}, a smaller error of the angular momentum is observed for a larger step size (less steps).}

\begin{figure}
\centering
\includegraphics[width=\linewidth]{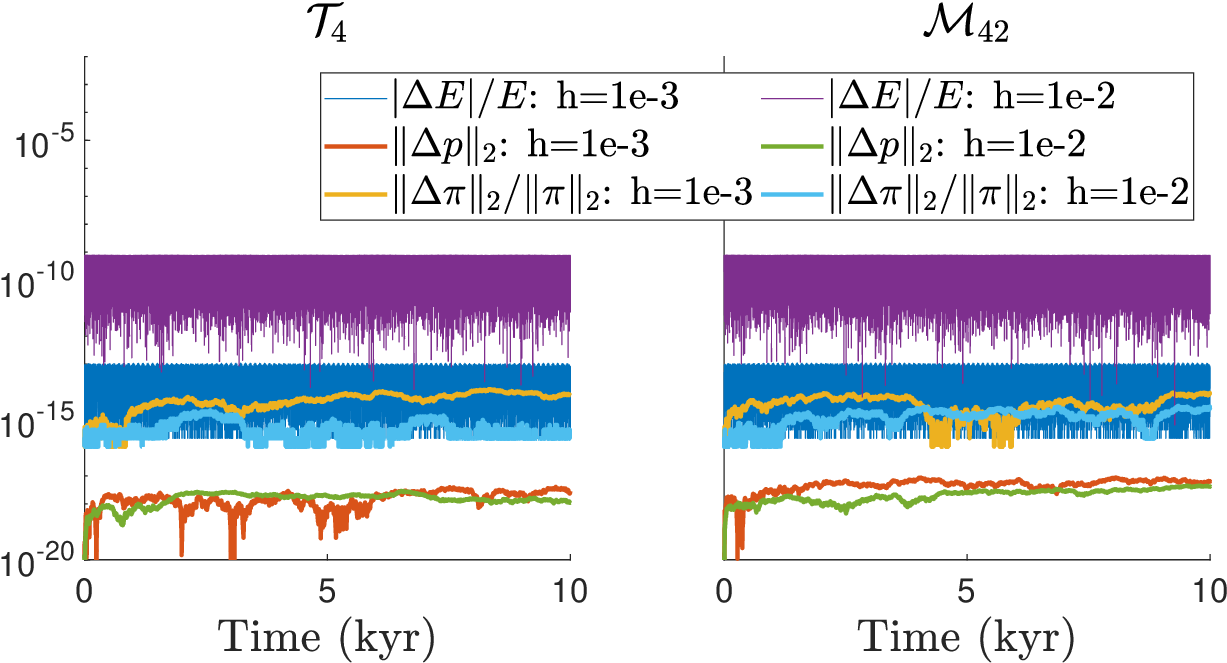}
\caption{Conservation of momentum maps and near conservation of energy by our methods. Relative error of energy $E$, error of the total linear momentum $\bm p$ and relative error of the total angular momentum $\bm \pi$ are measured. $\bm p(0)=[0,0,0]$. The potential order is set to be $2$.}
\label{fig:energy_momentum}
\end{figure}

\subsubsection{Convergence Tests and Accuracy Comparisons}
%\tao{I think this subsubsection needs some more descriptive language. say something like this at least in the beginning: system 1: when stepsize is too large, splitting into $H_1+\cdots+H_4$ which is the default option of \GRIT{} doesn't work, but splitting into Keplerian + \{disturbing function\} for the orbital part doesn't work either; system 2: if some orbit is far from being Keplerian, splitting into $H_1+\cdots+H_4$ instead of Keplerian + \{disturbing function\} is more accurate for all choices of $h$.}
%\chen{I reorged this section with an OVERALL paragraph as well as two description sentences in front of each convergence test paragraph. Let me know if it looks good.}
We now numerically illustrate how the integration error depends on $h$ for different numerical schemes, which include both methods we implemented in \GRIT{} and \smercuryt{}. % Overall, convergences tests are conducted in two systems to compare schemes based on splittings of $H=H_1+H_2+H_3+H_4$ (the default option of \GRIT{}) and schemes based on the splittings of $H_{Kepler}+H_{Interactions}$ for the orbital part (\smercuryt{}).
%\tao{we need to mention what SMERCURY-T is. Gongjie: can you confirm that it is an implementation of Touma-Wisdom, and give a reference?} \li{added, please check if the following looks good to you.} 
\texttt{SMERCURY-T} is a concurrent simulation package that can evolve an object's spin-axis under obliquity tide \citep{Kreyche21}. It is based on the \texttt{Mercury} simulation package \citep{Chambers99}. Specifically, it includes a subroutine to evolve the spin-axis dynamics following the procedure outlined in \citep{Lissauer12}, which is based on the Lie-Poisson integrator of rigid-body dynamics developed by \citet{Touma94}. In addition, it includes a subroutine for obliquity tide following the algorithms outlined in \citet{Bolmont15}. The model for tidal interaction of \texttt{SMERCURY-T} is different from what we included in \GRIT{}, which natually contains both obliquity tide and tidal effects due to non-tidally synchronized orbits. Thus, we focus on the rigid-body dynamics here, where we do not include tidal interactions in our convergence test. We also turned off, in comparisons presented here, our rigid-body rigid-body interaction option, which is mainly for accurate simulations of rigid bodies' close encounters, because such interactions are supported only in \GRIT{}.
% \chen{In addition, as rigid -- rigid interactions is only supported in \GRIT{} (for accurate simulations of rigid bodies' close encounters), such interactions are not considered in numerical comparisons.} \chen{Plan to add my previous comment in the end of this paragraph after the description of \smercuryt{}}
% \tao{excellent. i already incorporated them.}

% \Cref{fig:convergence_SEM} show pleasant convergence scales \tao{what does `convergence scale' mean?} for \GRIT{}'s $\Mc_{42}$ and $\Mc_{642}$ schemes.
% Note that in the \textit{Sun-Earth-Moon} system, the Moon is orbiting around the Earth.
% While the Moon can also be viewed as an object orbiting around the Sun and meanwhile being perturbed by the Earth. So, schemes based on the splitting of $H_{orb}=H_{Kepler}+H_{Interaction}$ (SMERCURY-T) still work here, but requires small time steps to resolve the fast oscillations raised by the near Keplerian perturbation.

We first test on the \textit{Sun-Earth-Moon} system (\cref{fig:convergence_SEM}).
One observation in this case is, if the step size is too large so that splitting into $H_1+H_2+H_3+H_4$ (\GRIT{}'s $\Mc_{42}$, $\Mc_{642}$) doesn't work, \smercuryt{} doesn't work either (unlike expected by some). More precisely,
with $h=2\cdot10^{-2}$ yr, $\Mc_{42}$ and \smercuryt{} cannot resolve the the motion of the Moon orbiting around the Earth, whose period is a month, and even the performance of the 6th order method $\Mc_{642}$ is not ideal, and significant errors are observed in all methods.
Accuracy is improved for stepsizes below this stability limit, and the rate of improvement is, as expected, dependent on the order of the numerical scheme. Consequently, higher order methods such as $\Mc_{42}$ and $\Mc_{642}$ show substantially smaller errors when smaller step sizes are applied (readers interested in understanding this together with computational costs are referred to Sec.\ref{sec:NumericsEfficiency}).

We then test on a non-Keplerian system (note \smercuryt{} performs well for near Keplerian problems as designed): an Earth-like planet orbiting around two stars alternatively in a stellar binary system (\cref{fig:convergence_binary}).
As there is no single body that has the dominant mass of the system and the the planet is alternatively captured by the two stars, the planetary orbit is not nearly Keplerian, and splitting into $H_1+H_2+H_3+H_4$ is more accurate than \smercuryt{} for all choices of step sizes here.
Specifically, as shown in \cref{fig:convergence_binary}, the orbital position of \smercuryt{} saturates to $\mathcal{O}(1)$ relative error after a relatively short period of time, no matter if $h=10^{-3}, 10^{-4}, \text{or } 10^{-5}$ yr. The orbital inaccuracy naturally affects the spin angle as well. Meanwhile, $\Mc_{42}$ and $\Mc_{642}$ do not have this issue.

For reproducibility, the initial condition used is
\begin{align*}
& \bm q_{star_1} = \begin{bmatrix}
    -0.5 & 0 & 0 \\
\end{bmatrix}^T,\,
&& \bm v_{star_1} = \begin{bmatrix}
    0 & -0.0086012119 & 0 \\
\end{bmatrix}^T, \\
& \bm q_{star_2} = \begin{bmatrix}
    0.5 & 0 & 0 \\
\end{bmatrix}^T,\,
&& \bm v_{star_2} = \begin{bmatrix}
    0 & 0.0086012119 & 0 \\
\end{bmatrix}^T, \\
& \bm q_{planet} = \begin{bmatrix}
    1.16 & 0 & 0 \\
\end{bmatrix}^T,\,
&& \bm v_{planet} = \begin{bmatrix}
    0 & 0.0164271047 & 0 \\
\end{bmatrix}^T
\end{align*}
in units of $AU$ and $AU/day$, and $m_{star_1} = m_{star_2} = 0.5 m_{\odot}$, $m_{planet} = m_{\oplus}$.
% \chen{I am not sure if we should specify non-Keplerian here as non-Keplerian orbit can be evolved via universal variables.}\tao{I don't get it. Wisdom-Holman + orbital elements = blow up; Wisdom-Holman + `universal variables' = inaccuracy; 
% ChenLiTao = accuracy. why shouldn't we say we are better than Keplarian based approach (Wisdom-Holman) ?}

% \textbf{For near Keplerian systems, the orbital splitting of $H_{orb}=H_{Kepler} + H_{Interaction}$ (e.g. SMERCURY-T) in general is able to integrate with larger step sizes due to the solvability of the dominant $H_{Kepler}$ and the small scale of $H_{Interaction}$.
% As illustrated in \cref{fig:convergence_trappistI}, using relative large step size $h=10^{-3}$ yr (given the orbital period of the inner most planet is ~1.5 days), $\Mc_{42}$ blows up while SMERCURY-T shows similar accuracy with a more expensive splitting scheme $\Mc_{642}$. 
% Even SMERCURY-T still shows better accuracy than $\Mc_{42}$ for $h=10^{-4}$ yr, the $\Mc_{642}$ is much more accurate than SMERCURY-T due to the higher order. Even though the near Keplerian splitting makes large step sizes available, the dilemma for such splitting is that larger step sizes will bring larger error due to the inaccuracy of $H_{Interaction}$ while integrating with small step sizes can also be inaccurate as an approximated solution of the Kepler equation need to be evaluated each time step and small step sizes basically mean more evaluations of the Kepler equation and such errors will accumulate.}

\begin{figure}
\includegraphics[width=\linewidth]{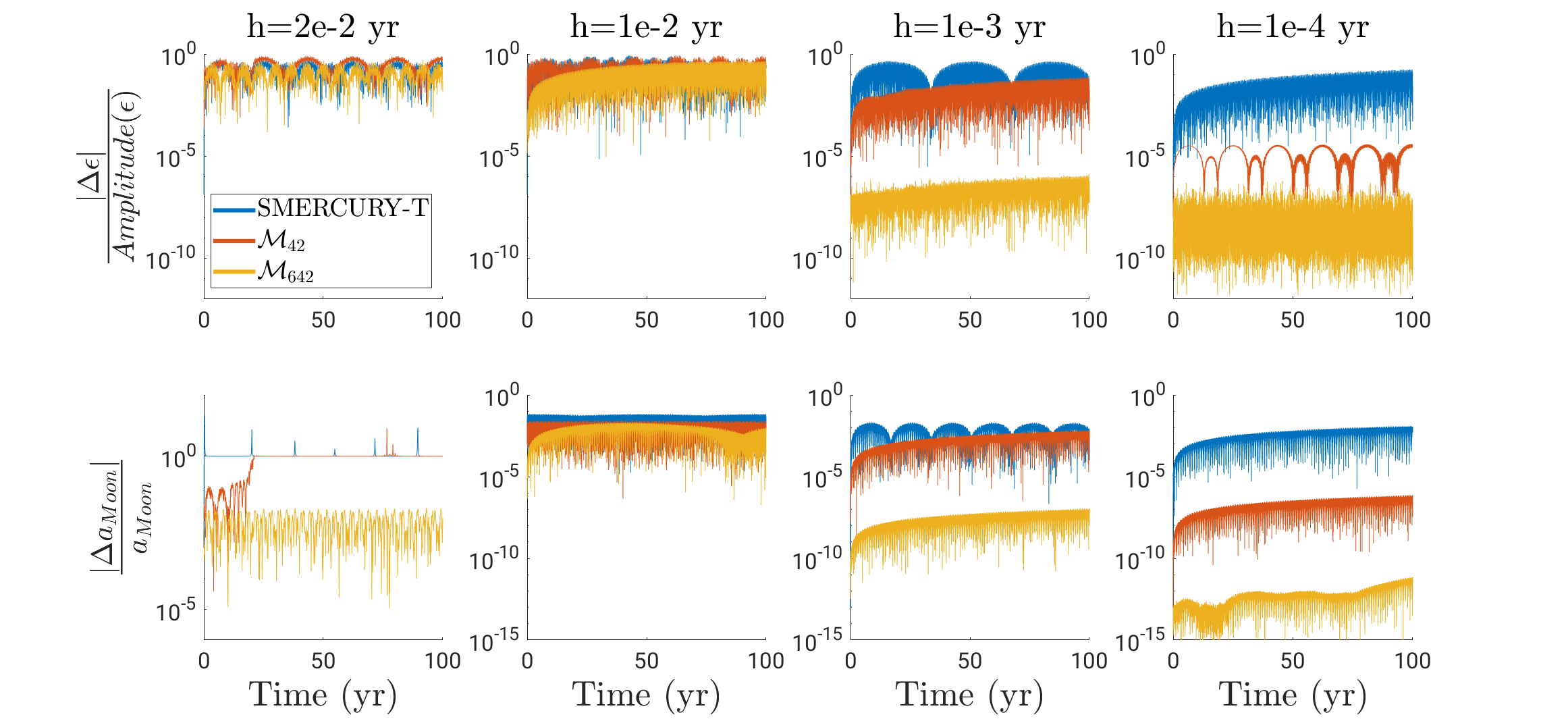}
\centering
\caption{Error of Earth's obliquity ($\epsilon$) over the range of $\epsilon$'s fluctuation and the relative error of the semi-major axis of the Moon for the \textit{Sun-Earth-Moon} system. Earth, Sun, Moon are rigid body, point mass and point mass respectively. The benchmark is simulated using the $\Tc_6$ scheme with $h=10^{-5}$ yr.}
%\ tao{change obliquity error to relative obliquity error. do both blow up when h=1e-1?}
% \chen{Upper plot: T=100, long double is used for the convergence of $\Mc_{642}$ when h changes from 1e-3 to 1e-4}
% \chen{Lower plot: T=1000 for long enough (not around 1 second) simulation time, double is used for simulation speed.}
\label{fig:convergence_SEM}
\vskip 0.1in
\end{figure}

% \begin{table}
% \centering
% \begin{tabular}{ c | c | c | c | c } 
%  h (yr) & $2\cdot 10^{-2}$ & $10^{-2}$ & $10^{-3}$ & $10^{-4}$ \\
%  \hline
%  $\Mc_{42}$  & 1.461 & 1.599 & 4.027 &  28.566 \\
%  $\Mc_{642}$ & 1.907 & 2.439 & 12.787 & 114.54 \\
%  SMERCURY-T  & 16.771 & 17.091 & 18.100 & 29.353 \\
% \end{tabular}
% \caption{Corresponding simulation time (seconds) for simulations in \cref{fig:convergence_SEM} ($T=1000$ yr). Data are output with time scale $0.02$ year for all methods.}
% \chen{Facts that can not be ignored: IO time; SMERCURTY-T may need more time to evaluate non regular orbits.}
% \chen{Feel like, we should isolate column $10^{-3}$ to the efficiency compare part in case of misleading when comparing in row.}
% \label{tb:time_T_vs_M}
% \end{table}

% \begin{figure}
% \includegraphics[width=\linewidth]{pic/convergence_test_trappistI.eps}
% \centering
% \caption{Relative error of obliquity and semi-major axis. The experiments are performed on the Trappist-I system with planet-b the rigid body. The benchmark is simulated using the $\Tc_6$ scheme with h=1e-5.}
% \tao{what is the message we're trying to deliver here? we're worse than smercury-t, even though we're 4th order and they are just 2nd-order?} \tao{Chen and I agreed to remove the Trappist-I example.}
% \label{fig:convergence_trappistI}
% \end{figure}

\begin{figure}
\includegraphics[width=\linewidth]{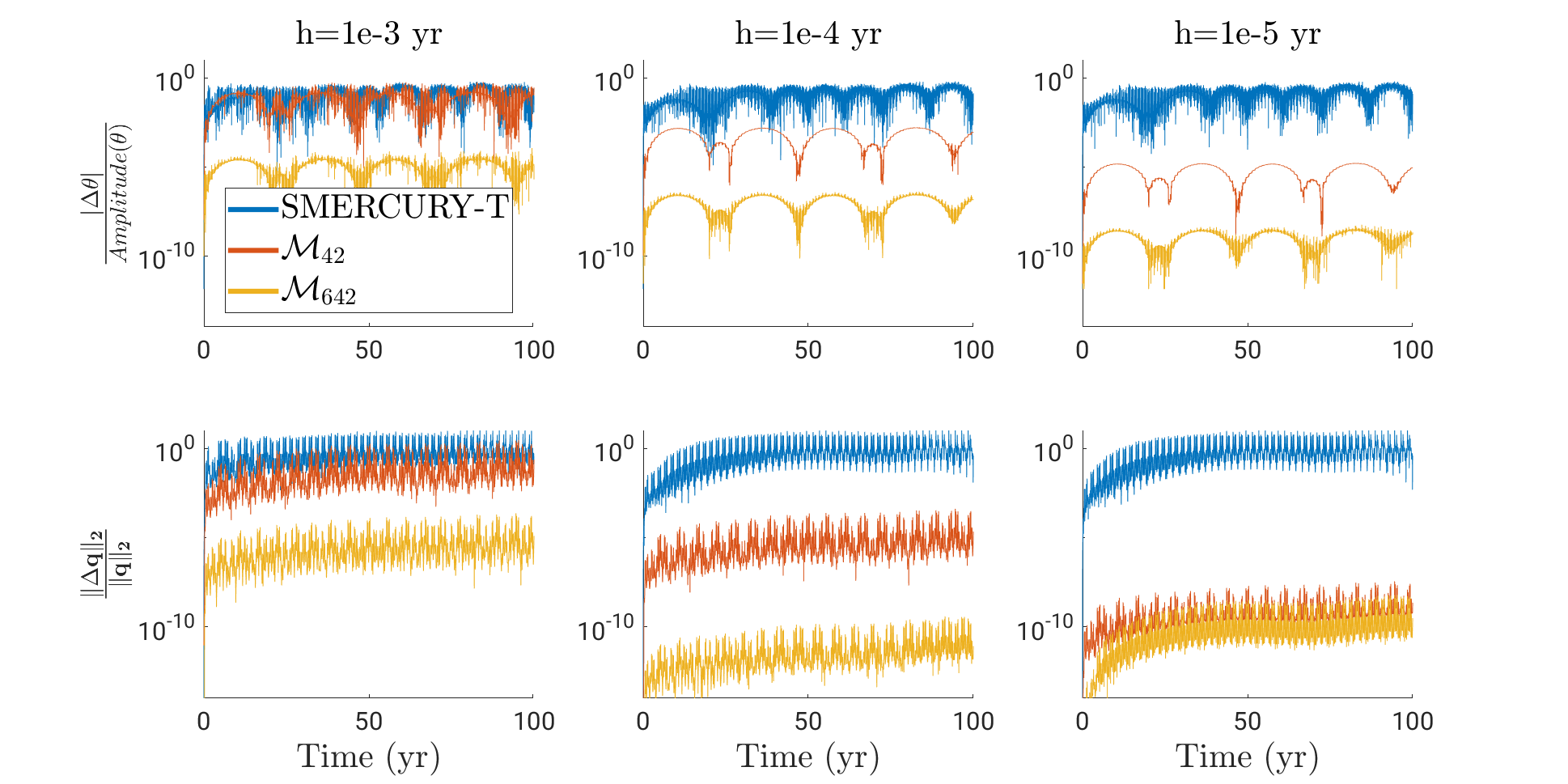}
\centering
\caption{Error of spin angle (the angle between the angular momentum and the $z$-axis of the inertia frame) and position for an Earth-like planet orbiting around two stars alternatively. The Earth-like planet, star 1 and star 2 are set to be rigid body, point mass and point mass respectively. The benchmark is simulated using the $\Tc_6$ scheme with h=1e-5.}
%\tao{state what problem this figure is about.}
%\tao{also, i can't really tell from the results that this is a case smercury-t can't do. results of all three examples look alike and they also appear to be that smercury-t is worse because it's lower-order.}
% \tao{let's take a look at h=1e-5 results too.}
% \tao{provide initial conditions and mass values.} 

%\chen{in appendix or here, I also want to add this example in our package then acknowledge the finder, may I?}
%\tao{I'd like to understand this question and acknowledge generously, of course.}
%\chen{It is an ancient question: I was thinking about the examples in our package (now the examples are fixed as we discussed). But since the IC of binary capturing problem is not easy to find, how about acknowledging Prof Tao here around the IC.}
\label{fig:convergence_binary}
\end{figure}

\subsubsection{Investigation of Efficiency}
\label{sec:NumericsEfficiency}
%\tao{Update: (i) let's keep both the h=1e-3 and 1e-4 results (4 tables). (ii) when longer name of $K_2$ is needed, let's call it Touma-Wisdom-Holman if Renyi really believes it's just a higher-order version of Touma-Wisdom.}

We now demonstrate the improved computational efficiency of the tailored splitting schemes. A comparison of the time efficiency among the traditional splitting method $\Tc_4, \Tc_6$ and the tailored splitting scheme $\Mc_{42},\Mc_{642},\Kc_{\cdot 2}$ in the $10$ rigid body (Sun with 8 planets and the Moon) is shown in \cref{tb:time_10rigid}.
% \tao{please add K2, not only here but also in the figure caption.}
$\Mc_{42}$ ($\Mc_{642}$) is about twice the speed of $\Tc_4$ ($\Tc_6$) with comparable integrating accuracy.
Note \smercuryt{} cannot be compared against here, because its currently available version\footnote{\url{https://github.com/SMKreyche/SMERCURY-T/tree/cbc25299825559f255cee096c7650f379af41aa5}}
% \tao{Renyi: please add specification here}
can only set one of the objects as rigid-body.

\begin{table}[ht]
\centering
\begin{tabular}{ c | c | c } 
 $h=10^{-3}$ yr & Wall time (s) & MAE of Earth's Obliquity\\ 
 \hline
 $\Tc_4$     & 30.573 & 1.996646e-05 \\
 $\Mc_{42}$  & 15.488 & 1.997454e-05 \\
 $\Tc_6$     & 72.55  & 1.728156e-08 \\ 
 $\Mc_{642}$ & 40.626 & 4.365093e-10 \\
 $\Kc_{\cdot2}$ &  14.673 & 2.186379e-05 \\
 SMERCURY-T  &   N/A      & N/A          \\
 \hline
 $h=10^{-4}$ yr &  & \\ 
 \hline
 $\Tc_4$     & 273.71 & 2.680897e-09 \\
 $\Mc_{42}$  & 140.58 & 3.817091e-09 \\
 $\Tc_6$     & 708.26 & 8.689218e-11 \\ 
 $\Mc_{642}$ & 395.52 & 2.039980e-10\\
 $\Kc_{\cdot2}$ &  131.56 & 1.292609e-05 \\
 SMERCURY-T  &   N/A      & N/A          \\
 \hline
 $h=10^{-5}$ yr &  & \\ 
 \hline
 $\Kc_{\cdot2}$ & 1299.9 & 1.378428e-07 \\
\end{tabular}
% \chen{$\Kc_{\cdot2}$ is very misleading here: the error of $h$=1e-3 (2.186e-05) is only twice the error of $h$=1e-4 (1.292e-05), in fact, when $h$=1e-5, error of $\Kc_{\cdot2}$ is around 1.388e-07 and the error keep converges with scale of around $10^{-2}$ as $h$ decreases to $h$=1e-6, 1e-7 etc.} \tao{very nice thinking and effort. i believe this and it is consistent with the theory: Wisdom-Holman can look good at large stepsize but it is actually not; only when h is small enough will it start to converge. in this sense, i recommend just keep the results as they are and remove our discussion in the comments.}\chen{I agree. Before comment this discussion out, one more question: do we add h=1e-5 K2 in the table so all the methods converge?}\tao{okay, you can first add that (just K2) and I will see if we keep it.}\chen{updated} \tao{great, let's keep it}

\caption{Efficiency comparison among scheme $\Tc_4$, $\Tc_6$, $\Mc_{42}$, $\Mc_{642}$ and $\Kc_{\cdot 2}$. The Solar system with 8 planets and the Moon (10 rigid bodies in total) is simulated till $1000$ years with $h=10^{-3}$ yr and $h=10^{-4}$ yr for all schemes using a single thread. The benchmark is simulated using the $\Tc_6$ scheme with $h=10^{-5}$ yr and long-double precision. Mean absolute errors (MAE) of the Earth's obliquity (rad) are measured. Data is output every $0.1$ yr.}
% \tao{merge Table 2 and 3. remove h=1E-4 in Table 2.}
% \tao{new idea: no more need to compare against smercuryt here. instead, add a comparison of computational time in the Sun-Earth-Moon system.}
% \tao{this table appeared at a strange location. also, what's bad about it when compared to the h=1e-4 result? also, add $K_{02}$?}
% \chen{also add $\Kc_{02}$ for discussion. will remove if needed.}
\label{tb:time_10rigid}
\end{table}

% \chen{If h=1e-3 is used}\li{It is fine to use h=1e-3, since in planetary dynamics, people usually use 5\% of of the shortest period, which is ~0.003 yr here using the period of the Moon} \tao{Renyi, you see that people won't understand what you mean unless you say, `if h=1e-3 is used, THEN XXX'}:
To gain additional understanding of the performance of \GRIT{}, complementary results that include comparisons to \smercuryt{} are also provided. For a fair comparison, we continue using the Solar system example, which is a near Keplerian problem that \smercuryt{} specializes in, but we had to alter it by setting only the Earth to be a rigid body and all others as point masses. Results are in \cref{tb:time_1rigid_9pointmass}, where $\Mc_{42}$ shows improved accuracy over \smercuryt{}, while $\Mc_{642}$ is even more accurate however with traded-off time complexity.

\begin{table}[ht]
\centering
\begin{tabular}{ c | c | c } 
 $h=10^{-3}$ yr & Wall time (s) & MAE of Earth's Obliquity\\ 
 \hline
 $\Mc_{42}$    & 6.408  & 1.997119e-05 \\
 $\Mc_{642}$   & 23.122 & 3.841649e-10 \\
 \smercuryt{}  & 8.638  & 2.157662e-05 \\
 \hline
 $h=10^{-4}$ yr & & \\ 
 \hline
 $\Mc_{42}$    & 53.782 & 3.833661e-09 \\
 $\Mc_{642}$   & 216.09 & 1.990336e-10 \\
 \smercuryt{}  & 39.079 & 1.903458e-05 \\
\end{tabular}
% \chen{\smercuryt{}'s convergence gets stuck from h=1e-3 to h=1e-4 due to their inaccuracy, is it okay to keep both h=1e-3 and h=1e-4 here? TODO: Depends on what is finally remained, modifications of descriptions in the main text are needed.} \tao{yes, keep both}
%\tao{why do we need this table?}\chen{People need to know the efficiency and accuracy comparisons between \GRIT{} and \smercuryt{}. Otherwise, it is simply compare schemes of \GRIT{}}
\caption{Efficiency comparison among scheme $\Mc_{42}$, $\Mc_{642}$ and SMERCURY-T. The Solar system with 8 planets and the Moon (9 point masses and 1 rigid body (the Earth) in total) is simulated till $1000$ years with $h=10^{-3}$ yr and $h=10^{-4}$ yr for all schemes using a single thread. The benchmark is simulated using the $\Tc_6$ scheme with $h=10^{-5}$ yr and long-double precision. Mean absolute errors (MAE) of the Earth's obliquity (rad) are measured. Data is output every $0.1$ yr.}
\label{tb:time_1rigid_9pointmass}
\end{table}

Also for the sake of fairness, note that wall-clock counts are platform dependent and therefore should only be used as a qualitative (not quantitative) indicator. Experiments reported here are conducted on a machine with AMD Ryzen 7 3700X 8-Core Processor, 16 GB memory and the Linux distribution of openSUSE Leap 15.2. \GRIT{} was compiled using GNU C\texttt{++} compiler and \smercuryt{} using GNU Fortran compiler, both with the default compilation options. %In terms of programming languages, C\texttt{++} is used for \GRIT{} and Fortran is used for \smercuryt{}. % Compilation options are set as defaults for both packages. 
% \tao{if SMERCURY-T is based on fortran, we should state that ours is based on CPP while theirs is fortran; if both are cpp then we dont have to say anything.} \chen{I modified. Please check.} 
% \tao{please fill in the blanks} \chen{Updated. Compiler version is not provided in purpose. \GRIT{} integrate all compilation options in its Makefile and the Makefile will be up to date (say this year we are using g++9, next year it might be g++10).}
Single-thread is used for experiments in both \cref{tb:time_10rigid,tb:time_1rigid_9pointmass} for fairness (note a parallelization option is available in \GRIT{}; we recommend turning it on when the simulated system has large numbers of rigid objects).
% \chen{both are using GCC compilers, but what's important is the options assigned and passed to the compiler.}
% and \GRIT{} is compiled using GCC compiler and \smercuryt is complied using gfortran compiler.
We also noted that \smercuryt{} slows down more significantly than \GRIT{} when its integration is outputted more frequently, and thus chose a large output step size to reduce \smercuryt{}'s I/O overhead so that the focus can be on the integration time itself.

\subsubsection{Summary of Sec.\ref{sec:numericalTests} Numerical Tests}

In general, \GRIT{} suits not only near-Keplerian orbits but also non-Keplerian ones. Multiple splitting and composition options are provided in \GRIT{} too so that, if preferred, a user can choose the classical Wisdom-Holman scheme for the orbital part which specializes in near-Keplerian orbits (e.g., $\mathcal{K}_{\cdot 2}$). Furthermore, equipped with higher order methods, \GRIT{} integrations have errors that decrease very rapidly as step size decreases in a reasonable range.

\subsection{%\revision{title updated}
Comparison with Secular Results}
To further verify the accuracy of our integration package, we compare our simulation results to secular theory here. We include two examples: the first one integrate the obliquity variation of a moon-less Earth without the influence of tidal interactions, and the second example considers tidal interactions between a hypothetical Earth-Moon system. We find good agreement between our simulation package with the results of the secular theory.

\subsubsection{Obliquity Variations of a Moon-less Earth}
Spin-orbit resonances lead to large obliquity variations for a moonless Earth \citep{Laskar93a}, and this classical example can serve as a test case for our simulation package. Specifically, planetary companions of the Earth (from Mercury to Neptune) all perturb Earth's orbit and lead to forced oscillations in the orbital plane of Earth. At the same time, torquing from the Sun leads to precession of Earth's spin axis. The natural precession frequency coincides with the forcing frequencies and drives resonant obliquity variations of Earth. Tidal interactions are weak in this case, so we neglected tidal effects in our code and considers the dynamical coupling between the planetary spin axes and its orbit. 

We include the eight Solar System planets in this system, and we adopt the position and velocity of the Solar System planets from JPL database \citep{giorgini1996jpl}. We only treat the Earth as a rigid object with oblateness of $0.00335$, and set the other planets and the Sun as point particles. 

\begin{figure}
\includegraphics[width=0.8\linewidth]{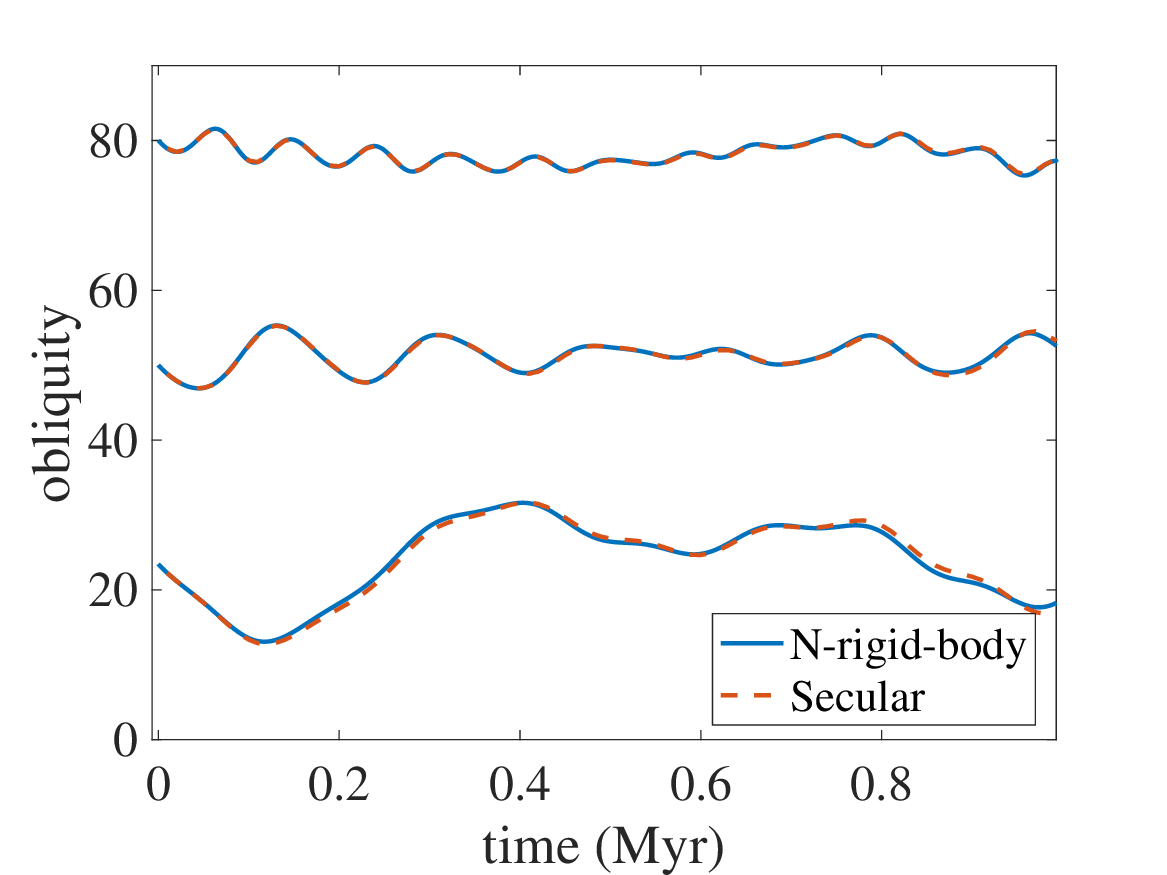}
\centering
\caption{Obliquity variations of a moon-less Earth. The solid lines represent the rigid-body simulations, and the dashed lines represent the secular results following \citep{Laskar93}. The results of our simulation package agree well with the secular theory.}
\label{fig:obl}
\end{figure}

Figure \ref{fig:obl} shows the comparison of the obliquity variations of the moon-less Earth with that from the secular theory shown in \citet{Laskar93a, li2014spin}. We included three examples starting with different initial obliquities, and all of them show good agreement with the secular results. In particular, below $\sim 40^\circ$, large obliquity variations can be seen due to the spin-orbit resonances. We chose a time step of $10^{-4}$yrs, in order to resolve the spin of the Earth. The fractional change in energy is at the order of $10^{-14}$ and the fractional change in angular momentum is at the order of $10^{-12}$ for all the three runs with different initial obliquities.

\subsubsection{Tidal Interactions of a Hypothetical Earth-Moon System}

To illustrate the accuracy of our simulation package including tidal interactions, we use a simple hypothetical Earth-Moon two-body system here. We set the initial semi-major axis and eccentricity to be $0.0018$AU and $0.4$. For the Earth, we set the spin period to be $1$ day, oblateness to be $0.00335$, love number to be $0.305$ and tidal time lag to be $698$sec. For the Moon, we set the spin period to be $14$ days, oblateness to be $0.0012$, love number to be $0.02416$ and tidal time lag to be $8,639$sec.
\begin{figure}
\includegraphics[width=0.8\linewidth]{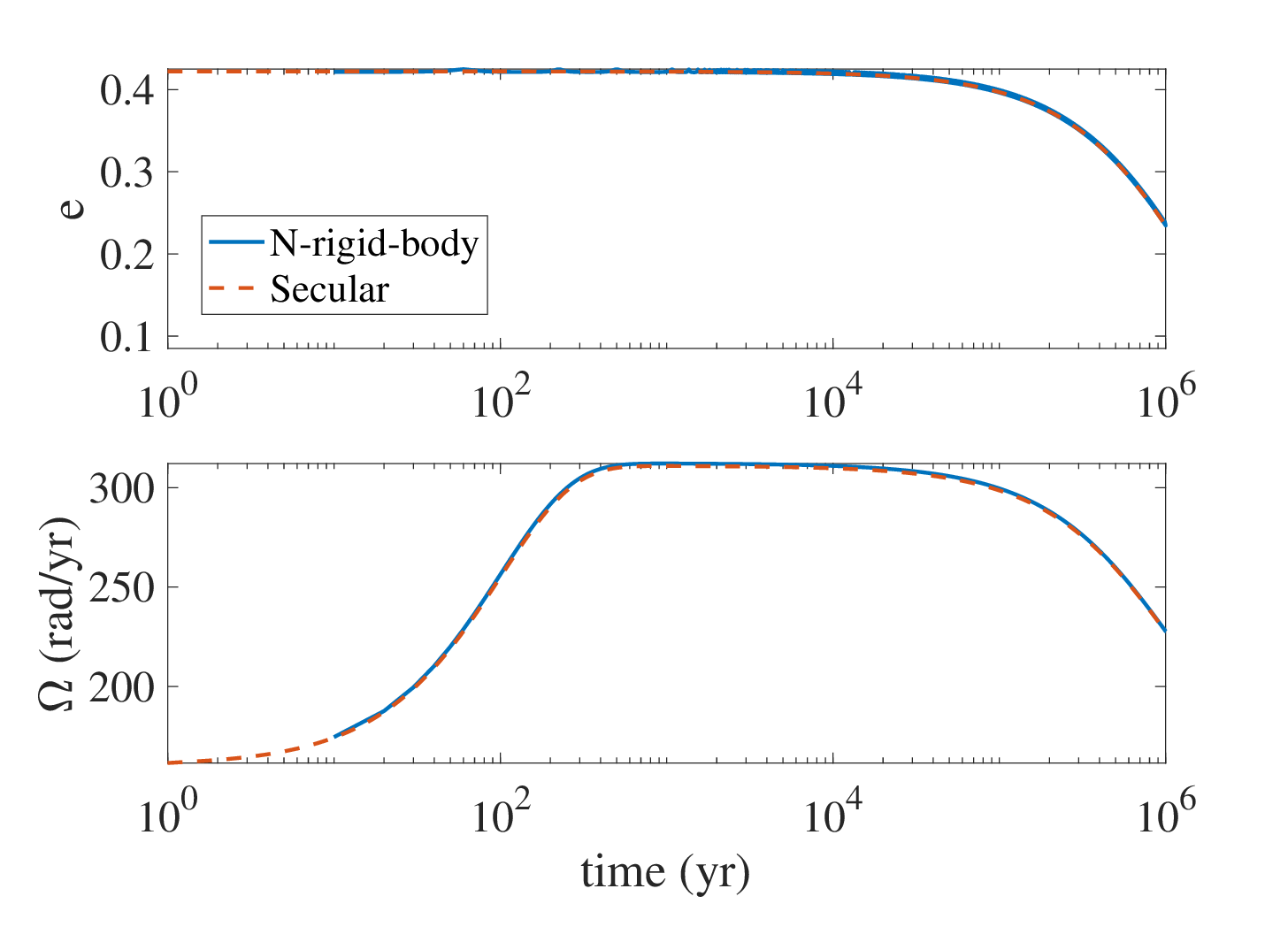}
\centering
\caption{Tidal interaction in a hypothetical Earth-Moon system. The spin rate ($\Omega$) increases rapidly to the pseudo-synchronized state, which is then followed by a much slower decay as the orbit circularizes under tide. The solid lines represent the simulation results and the dashed lines represent the secular results. The results of our simulation package agree well with the secular theory.}
\label{fig:tid}
\end{figure}

Figure \ref{fig:tid} shows the agreement between our simulation package (solid lines) with the secular results (dashed lines). The secular results are obtained following \citep{eggleton1998equilibrium}. The upper panel plots orbital eccentricity versus time and the lower panel plots the spin rate of the Moon versus time. It shows that the spin rate of the Moon increases to the pseudo-synchronized state within a few hundred years, and then slowly decreases as orbital eccentricity decays due to tide. We chose a time step of $10^{-4}$yr to resolve the spin of the Earth, and the total fractional change in angular momentum is $7\times10^{-12}$.

\section{Applications to Trappist-I}
\label{sec:Trap}
Spin-orbit coupling leads to profound dynamics in planetary systems, in particular for planets with close-in orbits. For Trappist-I, it is shown that tidal and rotational deformation of the planets leads to orbital precession that can be detected in the TTV measurements \citep{Bolmont20}. In addition, strong interactions between planets in resonant chains can push habitable zone Trappist-I planets into non-synchronous states \citep{Vinson19}. Recently, a high accuracy differentiable N-body code for transit timing and dynamical modeling has been developed, with applications to Trappist-I, yet tidal and GR effects have not been included \citep{Agol21}. 

To illustrate the effects of the spin-orbit coupling, we use our numerical package to simulate the long-term dynamics of spin-axis variations, as well as the short-term effects on TTV for Trappist-I. We note that both our numerical package and \textit{POSIDONIUS} \citep{Blanco-Cuaresma17, Bolmont20} consider tidal effects and spin-orbit coupling, beyond point mass dynamics based on Newtonian interactions and GR corrections. In particular, \citet{Bolmont20} obtained both dissipative and non-dissipative forces from tidal dissipation and tidal torquing separately, and considered forcing due to planetary rotational deformation. %In order to capture more accurately the spin-axis dynamics when the planets are not tidally synchronized due to strong planetary interactions, we use our N-rigid-body code to simulate the dynamical evolution of Trappist-I. 

Using our numerical package, we find that the habitable zone planets can indeed allow large spin-state variations, consistent with the findings by \citep{Vinson19}. In addition, we find that allowing the non-synchronized states could lead to significantly larger TTVs, which could reach a magnitude of $\sim min$ in ten-year timescale.

\subsection{System set up}
\label{sec:setup}
We use the same orbital initial condition and physical properties for the planets in Trappist-I following \citep{Bolmont20} (Table A.2 in \citealt{Bolmont20}), in order to compare the magnitude of TTVs, and we use the same reference tidal parameters for the star and the planet (e.g., $k_{2f,*} = 0.307$, $k_{2f,p} = 0.9532$, $\Delta \tau_{p} = 712.37$sec). The $Q$ coefficient in the tidal model can then be calculated ($k_2 = Q/(1-Q)$) \citep{eggleton1998equilibrium, Eggleton01, Fabrycky07}.

To calculate the moment of inertia along the three principal axes ($A$, $B$, $C$), we follow the derivation by \citep{VanHoolst08}, assuming a homogeneous model for simplicity and assuming the rotation velocity of the planet is close to the orbital velocity. Specifically, the moment of inertia can be expressed as the following:
\begin{align}
    A&=I_0(1-\frac{1}{3}\alpha-\frac{1}{2}\beta) \nonumber \\
    B&=I_0(1-\frac{1}{3}\alpha+\frac{1}{2}\beta) \nonumber \\
    C&=I_0(1+\frac{2}{3}\alpha) \nonumber
\end{align}
where
\begin{align}
    \alpha&=\frac{5}{4}q(1+k_f) \nonumber \\
    \beta&=\frac{3}{2}q(1+k_f) \nonumber
\end{align}
and $k_f$ is the love number, and $q$ is the ratio of the centrifugal acceleration to the gravitational acceleration. We assume all the planets have the same radius of gyration squared $rg_p^2 = 0.3308$ following \citet{Bolmont20}, and we include in Table \ref{tb:parameters} the moment of inertia of the planets.

Moreover, because the planets are very close to their host star, general relativistic precession plays a non-negligible role in the transit time. Thus, we also included the first order post-Newtonian correction in our simulation code (see \textsection \ref{sec:GR}).% \tao{see sec.3.4. consider changing its title to be `Adding Non-Conservative Forces'}\chen{updated}.

% \begin{table}
% \label{tb:parameters}
% \centering
% \begin{tabular}{||c c c c ||} 
%  \hline
%  Planet & A & B & C \\ 
%  \hline\hline
%  b & 6 & 87837 & 787 \\ 
%  \hline
% \end{tabular}
% \caption{Table to test captions and labels}
% \end{table}

\begin{table}
\begin{center}
\begin{tabular}{| c | c c c |}
\hline
 Planet & A (M$_\odot$km$^2$) & B (M$_\odot$km$^2$) & C (M$_\odot$km$^2$) \\ 
 \hline
b & 50.5245 & 50.8474 & 50.955 \\
c & 54.3319 & 54.4432 & 54.4803 \\ 
d & 7.6321 & 7.6396 & 7.6421 \\ 
e & 26.384 & 26.391 & 26.3933 \\
f & 39.9836 & 39.9898 & 39.9918 \\
g & 58.8644 & 58.8698 & 58.8716 \\ 
h & 7.8901 & 7.8904 & 7.8905 \\
\hline
\end{tabular}
\label{tb:parameters}
\end{center}
\caption{Principal moment of inertia of the Trappist-I planets.}
\end{table}

\subsection{Transit-timing Variations}
\label{sec:TTV}

The measurement of transit-timing variations (TTVs) is a powerful method to derive physical properties of planets, in particular masses and eccentricity of planets \citep{Agol18}. Most studies consider only point-mass dynamics. However, full-body dynamics including tidal effects and distortion of the planets could also play an important role \citep{Miralda-Escude02, Heyl07, Ragozzine09, Maciejewski18}. It is recently shown that new measurements of the TTV of the Trappist-I system lead to significant increase in the mass estimate for planet b and c, which may be due to unaccounted physical processes including tidal effects and rotational distortion of the planets \citep{Grimm18,Agol20, Bolmont20}. Thus, we use our simulation package to estimate the TTV of the inner planets in Trappist-I here as an example, in comparison with the study by \cite{Bolmont20}.

\begin{figure*}
\includegraphics[width=0.8\linewidth]{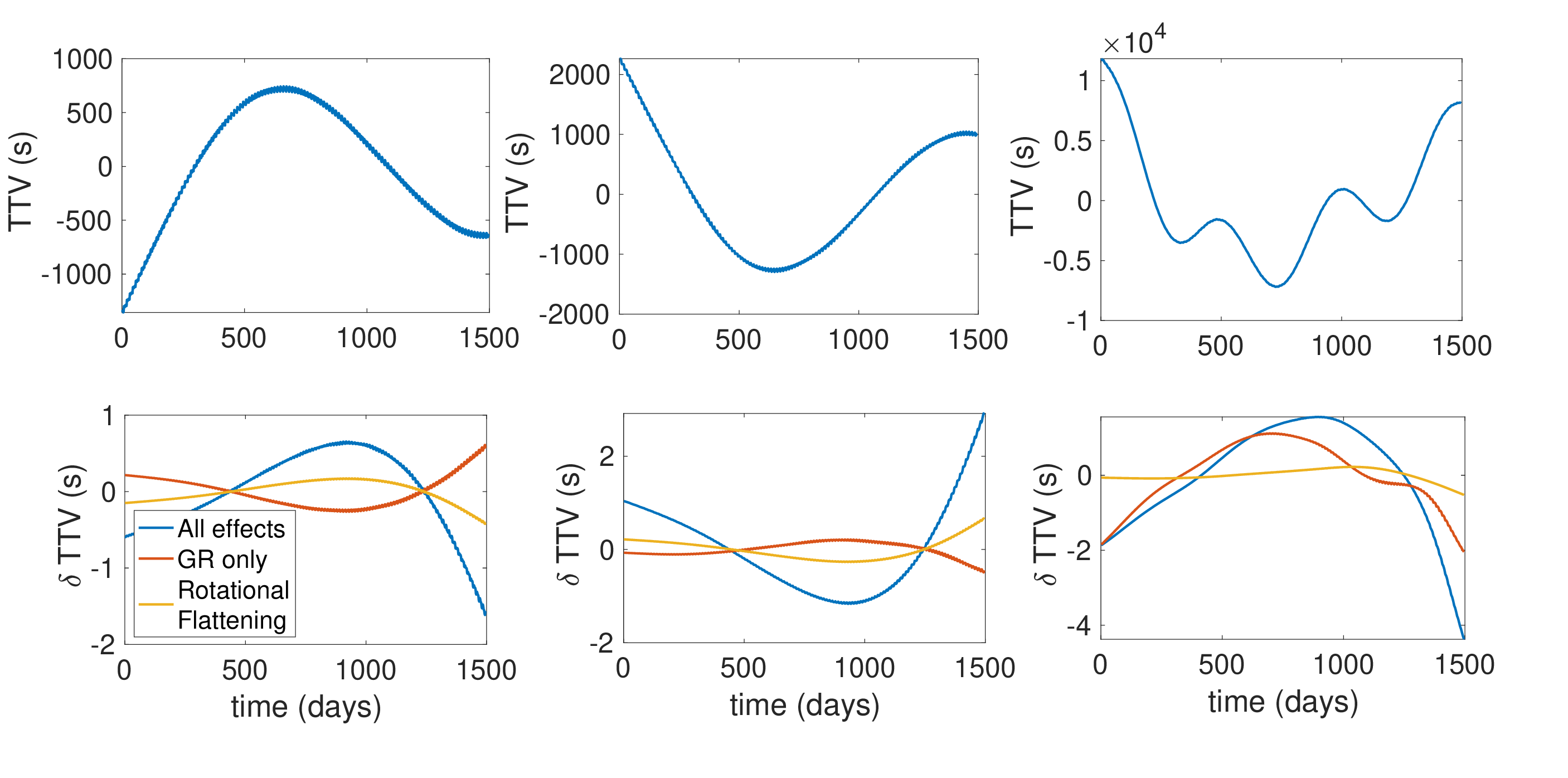}
\centering
\caption{Transit-timing variations (TTVs) of planets (b, c, d) in Trappist-I. The upper panels show TTVs of the planets assuming they are point mass particles and neglect effects due to GR. The lower panels show the differences in TTVs due to GR, rotational flattening of the planets and all the effects (GR, rotational flattening, tidal precession and tidal dissipation combined). The differences due to GR and rotational flattening are consistent with the results in \citet{Bolmont20}, while assuming the planets to be rigid bodies, the TTV differences are larger.}
\label{fig:TTV}
\end{figure*}

We include the result of the transit-timing variations for Trappist-I b,c and d in $1500$ days in Figure \ref{fig:TTV}, to compare our results with those in \citet{Bolmont20}. Similar to Figure 1 in \citet{Bolmont20}, the upper panels show the transit timing variations assuming the planets are all point-mass particles, and the lower panels show the differences in the TTVs due to different effects. The differences due to GR and rotational flattening of the planets computed using our simulation package are agreeable with that in \citet{Bolmont20}. Different from \citet{Bolmont20}, we assume the objects are rigid-bodies when considering tidal interactions with the central star using our rigid-body simulator. This leads to slightly larger TTV differences. We note that the magnitude of the differences in the TTVs depend on the misalignment between the elongated principal axis and the location direction of the planet from the central star. For the illustrative example, we assume the planets all start with their long-axes perfectly aligned to the direction of the central star. 

\begin{figure}
\includegraphics[width=0.8\linewidth]{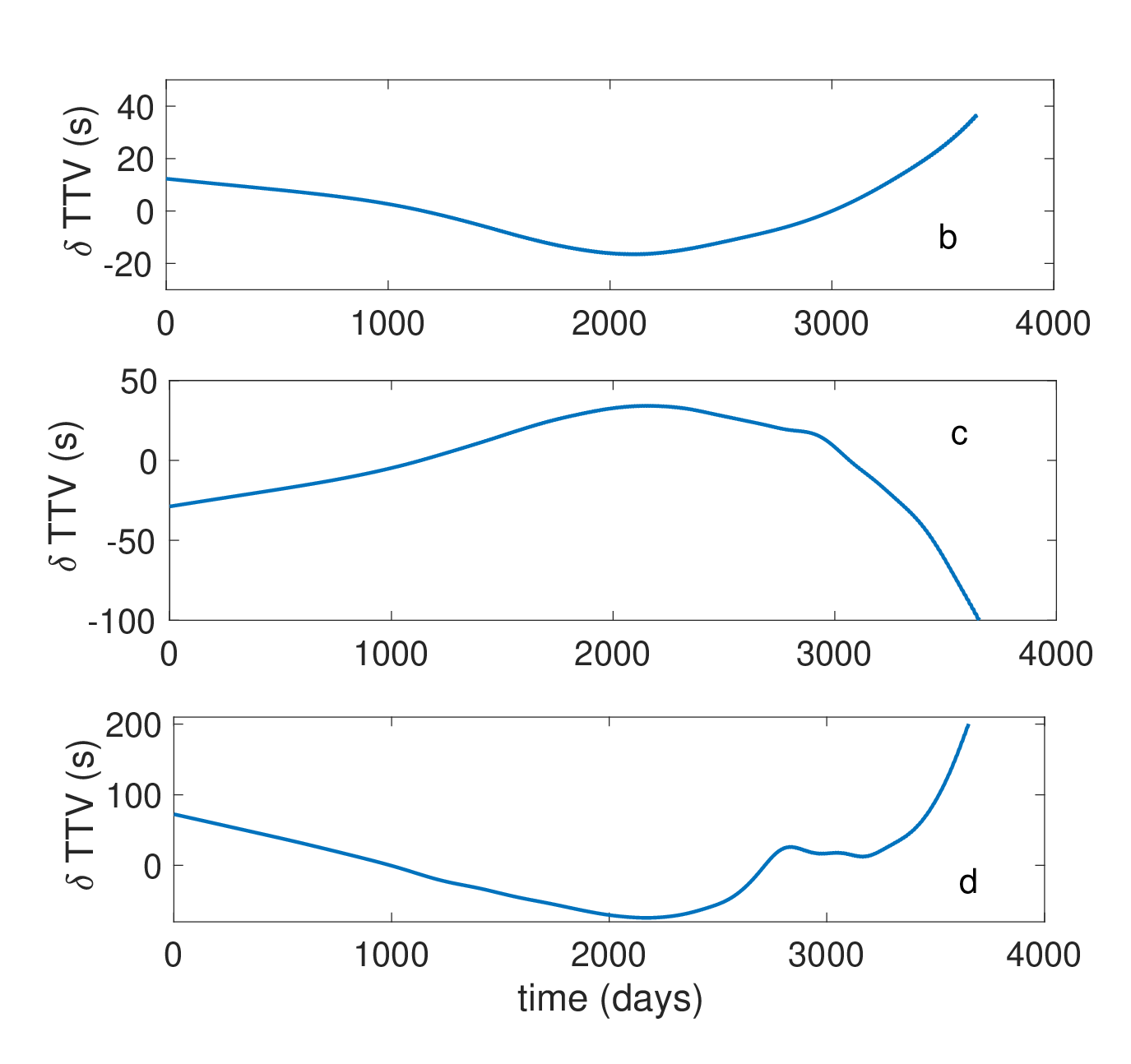}
\centering
\caption{Differences in transit-timing variations (TTVs) of planets (b, c, d) in Trappist-I over ten year measurements. With larger spin-misalignment, transit-timing variations could reach $\sim$ mins.}
\label{fig:10yrTTV}
\end{figure}

As the system evolve further, the misalignment could be excited to larger values (as discussed further in section \ref{sec:LT}). The differences in TTVs could reach $\sim 5$sec for 1500 days, and a few minutes in 10 year measurements, shown in Figure \ref{fig:10yrTTV}. A detailed study of how the TTVs depend on the physical properties of the planets (e.g., the love number, tidal time lag, etc.) %\tao{what do you mean by `the dependence of the differences'? depend on what?} \li{corrected} 
is out of the scope of this paper, and will be discussed in a follow up project.

\subsection{Long-term dynamics}
\label{sec:LT}

Long-term dynamics of spin-axes of planets, in particular their synchronized states, play an important role in the atmosphere circulation of the planets. When the planets are tidally locked, the extreme temperature differences on one side of the planet facing the star from the other side may lead to the collapse of planetary atmosphere \citep{Kasting93, Joshi97, Wordsworth15}. For Trappist-I, \citet{Vinson19} developed a framework studying the spin-axis variations of the planets and found that the mean motion resonant chain could drive the habitable zone planets out of the synchronized state.

\begin{figure}
\includegraphics[width=0.9\linewidth]{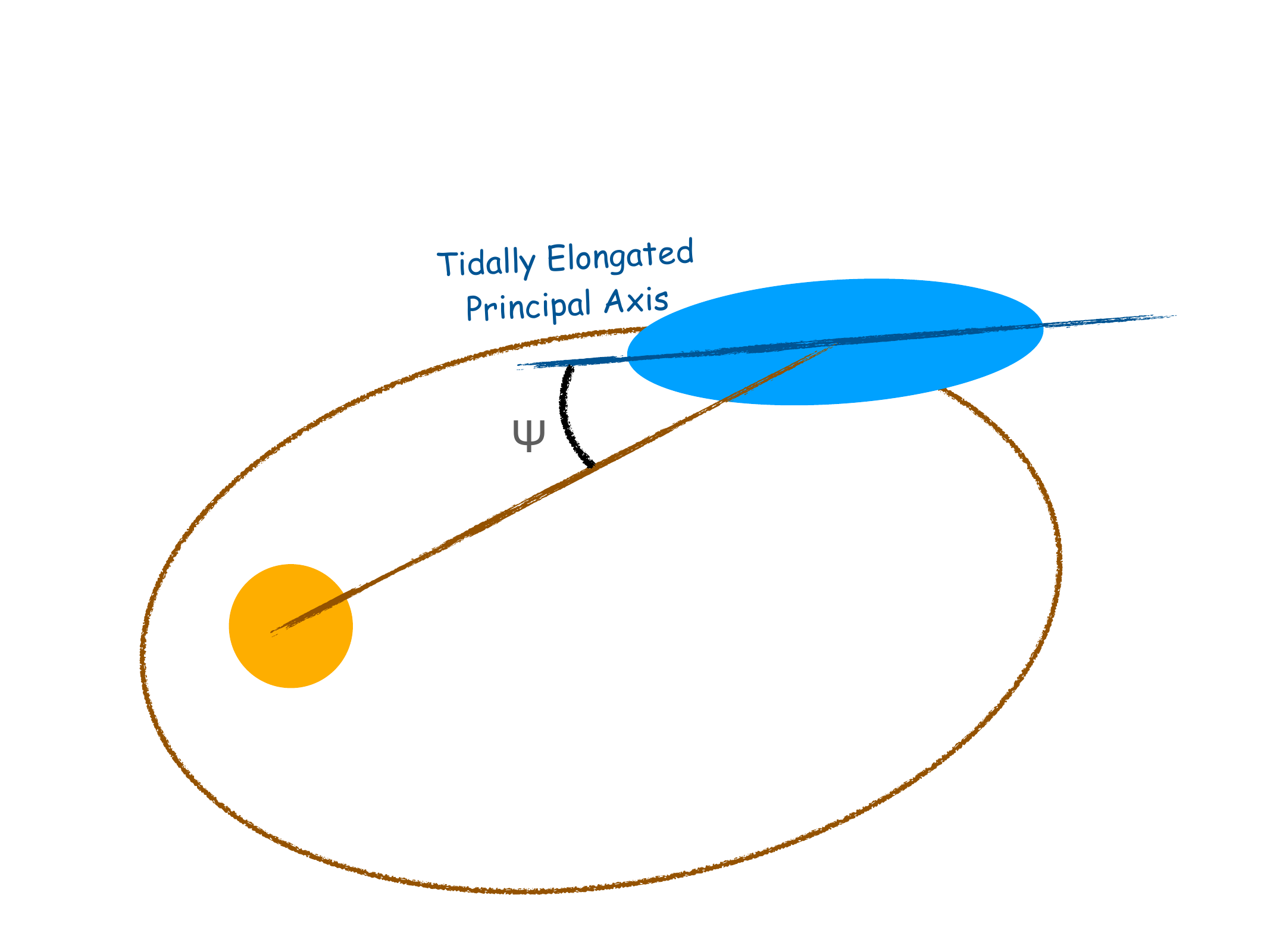}
\centering
\caption{Illustration of the long-axis misalignment. Low variations in $\psi$ correspond to a tidally locked planet.}
\label{fig:psi}
\end{figure}

Specifically, \citet{Vinson19} evolves the longitude of the substellar point separately based on results of orbital evolution of Trappist-I using the \textit{Rebound} simulation package \citep{Tamayo17}. This does not include effects of the variation of the spin-axis on the orbits, and the developed framework neglected the 3-D variations of the planetary spin-axis (i.e., assuming zero planetary obliquities) for simplicity. To evaluate the spin-axis dynamics more accurately, we use our simulation package, which allows backreactions of the spin-axis dynamics on the orbit, as well as the full 3-D dynamics of the planetary spin-axis.

\begin{figure}
\includegraphics[width=0.9\linewidth]{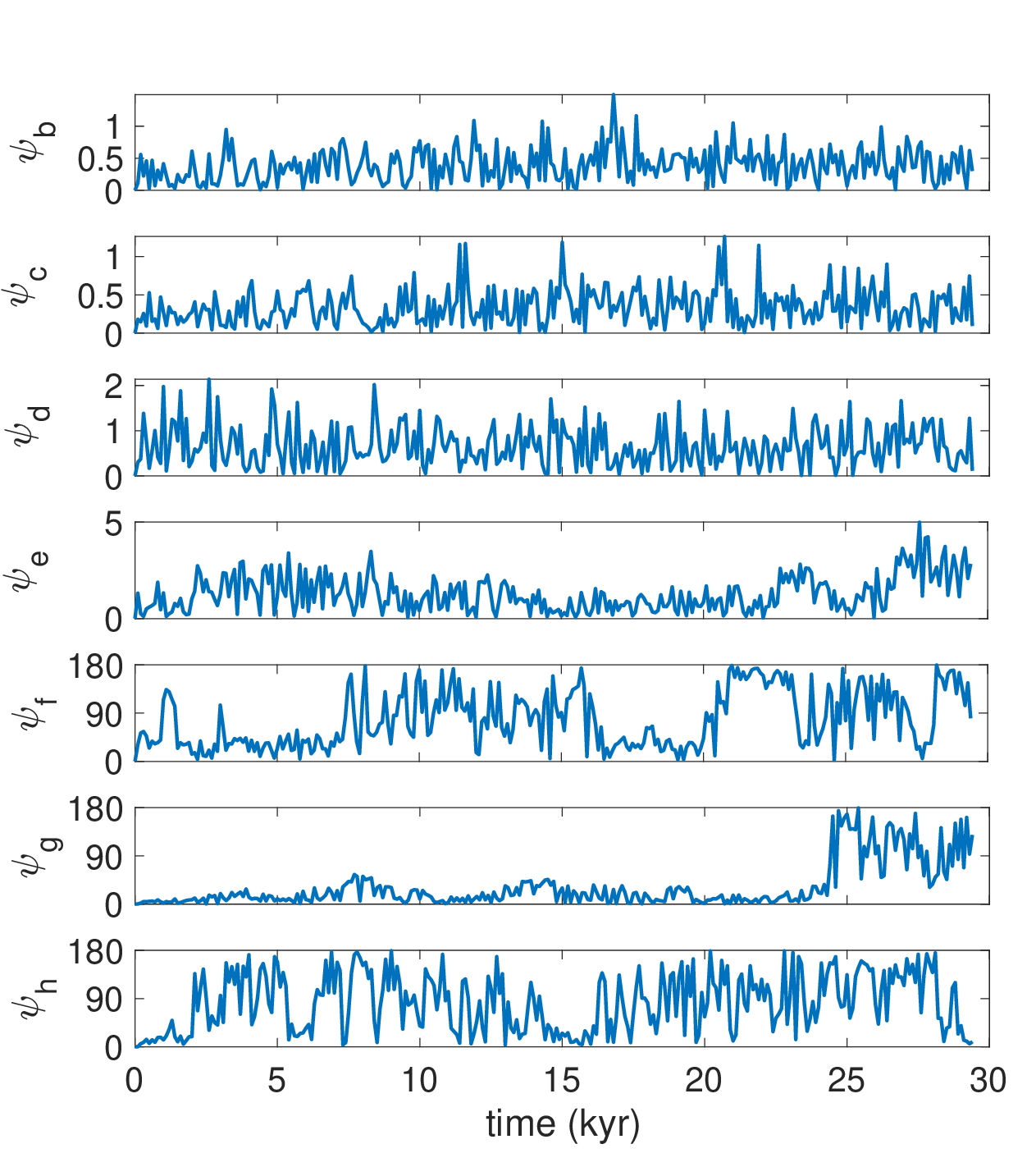}
\centering
\caption{Spin-axis misalignment as a function of time. Planet f, g and h all have large long-axes variations, and are not tidally locked. }
\label{fig:spin_LT}
\end{figure}

We use the same initial condition as those in section \ref{sec:TTV} for the long-term dynamical simulation over $100,000$yrs. We start the planets in synchronized configurations and we calculate the misalignment between the long axes of the planets and their radial direction from the host star, which is $\psi$ illustrated in Figure \ref{fig:psi}.  

Figure \ref{fig:spin_LT} shows this misalignment ($\psi$) %\tao{what is location direction?} \li{corrected} 
of the planets. Planet b, c, d and e are closer to the host star, and allow stronger tidal interactions. This leads to low variations in the long-axes of the planets. However, planet f, g and h are further away, where planetary interactions could compete with tidal re-alignment and drive larger spin-axis variations. We note that the obliquities of these planets still remain low (within a few degrees). %As a caveat, we note that for planets with rapidly varying spin-axis, distortion due to tide would be less pronounced \tao{i think i missed why this is a caveat.}. However, this doesn't change our qualitative conclusion on the unsynchronized state of the planets. 
The detailed dependence of the spin-axis variations on the parameters of the planets are beyond the scope of this article, and will be invested in a follow up paper.

\section{Conclusions}

In this article, we developed symplectic integrators and provided a package ``\texttt{GRIT}'' for studying the spin-orbit coupling of N-rigid-body systems. We split the Hamiltonian into four parts with different evolution timescales (tailored splitting), and compose the four parts together in a hierarchical way so that the expensive slow scale evolution is more efficient. In general, the tailored splitting is more flexible and efficient than the traditional splitting. 

To illustrate the validity of the integrator, we showed that it provides results consistent with the secular theories for the obliquity variation of a moonless Earth, and the tidal evolution of a hypothetical Earth-Moon system. This allowed us to confidently apply it to the less well understood system Trappist-I, and show that the differences in transit-timing variations could reach a few seconds for a four year measurements, and planetary interactions could push planet f, g and h out of the synchronized states, which are consistent with \citet{Bolmont20} and \citet{Vinson19}.

We assume the objects are rigid bodies in our simulation package. This is a good approximation when the deformation of the objects are slow. Thus, our simulation package can be applied for objects with a slow change of rotation rate or tidal distortion. When the deformation rate is faster than the orbital variation timescales, spin-orbit coupling using hydrodynamical simulations could provide more accurate results \citep[e.g.,][]{Li21}. Beyond planetary systems, the rigid-body integrator can also be applied to asteroid binaries, which exhibit interesting dynamical properties due to spin-orbit coupling \citep{Fahnestock08, Davis20, Meyer21}.

\section*{Acknowledgement}
%\tao{Gongjie, please add your grant info and people to acknowledge. Also, please see if you'd like to incorporate (responses to) comments that you received.}\li{Added. We've incorporated the comments by Matija (implementing the W-H integrator), and I just added a line to cite Tamayo+20 following David's comment. The other comments that I received were mostly for clarifications on the package, which we've already made clear in the paper.}
%\li{Renyi, let us make sure that we mentioned that we have up to the fourth order (spin-spin coupling), also, let us add this reference: Tamayo et al. (2019) when we add the additional forces (tides and GR)}
%\chen{4th order potential is mentioned in the main text (in the expression of the potential and in the experiments section).}\li{SG, thanks Renyi!}
%\revision{updated acknowledgement}
The authors thank Sergio Blanes, Matija \'{C}uk, David Michael Hernandez, and Billy Quarles for helpful discussions. We also thank the anonymous review which significantly improved the quality of this article. RC and MT are grateful for the partial support by NSF DMS-1847802. GL is grateful for the partial support by NASA 80NSSC20K0641 and 80NSSC20K0522.

\bibliographystyle{hapj}
\bibliography{references.bib}

\begin{thebibliography}{72}
\expandafter\ifx\csname natexlab\endcsname\relax\def\natexlab#1{#1}\fi

\bibitem[{{Agol} {et~al.}(2020){Agol}, {Dorn}, {Grimm}, {Turbet}, {Ducrot},
  {Delrez}, {Gillon}, {Demory}, {Burdanov}, {Barkaoui}, {Benkhaldoun},
  {Bolmont}, {Burgasser}, {Carey}, {de Wit}, {Fabrycky}, {Foreman-Mackey},
  {Haldemann}, {Hernandez}, {Ingalls}, {Jehin}, {Langford}, {Leconte},
  {Lederer}, {Luger}, {Malhotra}, {Meadows}, {Morris}, {Pozuelos}, {Queloz},
  {Raymond}, {Selsis}, {Sestovic}, {Triaud}, \& {Van Grootel}}]{Agol20}
{Agol}, E. {et~al.} 2020, arXiv e-prints, arXiv:2010.01074, 2010.01074

\bibitem[{{Agol} \& {Fabrycky}(2018)}]{Agol18}
{Agol}, E., \& {Fabrycky}, D.~C. 2018, {Transit-Timing and Duration Variations
  for the Discovery and Characterization of Exoplanets}, ed. H.~J. {Deeg} \&
  J.~A. {Belmonte}

\bibitem[{{Agol} {et~al.}(2021){Agol}, {Hernandez}, \& {Langford}}]{Agol21}
{Agol}, E., {Hernandez}, D.~M., \& {Langford}, Z. 2021, arXiv e-prints,
  arXiv:2106.02188, 2106.02188

\bibitem[{{Anderson} {et~al.}(1975){Anderson}, {Esposito}, {Martin},
  {Thornton}, \& {Muhleman}}]{Anderson75}
{Anderson}, J.~D., {Esposito}, P.~B., {Martin}, W., {Thornton}, C.~L., \&
  {Muhleman}, D.~O. 1975, \apj, 200, 221

\bibitem[{{Benitez} \& {Gallardo}(2008)}]{Benitez08}
{Benitez}, F., \& {Gallardo}, T. 2008, Celestial Mechanics and Dynamical
  Astronomy, 101, 289, 0709.1160

\bibitem[{{Blanchet}(2006)}]{Blanchet06}
{Blanchet}, L. 2006, Living Reviews in Relativity, 9, 4

\bibitem[{{Blanco-Cuaresma} \& {Bolmont}(2017)}]{Blanco-Cuaresma17}
{Blanco-Cuaresma}, S., \& {Bolmont}, E. 2017, in EWASS Special Session 4
  (2017): Star-planet interactions (EWASS-SS4-2017), 1712.01281

\bibitem[{Blanes \& Casas(2017)}]{blanes2017concise}
Blanes, S., \& Casas, F. 2017, A concise introduction to geometric numerical
  integration (CRC press)

\bibitem[{Blanes {et~al.}(2013)Blanes, Casas, Farres, Laskar, Makazaga, \&
  Murua}]{blanes2013new}
Blanes, S., Casas, F., Farres, A., Laskar, J., Makazaga, J., \& Murua, A. 2013,
  Applied Numerical Mathematics, 68, 58

\bibitem[{Blanes {et~al.}(2008)Blanes, Casas, \& Murua}]{blanes2008splitting}
Blanes, S., Casas, F., \& Murua, A. 2008, arXiv preprint arXiv:0812.0377

\bibitem[{{Bolmont} {et~al.}(2020){Bolmont}, {Demory}, {Blanco-Cuaresma},
  {Agol}, {Grimm}, {Auclair-Desrotour}, {Selsis}, \& {Leleu}}]{Bolmont20}
{Bolmont}, E., {Demory}, B.~O., {Blanco-Cuaresma}, S., {Agol}, E., {Grimm},
  S.~L., {Auclair-Desrotour}, P., {Selsis}, F., \& {Leleu}, A. 2020, \aap, 635,
  A117, 2002.02015

\bibitem[{{Bolmont} {et~al.}(2015){Bolmont}, {Raymond}, {Leconte}, {Hersant},
  \& {Correia}}]{Bolmont15}
{Bolmont}, E., {Raymond}, S.~N., {Leconte}, J., {Hersant}, F., \& {Correia}, A.
  C.~M. 2015, \aap, 583, A116, 1507.04751

\bibitem[{Bou-Rabee \& Marsden(2009)}]{bou2009hamilton}
Bou-Rabee, N., \& Marsden, J.~E. 2009, Foundations of Computational
  Mathematics, 9, 197

\bibitem[{{Breiter} {et~al.}(2005){Breiter}, {Nesvorn{\'y}}, \&
  {Vokrouhlick{\'y}}}]{Breiter05}
{Breiter}, S., {Nesvorn{\'y}}, D., \& {Vokrouhlick{\'y}}, D. 2005, \aj, 130,
  1267

\bibitem[{Celledoni {et~al.}(2008)Celledoni, Fass{\`o}, S{\"a}fstr{\"o}m, \&
  Zanna}]{celledoni2008exact}
Celledoni, E., Fass{\`o}, F., S{\"a}fstr{\"o}m, N., \& Zanna, A. 2008, SIAM
  Journal on Scientific Computing, 30, 2084

\bibitem[{Celledoni {et~al.}(2014)Celledoni, Marthinsen, \&
  Owren}]{celledoni2014introduction}
Celledoni, E., Marthinsen, H., \& Owren, B. 2014, Journal of Computational
  Physics, 257, 1040

\bibitem[{{Chambers}(1999)}]{Chambers99}
{Chambers}, J.~E. 1999, \mnras, 304, 793

\bibitem[{{Davis} \& {Scheeres}(2020)}]{Davis20}
{Davis}, A.~B., \& {Scheeres}, D.~J. 2020, The Planetary Science Journal, 1, 25

\bibitem[{Dullweber {et~al.}(1997)Dullweber, Leimkuhler, \&
  McLachlan}]{dullweber1997symplectic}
Dullweber, A., Leimkuhler, B., \& McLachlan, R. 1997, The Journal of chemical
  physics, 107, 5840

\bibitem[{Duncan {et~al.}(1998)Duncan, Levison, \& Lee}]{duncan1998multiple}
Duncan, M.~J., Levison, H.~F., \& Lee, M.~H. 1998, The Astronomical Journal,
  116, 2067

\bibitem[{Eggleton {et~al.}(1998)Eggleton, Kiseleva, \&
  Hut}]{eggleton1998equilibrium}
Eggleton, P.~P., Kiseleva, L.~G., \& Hut, P. 1998, The Astrophysical Journal,
  499, 853

\bibitem[{{Eggleton} \& {Kiseleva-Eggleton}(2001)}]{Eggleton01}
{Eggleton}, P.~P., \& {Kiseleva-Eggleton}, L. 2001, \apj, 562, 1012,
  astro-ph/0104126

\bibitem[{{Fabrycky} \& {Tremaine}(2007)}]{Fabrycky07}
{Fabrycky}, D., \& {Tremaine}, S. 2007, \apj, 669, 1298, 0705.4285

\bibitem[{{Fahnestock} \& {Scheeres}(2008)}]{Fahnestock08}
{Fahnestock}, E.~G., \& {Scheeres}, D.~J. 2008, \icarus, 194, 410

\bibitem[{Fass{\`o}(2003)}]{fasso2003comparison}
Fass{\`o}, F. 2003, Journal of computational physics, 189, 527

\bibitem[{Giorgini {et~al.}(1996)Giorgini, Yeomans, Chamberlin, Chodas,
  Jacobson, Keesey, Lieske, Ostro, Standish, \& Wimberly}]{giorgini1996jpl}
Giorgini, J. {et~al.} 1996, in AAS/Division for Planetary Sciences Meeting
  Abstracts\# 28, Vol.~28, 25--04

\bibitem[{{Grimm} {et~al.}(2018){Grimm}, {Demory}, {Gillon}, {Dorn}, {Agol},
  {Burdanov}, {Delrez}, {Sestovic}, {Triaud}, {Turbet}, {Bolmont}, {Caldas},
  {de Wit}, {Jehin}, {Leconte}, {Raymond}, {Van Grootel}, {Burgasser}, {Carey},
  {Fabrycky}, {Heng}, {Hernandez}, {Ingalls}, {Lederer}, {Selsis}, \&
  {Queloz}}]{Grimm18}
{Grimm}, S.~L. {et~al.} 2018, \aap, 613, A68, 1802.01377

\bibitem[{Hairer {et~al.}(2006{\natexlab{a}})Hairer, Lubich, \&
  Wanner}]{Hairer06}
Hairer, E., Lubich, C., \& Wanner, G. 2006{\natexlab{a}}, Geometric Numerical
  Integration: Structure-Preserving Algorithms for Ordinary Differential
  Equations, 2nd edn. (Berlin Heidelberg New York: Springer)

\bibitem[{Hairer {et~al.}(2006{\natexlab{b}})Hairer, Lubich, \&
  Wanner}]{hairer2006geometric}
------. 2006{\natexlab{b}}, Geometric numerical integration:
  structure-preserving algorithms for ordinary differential equations, Vol.~31
  (Springer Science \& Business Media)

\bibitem[{{Heyl} \& {Gladman}(2007)}]{Heyl07}
{Heyl}, J.~S., \& {Gladman}, B.~J. 2007, \mnras, 377, 1511, astro-ph/0610267

\bibitem[{Holm {et~al.}(2009)Holm, Schmah, \& Stoica}]{HoScSt2009}
Holm, D., Schmah, T., \& Stoica, C. 2009, Geometric mechanics and symmetry:
  from finite to infinite dimensions, Oxford texts in applied and engineering
  mathematics (Oxford University Press)

\bibitem[{{Hou} {et~al.}(2017){Hou}, {Scheeres}, \& {Xin}}]{Hou17}
{Hou}, X., {Scheeres}, D.~J., \& {Xin}, X. 2017, Celestial Mechanics and
  Dynamical Astronomy, 127, 369

\bibitem[{{Hut}(1981)}]{Hut81}
{Hut}, P. 1981, \aap, 99, 126

\bibitem[{Iserles {et~al.}(2000)Iserles, Munthe-Kaas, N{\o}rsett, \&
  Zanna}]{iserles2000lie}
Iserles, A., Munthe-Kaas, H.~Z., N{\o}rsett, S.~P., \& Zanna, A. 2000, Acta
  numerica, 9, 215

\bibitem[{{Joshi} {et~al.}(1997){Joshi}, {Haberle}, \& {Reynolds}}]{Joshi97}
{Joshi}, M.~M., {Haberle}, R.~M., \& {Reynolds}, R.~T. 1997, \icarus, 129, 450

\bibitem[{{Kasting} {et~al.}(1993){Kasting}, {Whitmire}, \&
  {Reynolds}}]{Kasting93}
{Kasting}, J.~F., {Whitmire}, D.~P., \& {Reynolds}, R.~T. 1993, \icarus, 101,
  108

\bibitem[{{Kreyche} {et~al.}(Submitted){Kreyche}, {Barnes}, {Quarles}, \&
  {Chambers}}]{Kreyche21}
{Kreyche}, S., {Barnes}, J., {Quarles}, B., \& {Chambers}, J. Submitted,
  Planetary Science Journal

\bibitem[{{Laskar} {et~al.}(1993){Laskar}, {Joutel}, \& {Robutel}}]{Laskar93a}
{Laskar}, J., {Joutel}, F., \& {Robutel}, P. 1993, \nat, 361, 615

\bibitem[{{Laskar} \& {Robutel}(1993)}]{Laskar93}
{Laskar}, J., \& {Robutel}, P. 1993, \nat, 361, 608

\bibitem[{Laskar \& Robutel(2001)}]{laskar2001high}
Laskar, J., \& Robutel, P. 2001, Celestial Mechanics and Dynamical Astronomy,
  80, 39

\bibitem[{{Lee} {et~al.}(2007){Lee}, {Leok}, \& {McClamroch}}]{Lee07}
{Lee}, T., {Leok}, M., \& {McClamroch}, N.~H. 2007, Celestial Mechanics and
  Dynamical Astronomy, 98, 121

\bibitem[{Lee {et~al.}(2005)Lee, McClamroch, \& Leok}]{lee2005lie}
Lee, T., McClamroch, N.~H., \& Leok, M. 2005, in Proceedings of 2005 IEEE
  Conference on Control Applications, 2005. CCA 2005., IEEE, 962--967

\bibitem[{Leimkuhler \& Reich(2004)}]{MR2132573}
Leimkuhler, B., \& Reich, S. 2004, Cambridge Monographs on Applied and
  Computational Mathematics, Vol.~14, Simulating {H}amiltonian dynamics
  (Cambridge: Cambridge University Press), xvi+379

\bibitem[{Li \& Batygin(2014)}]{li2014spin}
Li, G., \& Batygin, K. 2014, The Astrophysical Journal, 790, 69

\bibitem[{{Li} {et~al.}(2021){Li}, {Lai}, {Anderson}, \& {Pu}}]{Li21}
{Li}, J., {Lai}, D., {Anderson}, K.~R., \& {Pu}, B. 2021, \mnras, 501, 1621,
  2006.10067

\bibitem[{{Lissauer} {et~al.}(2012){Lissauer}, {Barnes}, \&
  {Chambers}}]{Lissauer12}
{Lissauer}, J.~J., {Barnes}, J.~W., \& {Chambers}, J.~E. 2012, \icarus, 217, 77

\bibitem[{{Maciejewski} {et~al.}(2018){Maciejewski}, {Fern{\'a}ndez},
  {Aceituno}, {Mart{\'\i}n-Ruiz}, {Ohlert}, {Dimitrov}, {Szyszka}, {von Essen},
  {Mugrauer}, {Bischoff}, {Michel}, {Mallonn}, {Stangret}, \&
  {Mo{\'z}dzierski}}]{Maciejewski18}
{Maciejewski}, G. {et~al.} 2018, ACTA ASTRONOMICA, 68, 371, 1812.02438

\bibitem[{{Mardling} \& {Lin}(2002)}]{Mardling02}
{Mardling}, R.~A., \& {Lin}, D.~N.~C. 2002, \apj, 573, 829

\bibitem[{Marsden \& Ratiu(1994)}]{marsden1994introduction}
Marsden, J.~E., \& Ratiu, T.~S. 1994, Introduction to mechanics and symmetry: a
  basic exposition of classical mechanical systems (Springer)

\bibitem[{McLachlan(1995)}]{mclachlan1995composition}
McLachlan, R.~I. 1995, BIT numerical mathematics, 35, 258

\bibitem[{McLachlan \& Quispel(2002)}]{mclachlan2002splitting}
McLachlan, R.~I., \& Quispel, G. R.~W. 2002, Acta Numerica, 11, 341

\bibitem[{{Meyer} \& {Scheeres}(2021)}]{Meyer21}
{Meyer}, A.~J., \& {Scheeres}, D.~J. 2021, \icarus, 367, 114554

\bibitem[{{Millholland} \& {Laughlin}(2019)}]{Millholland19}
{Millholland}, S., \& {Laughlin}, G. 2019, Nature Astronomy, 3, 424, 1903.01386

\bibitem[{{Miralda-Escud{\'e}}(2002)}]{Miralda-Escude02}
{Miralda-Escud{\'e}}, J. 2002, \apj, 564, 1019, astro-ph/0104034

\bibitem[{{Muirhead} {et~al.}(2015){Muirhead}, {Mann}, {Vanderburg}, {Morton},
  {Kraus}, {Ireland}, {Swift}, {Feiden}, {Gaidos}, \& {Gazak}}]{Muirhead15}
{Muirhead}, P.~S. {et~al.} 2015, \apj, 801, 18, 1501.01305

\bibitem[{{Ragozzine} \& {Wolf}(2009)}]{Ragozzine09}
{Ragozzine}, D., \& {Wolf}, A.~S. 2009, \apj, 698, 1778, 0807.2856

\bibitem[{Reich \& Zentrum(1996)}]{reich1996symplectic}
Reich, S., \& Zentrum, K.-Z. 1996, Fields Inst. Commun, 10, 181

\bibitem[{Sanz-Serna \& Calvo(1994)}]{calvo1994numerical}
Sanz-Serna, J., \& Calvo, M. 1994, Numerical {H}amiltonian problems, 1st edn.
  (Chapman and Hall/CRC)

\bibitem[{Suzuki(1990)}]{suzuki1990fractal}
Suzuki, M. 1990, Physics Letters A, 146, 319

\bibitem[{{Tamayo} {et~al.}(2017){Tamayo}, {Rein}, {Petrovich}, \&
  {Murray}}]{Tamayo17}
{Tamayo}, D., {Rein}, H., {Petrovich}, C., \& {Murray}, N. 2017, \apjl, 840,
  L19, 1704.02957

\bibitem[{{Tamayo} {et~al.}(2020){Tamayo}, {Rein}, {Shi}, \&
  {Hernandez}}]{Tamayo20}
{Tamayo}, D., {Rein}, H., {Shi}, P., \& {Hernandez}, D.~M. 2020, \mnras, 491,
  2885, 1908.05634

\bibitem[{Tao \& Ohsawa(2020)}]{tao2020variational}
Tao, M., \& Ohsawa, T. 2020, in International Conference on Artificial
  Intelligence and Statistics, PMLR, 4269--4280

\bibitem[{Tao {et~al.}(2010)Tao, Owhadi, \& Marsden}]{tao2010nonintrusive}
Tao, M., Owhadi, H., \& Marsden, J.~E. 2010, Multiscale Modeling \& Simulation,
  8, 1269

\bibitem[{{Touma} \& {Wisdom}(1994)}]{Touma94}
{Touma}, J., \& {Wisdom}, J. 1994, \aj, 107, 1189

\bibitem[{{Van Hoolst} {et~al.}(2008){Van Hoolst}, {Rambaux}, {Karatekin},
  {Dehant}, \& {Rivoldini}}]{VanHoolst08}
{Van Hoolst}, T., {Rambaux}, N., {Karatekin}, {\"O}., {Dehant}, V., \&
  {Rivoldini}, A. 2008, \icarus, 195, 386

\bibitem[{van Zon \& Schofield(2007)}]{van2007symplectic}
van Zon, R., \& Schofield, J. 2007, Physical Review E, 75, 056701

\bibitem[{Vilmart(2008)}]{vilmart2008reducing}
Vilmart, G. 2008, Journal of computational physics, 227, 7083

\bibitem[{{Vinson} {et~al.}(2019){Vinson}, {Tamayo}, \& {Hansen}}]{Vinson19}
{Vinson}, A.~M., {Tamayo}, D., \& {Hansen}, B. M.~S. 2019, \mnras, 488, 5739,
  1905.11419

\bibitem[{Wisdom \& Holman(1991)}]{wisdom1991symplectic}
Wisdom, J., \& Holman, M. 1991, The Astronomical Journal, 102, 1528

\bibitem[{{Wordsworth}(2015)}]{Wordsworth15}
{Wordsworth}, R. 2015, \apj, 806, 180, 1412.5575

\bibitem[{Yoshida(1990)}]{yoshida1990construction}
Yoshida, H. 1990, Physics letters A, 150, 262

\bibitem[{{Zhu} {et~al.}(2018){Zhu}, {Petrovich}, {Wu}, {Dong}, \&
  {Xie}}]{Zhu18}
{Zhu}, W., {Petrovich}, C., {Wu}, Y., {Dong}, S., \& {Xie}, J. 2018, \apj, 860,
  101, 1802.09526

\end{thebibliography}

\appendix

\section{Approximation of the Potential Energy}\label{appendix:potential}

%\revision{typos in the potential expression are fixed in the following three potential energy formulas.}

% \tao{more explanation than this is needed. if I understood correctly, you are not considering (rigid body + point mass), but (rigid body + rigid body). it is just when you expand to certain (low) order (in what?), it is the same as (rigid body + point mass), but then you also expand to higher order too. we need connecting words and descriptive language.}\chen{connecting words added in the main text.}

The procedure to approximate $V\left( \bm q_i, \bm q_j, \bm R_i, \bm R_j \right)$ in \cref{eq:true_potential} by Taylor expansion is shown below:

% Consider two ellipsoids $\Bc_i, \Bc_j$ (with mass $m_i$, $m_j$, position $\bm q_i \in \R^3$, rotation matrix $\bm R \in \mathsf{SO(3)}$
% and tensor of moment inertia $\bm J \left( \bm J_d \right) \in \R^{3 \times 3}$),
% and a  body (with mass $M$ and position $\bm q_0 \in \R^3$).
% The moment of inertia tensor $\bm J, \bm J_d$ are defined in section\ref{subsec:rigid_body_representation_1}.
% Assume $\mathcal{R} \ll \norm{q_0-q_1}$ with $\mathcal{R}$ the mean radius of  this rigid body.
\begin{align}
\begin{split}
      & V\left( \bm q_i, \bm q_j, \bm R_i, \bm R_j \right) \\
    = & \int_{\Bc_i} \int_{\Bc_j}
    - \frac{\Gc \rho(\bm x_i) \rho(\bm x_j)}
    {\norm{( \bm q_i + \bm R_i \bm x_i ) - ( \bm q_j + \bm R_j \bm x_j )}}
    \, d\bm x_i d\bm x_j \\
    = & \int_{\Bc_i} \int_{\Bc_j}
    - \frac{\Gc \rho(\bm x_i) \rho(\bm x_j)}
    {\sqrt{\norm{\qij}^2 + \norm{\bm R_i \bm x_i - \bm R_j \bm x_j}^2 + 2 (\qij)^T(\bm R_i \bm x_i - \bm R_j \bm x_j)}}
    \, d\bm x_i d\bm x_j \\
    = & \int_{\Bc_i} \int_{\Bc_j}
    - \frac{\Gc \rho(\bm x_i) \rho(\bm x_j)}{\norm{\qij}}
    \bigg( 1 - \frac{1}{2} \left[ \frac{\norm{\bm R_i \bm x_i - \bm R_j \bm x_j}^2 + 2 (\qij)^T(\bm R_i \bm x_i - \bm R_j \bm x_j)}{\norm{\qij}^2} \right] \\
    & + \frac{3}{8} {\left[ \frac{\norm{\bm R_i \bm x_i - \bm R_j \bm x_j}^2 + 2 (\qij)^T(\bm R_i \bm x_i - \bm R_j \bm x_j)}{\norm{\qij}^2} \right]}^2 \bigg)
    \, d\bm x_i d\bm x_j + \Oc(\eta^3) \\
    = & -\frac{\Gc m_i m_j}{\disqij} + \frac{\Gc\left( m_i Tr[\bm J_i^{(d)}] + m_j Tr[\bm J_j^{(d)}] \right)}{2 \disqij^3} - \frac{3 \Gc {\left( \qij \right)}^T \left( m_j \bm R_i \bm J_i^{(d)} \bm R_i^T + m_i \bm R_j \bm J_j^{(d)} \bm R_j^T \right)\left( \qij \right)}{2 \disqij^5}  + \Oc(\eta^3) \\
\end{split}
\label{eq:appendix_V_in_Jd}
\end{align}
where $\eta=\frac{\max(\Rc_i, \Rc_j)}{\disqij}$ ($\Rc_i$ is the largest distance from the center in the $i$th body).
If we use $\bm J$ instead of $\bm J_d$, we have
\begin{align}
\begin{split}
      & V\left( \bm q_i, \bm q_j, \bm R_i, \bm R_j \right) \\
    = & -\frac{\Gc m_i m_j}{\disqij} - \frac{\Gc\left( m_i Tr[\bm J_i] + m_j Tr[\bm J_j] \right)}{2 \disqij^3} + \frac{3 \Gc {\left( \qij \right)}^T \left( m_j \bm R_i \bm J_i \bm R_i^T + m_i \bm R_j \bm J_j \bm R_j^T \right)\left( \qij \right)}{2 \disqij^5}  + \Oc(\eta^3) \\
\end{split}
\label{eq:appendix_V_in_J}
\end{align}

Higher order expansions: %\li{Please change it using principal moment of inertia directly, since most of the objects are not uniform in density, or add a discussion on how it would be generalized to objects without uniform density.}:
\begin{align}
    \begin{split} 
        & V\left( \bm q_i, \bm q_j, \bm R_i, \bm R_j \right) \\
        = & \Gc m_i m_j \Bigg\{ -\frac{1}{\disqij}
            +\frac{1}{2\disqij^3} \left[ \frac{1}{5} \left( a_i^2 + b_i^2 + c_i^2 + a_j^2 + b_j^2 + c_j^2 \right) \right] \\
          &  - \frac{3}{2\disqij^5} {\left( \qij \right)}^T  \left(
                \bm R_i  \frac{1}{5} \begin{bmatrix}
                        a_i^2 & 0 & 0 \\
                        0 & b_i^2 & 0 \\
                        0 & 0 & c_i^2 \\
                \end{bmatrix}  \bm R_i^T 
                +
                \bm R_j  \frac{1}{5} \begin{bmatrix}
                        a_j^2 & 0 & 0 \\
                        0 & b_j^2 & 0 \\
                        0 & 0 & c_j^2 \\
                \end{bmatrix}  \bm R_j^T 
            \right) \left( \qij \right) \\
        & -\frac{3}{8\disqij^5} \left( \frac{1}{35} (3 a_i^4 + 3 b_i^4 + 3 c_i^4 + 2 (a_i^2 b_i^2 + a_i^2 c_i^2 + b_i^2 c_i^2))
            + \frac{1}{35} (3 a_j^4 + 3 b_j^4 + 3 c_j^4 + 2 (a_j^2 b_j^2 + a_j^2 c_j^2 + b_j^2 c_j^2)) \right) \\
        & -\frac{3}{4\disqij^5} Tr\left[ \bm R_i^T \bm R_j \frac{1}{5} \begin{bmatrix}
                    a_j^2 & 0 & 0 \\
                    0 & b_j^2 & 0 \\
                    0 & 0 & c_j^2 \\
            \end{bmatrix} \bm R_j^T \bm R_i \frac{1}{5} \begin{bmatrix}
                    a_i^2 & 0 & 0 \\
                    0 & b_i^2 & 0 \\
                    0 & 0 & c_i^2 \\
            \end{bmatrix} \right] \\
            & + \frac{15}{4\disqij^7} \left( {\left( \qij \right)}^T 
                \bm R_i  \frac{1}{35} \begin{bmatrix}
                    a_i^2 \left( 3 a_i^3 + b_i^2 + c_i^2 \right) & 0 & 0 \\
                    0 & b_i^2 \left( a_i^2 + 3 b_i^2 + c_i^2 \right) & 0 \\
                    0 & 0 & c_i^2 \left( a_i^2 + b_i^2  + 3 c_i^2\right) \\
                \end{bmatrix}  \bm R_i^T  \left( \qij \right) \right) \\
            & + \frac{15}{4\disqij^7} \left( {\left( \qij \right)}^T 
                \bm R_j  \frac{1}{35} \begin{bmatrix}
                    a_j^2 \left( 3 a_j^2 + b_j^2 + c_j^2 \right) & 0 & 0 \\
                    0 & b_j^2 \left( a_j^2 + 3 b_j^2 + c_j^2 \right) & 0 \\
                    0 & 0 & c_j^2 \left( a_j^2 + b_j^2  + 3 c_j^2\right) \\
                \end{bmatrix}  \bm R_j^T  \left( \qij \right) \right) \\ 
            & + \frac{15}{4\disqij^7} \left( \frac{1}{5} \left( a_i^2 + b_i^2 + c_i^2 \right)
                {\left( \qij \right)}^T  \bm R_j 
                \frac{1}{5} \begin{bmatrix}
                    a_j^2 & 0 & 0 \\
                    0 & b_j^2 & 0 \\
                    0 & 0 & c_j^2 \\
                \end{bmatrix}  \bm R_j^T  \left( \qij \right) \right) \\
            & + \frac{15}{4\disqij^7} \left( \frac{1}{5} \left( a_j^2 + b_j^2 + c_j^2 \right)
                {\left( \qij \right)}^T  \bm R_i 
                \frac{1}{5} \begin{bmatrix}
                    a_i^2 & 0 & 0 \\
                    0 & b_i^2 & 0 \\
                    0 & 0 & c_i^2 \\
                \end{bmatrix}  \bm R_i^T  \left( \qij \right) \right) \\
            & + \frac{15}{\disqij^7} \left( {\left( \qij \right)}^T
                 \bm R_i 
                \frac{1}{5} \begin{bmatrix}
                    a_i^2 & 0 & 0 \\
                    0 & b_i^2 & 0 \\
                    0 & 0 & c_i^2 \\
                \end{bmatrix}  \bm R_i^T
                 \bm R_j 
                \frac{1}{5} \begin{bmatrix}
                    a_j^2 & 0 & 0 \\
                    0 & b_j^2 & 0 \\
                    0 & 0 & c_j^2 \\
                \end{bmatrix}  \bm R_j^T
             \left( \qij \right) \right) \\
            & - \frac{35}{8\disqij^9} \left(
                Tr\left[ 
                % \bm R_i^T: Tr[R_i^T A R_i] = Tr[A]
                \left( \qij \right) {\left( \qij \right)}^T  \bm R_i 
                    \frac{3}{35} \begin{bmatrix}
                        a_i^4 & 0 & 0 \\
                        0 & b_i^4 & 0 \\
                        0 & 0 & c_i^4 \\
                    \end{bmatrix}
                \bm R_i^T \left( \qij \right)  {\left( \qij \right)}^T
                % \bm R_i
                \right]
                \right) \\
            & - \frac{35}{8\disqij^9} \left(
                Tr\left[
                % \bm R_j^T: Tr[R_j^T A R_j] = Tr[A]
                \left( \qij \right) {\left( \qij \right)}^T  \bm R_j 
                    \frac{3}{35} \begin{bmatrix}
                        a_j^4 & 0 & 0 \\
                        0 & b_j^4 & 0 \\
                        0 & 0 & c_j^4 \\
                    \end{bmatrix}
                \bm R_j^T \left( \qij \right)  {\left( \qij \right)}^T
                % \bm R_j
                \right]
            \right) \\
            & -\frac{105}{4 \norm{\qij}^9} {\left( \qij \right)}^T  \bm R_i 
                    \frac{1}{5} \begin{bmatrix}
                        a_i^2 & 0 & 0 \\
                        0 & b_i^2 & 0 \\
                        0 & 0 & c_i^2 \\
                    \end{bmatrix}  \bm R_i^T  \left( \qij \right)
                 {\left( \qij \right)}^T  \bm R_j 
                    \frac{1}{5} \begin{bmatrix}
                        a_j^2 & 0 & 0 \\
                        0 & b_j^2 & 0 \\
                        0 & 0 & c_j^2 \\
                    \end{bmatrix}  \bm R_j^T  \left( \qij \right) \Bigg\}
              + \mathcal{O} \left( \eta^5 \right). \\
    \end{split}
    \label{eq:potential_4th_order_expansion}
\end{align}

\subsection{Properties of the \textit{hat-map}}\label{appendix:hat-map}

With $\bm{u}, \bm{v}, \bm{w} \in \R^3$, $\mathbf{D} = \begin{bmatrix}
   d_1 & 0 & 0 \\ 
   0 & d_2 & 0 \\ 
   0 & 0 & d_3 \\ 
\end{bmatrix}$, we have
\begin{itemize}
    \item $\hat{\bm{u}} \bm{v} = \bm{u} \times \bm{v}$.
    \item $\widehat{\bm{u} \times \bm{v}} = \hat{\bm{u}} \hat{\bm{v}}-\hat{\bm{v}} \hat{\bm{u}}$.
    \item $ \hat{\bm{u}} \bm{D} - \bm{D} \hat{\bm{u}}^T = Tr \left[ \bm{D} \right] \hat{\bm{u}} - \widehat{\bm{D}\bm{u}}$.
    \item $ \hat{\bm{u}}^T \hat{\bm{u}} \bm{D}  - \bm{D} \hat{\bm{u}}^T \hat{\bm{u}} = \widehat{\bm{u} \times \bm{D} \bm{u}}$
\end{itemize}

\section{Review: equations of motion of one rigid body in a potential} \label{appendix:eom}
We will review two equivalent approaches.
\subsection{Approach 1: Derivation from Constrained Hamiltonian System}
\label{appendix:approach1}
We can view $\bm{R}$ to be in the embedded Euclidean space $\mathbb{R}^{3\times 3} \hookleftarrow \mathsf{SO}(3)$ and use $\bm{R} \in \mathsf{SO}(3)$ as a holonomic constraint. The Lagrangian $L$ (\cref{eq:Lagrangian}) has $9$-DOF %(\chen{right?}\tao{you can say so.})
before applying the constraint $\bm{R} \in \mathsf{SO}(3)$. The conjugate variable of $\bm{R}(t)$ will be denoted by $\bm{P}(t)$.

The constraint of a system forces the evolution of the system in a specific manifold,
and the manifold can be directly calculated from the constraint
(one may refer Chapter VII of \citet{hairer2006geometric} for details).
For a rigid body dynamics represented by a rotation matrix $\bm{R}(t)$,
the constraint is ${\bm{R}(t)}^T \bm{R}(t) - {\bm{I}}_{3 \times 3} = {\bm{0}}_{3 \times 3}$.
\citet{reich1996symplectic}, \citet{hairer2006geometric} have shown the procedure of finding equations of motion
by utilizing the constraint for a rigid body system with a $\bm{R}$ dependent potential.
% Here, we shows an elementary approach of getting equations of motion of the system.
Using Lagrange multipliers \citep{hairer2006geometric} for the constraint $\bm{R}^T \bm{R}-\bm{I}_{3 \times 3}=\bm{0}$, we have the following Lagrangian,
\begin{align}
L\left( \bm{R}, \dot{\bm{R}} \right) 
= \frac{1}{2} Tr\left[ \dot{\bm{R}} \bm{J}_d \dot{\bm{R}}^T  \right] - V(\bm{R}) -
    \frac{1}{2} Tr\left(\bm{\bm{\Lambda}}^T \left( \bm{R}^T \bm{R} - \bm{I}_{3 \times 3} \right)\right),
    \label{eq:constrained_Lagrangian}
\end{align}
% \tao{what is $dot$? did you mean $Tr(A^T B)$?}
with $6$-dim Lagrange multipliers $\bm{\bm{\bm{\Lambda}}} = \begin{bmatrix}
    \lambda_1 & \lambda_4 & \lambda_6 \\
    \lambda_4 & \lambda_2 & \lambda_5 \\
    \lambda_6 & \lambda_5 & \lambda_3 \\
\end{bmatrix}\in \R^{3 \times 3}$ a symmetric matrix.

Doing Legendre transform for \cref{eq:constrained_Lagrangian}, we have
\begin{align}
    \bm{P} = \frac{\partial L\left( \bm R, \dot{\bm{R}} \right)}{\partial \dot{\bm{R}}} = \dot{\bm{R}} \bm{J}_d,
    \label{Legendre_transform}
\end{align}
and the corresponding Hamiltonian,
\begin{align}
    H(\bm{R},\bm{P}) = \frac{1}{2} Tr\left[ \bm{P} \bm{J}_d^{-1} \bm{P}^T  \right] + V(\bm{R}) +
    \frac{1}{2} Tr\left({\bm{\Lambda}}^T \left( \bm{R}^T \bm{R} - \bm{I}_{3 \times 3} \right)\right).
    \label{eq:EoM_Hamiltonian}
\end{align}

As the constraint for $\bm{R}$ is $\bm{R}^T \bm{R} = \bm{I}_{3 \times 3}$, according to~\cite{hairer2006geometric}, the constraint
for $\bm{P}$ can be obtained by taking time derivative for $\bm{R}^T \bm{R} - \bm{I}_{3 \times 3} = \bm{0}_{3 \times 3}$, i.e.
$\bm{J}_d^{-1} \bm{P}^T \bm{R} + \bm{R}^T \bm{P} \bm{J}_d^{-1} = \bm{0}_{3 \times 3}$.

So,
\begin{align}
\left\{ 
    \begin{aligned}
        \dot{\bm{R}} & =  \frac{\partial H}{\partial \bm{P}} = \bm{P} {\bm{J}_d}^{-1}, \\
        \dot{\bm{P}} & = -\frac{\partial H}{\partial \bm{R}} = -\frac{\partial V\left( \bm{R} \right)}{\partial \bm{R}} - \bm{R} \bm{\Lambda}, \\
    \end{aligned}
\right.
% \quad \text{such that\ } \bm{R} \in \mathsf{SO}(3),
\label{eq:RPdot}
\end{align}
on the manifold
\begin{align}
    \mathcal{M} = \left\{ \left( \bm{R}, \bm{P} \right) | \bm{R}^T \bm{R} = \bm{I}_{3 \times 3},
    \bm{J}_d^{-1} \bm{P}^T \bm{R} + \bm{R}^T \bm{P} \bm{J}_d^{-1} = \bm{0}_{3 \times 3} \right\}.
\end{align}

% we can rewrite the second constraint for the manifold in the following,
% \begin{align}
%     \begin{split}
%         \bm{X}^T + \bm{X} = \bm{0}_{3 \times 3}
%         \Leftrightarrow
%         \bm{X} \in \mathfrak{so}(3)
%         \Leftrightarrow
%         \dot{\bm{X}}, \bm{X}(0) \in \mathfrak{so}(3). \\
%     \end{split}
%     \label{eq:2nd_restriction_equivalence}
% \end{align}
% To find $\bm{\Lambda}$, using the last set of equivalence of \cref{eq:2nd_restriction_equivalence}.
% As $X(0) \in \mathfrak{so}(3)$ is naturally satisfied.
%\revision{$\Lambda$ was incorrectly expressed here, the procedures of getting EoMs are updated}
Note that $\hat{\bm{\Omega}} = \bm{R}^T \bm{P} \bm{J}_d^{-1}$ with $\bm \Omega$ being the body's angular velocity. 
Taking time derivative for $\hat{\bm \Omega}$, we have
\begin{align}
    \begin{split}
        \hat{\dot{\bm \Omega}} = \bm{J}_d^{-1} \bm{P}^T \bm{P} \bm{J}_d^{-1} 
                    + \bm{R}^T \left( -\frac{\partial V\left( \bm{R} \right)}{\partial \bm{R}} - \bm{R} \bm{\Lambda} \right) \bm{J}_d^{-1}.
        \label{eq:Omega_hat_dot_2}
    \end{split}
\end{align}
% To satisfy the restriction $\dot{\bm X} \in \mathfrak{so}(3)$ and make sure $\bm \Lambda$ is symmetric, we have
% \begin{align}
%     \bm \Lambda=\frac{\bm J_d^{-1} \bm P^T \bm P + \bm P^T \bm P \bm J_d^{-1}}{2}
%     + \frac{\bm J_d^{-1} {\left( \frac{\partial \bm V}{\partial \bm R} \right)}^T \bm R
%     + \bm R^T \left( \frac{\partial \bm V}{\partial \bm R} \right) \bm J_d^{-1}}{2}.
% \end{align}

Physically, we want to find dynamics of $\bm{R}$ and the body's angular momentum $\bm{\Pi}$.
Since $\hat{\bm \Pi} = \widehat{\bm J \bm \Omega}
= Tr \left[ \bm J_d \right] \hat{\bm \Omega} - \widehat{\bm J_d \bm \Omega}
= \hat{\bm \Omega} \bm J_d - \bm J_d \hat{\bm \Omega}^T$ (see appendix \ref{appendix:hat-map}),
we may find dynamics of $\bm \Pi$,
% \begin{align}
%     \begin{split}
%         \hat{\dot{\bm{\Omega}}} = & \frac{\bm{J}_d^{-1} \bm{P}^T \bm{P} \bm{J}_d^{-1} - \bm{P}^T \bm{P} \bm{J}_d^{-1} \bm{J}_d^{-1}}{2} \\
%         & + \frac{\bm{J}_d^{-1} {\left( \frac{\partial V\left( \bm{R} \right)}{\partial \bm{R}} \right)}^T \bm R
%                     - \bm R^T \frac{\partial V\left( \bm{R} \right)}{\partial \bm{R}} \bm J_d^{-1}}{2}
%         \label{eq:Omega_hat_dot}
%     \end{split}
% \end{align}
% we may find dynamics of $\bm \Pi$,
\begin{align}
    \begin{split} 
        \widehat{\dot{\bm\Pi}} = & \left( \bm{J}_d^{-1} \bm{P}^T \bm{P} - \bm{P}^T \bm{P} \bm{J}_d^{-1} \right)
           + \left({\left( \frac{\partial V\left( \bm{R} \right)}{\partial \bm{R}} \right)}^T \bm R
            - \bm R^T \frac{\partial V\left( \bm{R} \right)}{\partial \bm{R}} \right), \\
    \end{split}
    \label{eq:Pihat_dot_1}
\end{align}
with the symmetric $\bm \Lambda$ vanished \footnote{Since $\Lambda$ is symmetric, applying $\hat{\dot{\bm \Omega}} \in \mathfrak{so}(3)$, $\Lambda$ can actually be solved from \cref{eq:Omega_hat_dot_2}.}.

As $\bm P = \bm R \hat{\bm \Omega} \bm J_d$,  properties of hat-map (see appendix \ref{appendix:hat-map}) lead to
\begin{align}
    \begin{split}
        \widehat{\dot{\bm\Pi}}
        = & \left( \hat{\bm\Omega}^T \hat{\bm\Omega} \bm{J}_d - \bm{J}_d  \hat{\bm\Omega}^T \hat{\bm\Omega}\right)
            + \left({\left( \frac{\partial V\left( \bm{R} \right)}{\partial \bm{R}} \right)}^T \bm R
            - \bm R^T \frac{\partial V\left( \bm{R} \right)}{\partial \bm{R}} \right) \\
            & = \widehat{\bm\Omega \times \bm{J}_d \bm{\Omega}}
            + \left({\left( \frac{\partial V\left( \bm{R} \right)}{\partial \bm{R}} \right)}^T \bm R
            - \bm R^T \frac{\partial V\left( \bm{R} \right)}{\partial \bm{R}} \right). \\
        \label{eq:Pihat_dot_2}
    \end{split}
\end{align}
Thus
\begin{align}
    \begin{split}
        \dot{\bm{\Pi}} & = \bm{\Omega} \times \bm{J}_d \bm{\Omega}
        - {\left( \bm{R}^T \frac{\partial V\left( \bm{R} \right)}{\partial \bm{R}} 
        - {\left(\frac{\partial V\left( \bm{R} \right)}{\partial \bm{R}}\right)}^T \bm{R} \right)}^{\vee}\\
        & = \bm{\Omega} \times \left( Tr[\bm{J}_d] - \bm{J} \right) \bm{\Omega}
        - {\left( \bm{R}^T \frac{\partial V\left( \bm{R} \right)}{\partial \bm{R}} 
        - {\left(\frac{\partial V\left( \bm{R} \right)}{\partial \bm{R}}\right)}^T \bm{R} \right)}^{\vee}\\
        & = -\bm{\Omega} \times \bm{J} \bm{\Omega}
        - {\left( \bm{R}^T \frac{\partial V\left( \bm{R} \right)}{\partial \bm{R}} 
        - {\left(\frac{\partial V\left( \bm{R} \right)}{\partial \bm{R}}\right)}^T \bm{R} \right)}^{\vee}\\
        & = \bm{\Pi} \times \bm{J}^{-1} \bm{\Pi} 
        - {\left( \bm{R}^T \frac{\partial V\left( \bm{R} \right)}{\partial \bm{R}} 
        - {\left(\frac{\partial V\left( \bm{R} \right)}{\partial \bm{R}}\right)}^T \bm{R} \right)}^{\vee}\\
    \end{split}
    \label{eq:Pi_dot}
\end{align}

So, equations of motion with respect to $\bm R$ and $\bm \Pi$ for one rigid body system are
\begin{align}
\left\{ 
    \begin{aligned}
        \dot{\bm{R}} & = \bm R \widehat{\bm{J}^{-1} \bm{\Pi}}, \\
        \dot{\bm{\Pi}} & = \bm{\Pi} \times \bm{J}^{-1} \bm{\Pi} 
        - {\left( \bm{R}^T \frac{\partial V\left( \bm{R} \right)}{\partial \bm{R}} 
        - {\left(\frac{\partial V\left( \bm{R} \right)}{\partial \bm{R}}\right)}^T \bm{R} \right)}^{\vee}. \\
    \end{aligned}
\right.
\label{eq:appendix_eom_one_rb}
\end{align}

\subsection{Approach 2: Variational Principle for Mechanics on Lie Group}
\label{appendix:approach2}
How to obtain Euler-Lagrange equation for the Hamilton's variational principle on a Lie group has been well studied (e.g., \cite{marsden1994introduction, HoScSt2009}). Here we summarize the results for the special case of rigid bodies from the expository part of \citet{lee2005lie}.

Denote the infinitesimally varied rotation by
$\bm R_\epsilon = \bm R \exp(\epsilon \bm \hat{\bm \eta})$ with $\epsilon \in \R$ and $\bm \eta \in \R^3$, where $\exp(\cdot)$ is a mapping from $\mathfrak{so}(3)$ to $\mathsf{SO}(3)$.
The varied angular velocity is
\begin{align}
    \begin{split}
        \hat{\bm \Omega}_\epsilon & =\bm R_\epsilon^T \dot{\bm R}_\epsilon
        =e^{-\epsilon \hat{\bm \eta}} \bm R^T \left( \dot{\bm R} e^{\epsilon \hat{\bm \eta}}
        + \bm R \cdot e^{\epsilon \hat{\bm \eta}} \epsilon \hat{\dot{\bm \eta}} \right)  \\
        & =e^{-\epsilon \hat{\bm \eta}} \hat{\bm \Omega} e^{\epsilon \hat{\bm \eta}} + \epsilon \hat{\dot{\bm \eta}}
        = \hat{\bm \Omega} + \epsilon \left\{  \hat{\dot{\bm \eta}} + \hat{\bm \Omega} \hat{\bm \eta} - \hat{\bm \eta} \hat{\bm \Omega} \right\}
        + \mathcal{O}(\epsilon^2). \\
    \end{split}
\end{align}

Consider the action
\begin{align}
    S\left( \bm \Omega, \bm R \right) = \int_{t_0}^{t_1} L\left( \bm \Omega, \bm R \right) \, dt
    = \int_{t_0}^{t_1} \frac{1}{2} Tr \left[ \hat{\bm \Omega} \bm J_d \hat{\bm \Omega}^T \right] - V(\bm R) \, dt.
\end{align}
Taking the variation of the action $S$, we have
\begin{align}
    \begin{split}
        S_\epsilon \left( \bm \Omega, \bm R \right)
         = & S\left( \bm \Omega_\epsilon, \bm R_\epsilon \right) \\
         = & S\left( \bm \Omega, \bm R \right) + \epsilon \int_{t_0}^{t_1}
            \frac{1}{2} Tr \Big[
            - \hat{\dot{\bm \eta}} \left( \bm J_d \hat{\bm \Omega} + \hat{\bm \Omega} \bm J_d \right) \\
            & + \hat{\bm \eta} \hat{\bm \Omega} \left( \bm J_d \hat{\bm \Omega} + \hat{\bm \Omega} \bm J_d  \right)
            - \hat{\bm \eta} \left( \bm J_d \hat{\bm \Omega} + \hat{\bm \Omega} \bm J_d \hat{\bm \Omega} \right) 
        \Big] \\
        & + Tr \left[ \hat{\bm \eta} \bm R^T \frac{\partial \bm R}{\partial \bm R} \right] \, dt + \mathcal{O}(\epsilon^2).
    \end{split}
    \label{eq:varied_action}
\end{align}
Using \textit{Hamilton's Principle}, we have $\frac{d}{d \epsilon} \Big \lvert_{\epsilon=0} S_\epsilon = 0$, i.e.
\begin{align}
    \begin{split}
        \frac{1}{2} \int_{t_0}^{t_1} Tr \left[ \hat{\bm \eta} \left\{ \widehat{\bm J \dot{\bm \Omega}}
        + \widehat{\bm \Omega \times \bm J \Omega} + 2 \bm R^T \frac{\partial V}{\partial \bm R} \right\} \right] = 0
    \end{split}
    \label{eq:Hamilton_Principle_eom}
\end{align}
for any $\bm \eta \in \R^3$. Therefore, $\left\{ \widehat{\bm J \dot{\bm \Omega}}
+ \widehat{\bm \Omega \times \bm J \Omega} + 2 \bm R^T \frac{\partial V}{\partial \bm R} \right\}$ must be skew-symmetric,
which gives us
\begin{align}
    \widehat{\bm J \dot{\bm \Omega}} = -\widehat{\bm \Omega \times \bm J \bm \Omega} +
    \left( {\frac{\partial V}{\partial \bm R}}^T \bm R 
    - \bm R^T \frac{\partial V}{\partial \bm R} \right).
\end{align}
Thus
\begin{align}
    \widehat{\dot{\bm \Pi}} = \widehat{\bm \Pi \times \bm J^{-1} \bm \Pi} +
    \left( {\frac{\partial V}{\partial \bm R}}^T \bm R 
    - \bm R^T \frac{\partial V}{\partial \bm R} \right).
\end{align}

\section{Proof of the Hierarchical Composition Error}\label{appendix:proof_error}

\begin{theorem}
Given four Hamiltonian flows ${\{\varphi_t^{[i]}\}}_{i=1}^4$ of $H_i$ with $H=H_1+H_2+H_3+H_4$. Construct an integrator $\varphi_h:=\Cc_3(\Cc_1(\varphi_h^{[1]}, \varphi_h^{[2]}), \Cc_2(\varphi_h^{[3]}, \varphi_h^{[4]}))$ via composition methods $\Cc_i,\, i=1,2,3$ such that
\begin{align*}
\Cc_i(\varphi_h^A, \varphi_h^B) =
\varphi_{a_1^{[i]} h}^A \circ \varphi_{b_1^{[i]} h}^B \circ
\varphi_{a_2^{[i]} h}^A \circ \varphi_{b_2^{[i]} h}^B \circ
\cdots
\varphi_{a_{n_i}^{[i]} h}^A \circ \varphi_{b_{n_i}^{[i]} h}^B.
\end{align*}
Then $\Ec(\varphi_h)$ equals to the summation of orders of $\Cc_i,\,i=1,2,3$ with $\Ec(\cdot)$ being the global error function.
\end{theorem}

\begin{proof}
Assume the associated Lie operators of $\varphi_t^{[i]}$'s vector fields are $\Lc_{H_i}$.
There exists a Lie operator $A_1$ such that for $\Cc_1(\varphi_h^{[1]}, \varphi_h^{[2]})$,
\begin{align*}
    e^{a_1^{[1]} \Lc_{H_1}} e^{b_1^{[1]} \Lc_{H_2}}
    e^{a_2^{[1]} \Lc_{H_1}} e^{b_2^{[1]} \Lc_{H_2}}
    \cdots
    e^{a_{n_1}^{[1]} \Lc_{H_1}} e^{b_{n_1}^{[1]} \Lc_{H_2}}
    = e^{A_1}
    = e^{\Lc_{H_1}+\Lc_{H_2} + E_1}
\end{align*}
with the order of $E_1$ equals the order of $\Cc_1$.
Similarly, we have
\begin{align*}
    e^{a_1^{[2]} \Lc_{H_3}} e^{b_1^{[2]} \Lc_{H_4}}
    e^{a_2^{[2]} \Lc_{H_3}} e^{b_2^{[2]} \Lc_{H_4}}
    \cdots
    e^{a_{n_2}^{[2]} \Lc_{H_3}} e^{b_{n_2}^{[2]} \Lc_{H_4}}
    = e^{A_2}
    = e^{\Lc_{H_3}+\Lc_{H_4} + E_2}
\end{align*}
with the order of $E_2$ equals the order of $\Cc_2$. 
Further for $\Cc_3$,
\begin{align*}
    & e^{a_1^{[3]} A_1} e^{b_1^{[3]} A_2}
    e^{a_2^{[3]} A_1} e^{b_2^{[3]} A_2}
    \cdots
    e^{a_{n_3}^{[3]} A_1} e^{b_{n_3}^{[3]} A_2} \\
    = & e^{A_3}
    = e^{A_1+A_2+E_3} = e^{\Lc_{H_1}+\Lc_{H_2}+\Lc_{H_3}+\Lc_{H_4}+E_1+E_2+E_3},
\end{align*}
with the order of $E_3$ equals the order of $\Cc_3$. 
Therefore, the global error of $\varphi_h$ is the summation of the orders of $\Cc_i,\,i=1,2,3$.
\end{proof}

\section{Composition Methods}\label{appendix:composition}
Symplectic integrators of a Hamiltonian system $H=A+B$ can be constructed by composing the flows of $A$ and $B$. We list the composition methods used in the paper below for the general $H=A+B$ and perturbative Hamiltonian $H=A+\varepsilon B$ in \cref{tb:composition1} and \cref{tb:composition2} respectively.
% \tao{do they correspond to notations like $T_4$, $M_{42}$ etc.? if yes, please indicate the correspondence (in the table?)} \chen{No. They are composition methods used for our schemes like $T_4$, $M_{42}$}.

% For general $H=A+B$, assume the Hamiltonian flows of $A$ and $B$ are $\varphi_h^A$ and $\varphi_h^B$ respectively.

\begin{table}
\begin{center}
\renewcommand{\arraystretch}{1.5}
\setlength{\tabcolsep}{10pt}
\begin{tabular}{| p{15cm} | c |}
\hline
 composition method & order \\ 
 \hline
 \hline
 $\Cc_\Euler \left(\varphi_h^A, \varphi_h^B \right):=\varphi_h^A \circ \varphi_h^B$ & (1) \\
 \hline
 $\Cc_\Verlet \left(\varphi_h^A, \varphi_h^B \right):=\varphi_{h/2}^A \circ \varphi_h^B \circ \varphi_{h/2}^A$ & (2) \\ 
 \hline
 $\Cc_\TriJump \left(\varphi_h^A, \varphi_h^B \right):=
 \varphi_{\gamma_1 h} \circ \varphi_{\gamma_2 h}
 \circ \varphi_{\gamma_1 h} \newline$
with $\varphi_h:=\Cc_\Verlet(\varphi_h^A, \varphi_h^B)$ and $\gamma_1 = \nicefrac{1}{(2-2^{1/3})}$, $\gamma_2=1-2\gamma_1$\citep{suzuki1990fractal}.
 & (4) \\
 \hline
 $\Cc_{S6} \left(\varphi_h^A, \varphi_h^B \right):=
 \varphi_{a_1 h} \circ \varphi_{a_2 h} \circ
 \varphi_{a_3 h} \circ \varphi_{a_4 h} \circ
 \varphi_{a_3 h} \circ \varphi_{a_2 h} \circ
 \varphi_{a_1 h}$
 \newline with $\varphi_h:=\Cc_\Verlet(\varphi_h^A, \varphi_h^B)$
 and $a_1=0.784513610477560$, $a_2=0.235573213359357$, $a_3=-1.17767998417887$, $a_4=1-2(a_1+a_2+a_3)$\citep{yoshida1990construction} & (6) \\
%  \hline
%  $\Cc_{S6} \left(\varphi_h^A, \varphi_h^B \right):=
%  \varphi_{a_1 h} \circ \varphi_{a_2 h} \circ
%  \varphi_{a_3 h} \circ \varphi_{a_4 h} \circ
%  \varphi_{a_5 h} \circ \varphi_{a_4 h} \circ
%  \varphi_{a_3 h} \circ \varphi_{a_2 h} \circ
%  \varphi_{a_1 h}$
%  \newline with $\varphi_h:=\Cc_\Verlet(\varphi_h^A, \varphi_h^B)$ and $a_1=0.1867$, $a_2=0.55549702371247839916$, $a_3=0.12946694891347535806$,
%  $a_4=-0.84326562338773460855$, $a_5=1-2(a_1+a_2+a_3+a_4)$\citep{mclachlan2002splitting} & (6) \\
\hline
\end{tabular}
\end{center}
\caption{Composition methods $\Cc(\cdot, \cdot)$ of general $H=A+B$. $\varphi_h^A$ and $\varphi_h^B$ are flows of $A$ and $B$ respectively.}
\label{tb:composition1}
\end{table}

\begin{table}
\begin{center}
\renewcommand{\arraystretch}{1.5}
\setlength{\tabcolsep}{10pt}
\begin{tabular}{| p{15cm} | c |}
\hline
 composition method & order \\ 
 \hline
 \hline
    $\Cc_{BAB22} (\varphi_h^A, \varphi_h^B) =
    \varphi_{h/2}^B
    \circ \varphi_{h}^A
    \circ \varphi_{h/2}^B$ & $(\cdot,2,2)$ \\
 \hline
    $\Cc_{ABA22} (\varphi_h^A, \varphi_h^B) =
    \varphi_{h/2}^A
    \circ \varphi_{h}^B
    \circ \varphi_{h/2}^A$ & $(\cdot,2,2)$ \\
%  \hline
%     $\Cc_{ABA22*} (\varphi_h^A, \varphi_h^B) =
%     \varphi_{h/2}^A
%     \circ \varphi_{h}^B
%     \circ {\varphi_{h/2}^A}^*$
%     with ${\varphi_{h/2}^A}^*$ the adjoint of $\varphi_{h/2}^A$
%     and $\varphi_h^B$ a time reversible flow
%     & $(0,2,2)$ \\
 \hline
    $\Cc_{ABA42} (\varphi_h^A, \varphi_h^B) =
    \varphi_{\frac{(3-\sqrt{3})h}{6}}^A
    \circ \varphi_{\frac{h}{2}}^B
    \circ \varphi_{\frac{h}{\sqrt{3}}}^A
    \circ \varphi_{\frac{h}{2}}^B
    \circ \varphi_{\frac{(3-\sqrt{3})h}{6}}^A.$
    (${\Sc \Ac \Bc \Ac}_2$ in \citep{laskar2001high} or equivalently the order $(4,2)$ $ABA$ method with $s=2$ in \citep{mclachlan1995composition})
    & $(\cdot,4,2)$ \\
 \hline
\end{tabular}
\end{center}
\caption{Composition methods $\Cc(\cdot, \cdot)$ of perturbative $H=A+\varepsilon B$. $\varphi_h^A$ and $\varphi_h^B$ are flows of $A$ and $\varepsilon B$ respectively.}
\label{tb:composition2}
\end{table}

\end{document}